\documentclass[a4paper,USenglish,cleveref, autoref, thm-restate,numberwithinsect,pdfa]{lipics-v2021}

\usepackage{mathtools}
\usepackage{stmaryrd}
\usepackage{tikz}
\usetikzlibrary{arrows,automata}

\nolinenumbers 

\hideLIPIcs  
\newcommand{\stateSpace}{S}
\newcommand{\prob}{P}
\newcommand{\probMeasure}{\mathrm{Pr}}
\newcommand{\dtmc}{\mathcal{M}}

\newcommand{\labelFunction}{\ell}
\newcommand{\initialState}{s_{init}}

\newcommand{\Paths}{\mathrm{Paths}}
\newcommand{\trace}{\mathrm{tr}}

\newcommand{\Words}{\mathrm{Words}}

\newcommand{\automaton}{\mathcal{A}}
\newcommand{\automatonB}{\mathcal{B}}
\newcommand{\automatonD}{\mathcal{D}}
\newcommand{\astates}{Q}
\newcommand{\ainit}{q_{init}}
\newcommand{\afinal}{F}
\newcommand{\accLanguage}{\mathcal{L}}
\newcommand{\UP}{\mathrm{UP}}
\newcommand{\resLang}{Res^E}

\newcommand{\fdfa}{\mathcal{F}}
\newcommand{\lfdfa}{\mathcal{Q}}
\newcommand{\pfdfa}{\mathcal{P}}

\newcommand{\fufa}{\mathcal{U}}
\newcommand{\lfufa}{\mathcal{Q}}
\newcommand{\pfufa}{\mathcal{P}}
\newcommand{\fufaH}{\mathcal{H}}

\newcommand{\mstates}{A}

\title{\texorpdfstring{Probabilistic Model Checking via Families of Deterministic and Unambiguous Finite Automata\thanks{Authors are ordered alphabetically.}}{Probabilistic Model Checking via Families of Deterministic and Unambiguous Finite Automata}}
\titlerunning{PMC via Families of Deterministic and Unambiguous Finite Automata}

\author{Christel Baier}{Technische Universität Dresden, Dresden, Germany}{christel.baier@tu-dresden.de}{https://orcid.org/0000-0002-5321-9343}{}
\author{Sascha Kl\"uppelholz}{Technische Universität Dresden, Dresden, Germany}{sascha.klueppelholz@tu-dresden.de}{https://orcid.org/0000-0003-1724-2586}{}

\author{Timm Spork}{Technische Universität Dresden, Dresden, Germany}{timm.spork@tu-dresden.de}{https://orcid.org/0009-0008-4461-0667}{}

\authorrunning{C. Baier, S. Klüppelholz, and T. Spork}

\Copyright{Christel Baier, Sascha Klüppelholz, and Timm Spork} 

\relatedversion{Full version of a paper accepted for publication at CONCUR 2026.}

\funding{This work was partly funded by the DFG grant 389792660 as part of TRR~248 (Foundations of Perspicuous Software Systems, see \url{https://perspicuous-computing.science}), the Cluster of Excellence EXC 2050/2 (CeTI, project ID 390696704, as part of Germany's Excellence Strategy), and by the
BMFTR (Federal Ministry of Research, Technology and Space) in DAAD project 57616814 (SECAI, School of Embedded and Composite AI) as
part of the program Konrad Zuse Schools of Excellence in Artificial Intelligence.
}

\acknowledgements{We thank Dana Fisman for telling us about the double-exponential lower bound for the translation from LTL to FDFA, Jianlin Li for providing us with a copy of \cite{BALi}, and Robin Ziemek for many interesting discussions on the topic.}

\EventEditors{Ana Sokolova and Patrick Totzke}
\EventNoEds{2}
\EventLongTitle{37th International Conference on Concurrency Theory (CONCUR 2026)}
\EventShortTitle{CONCUR 2026}
\EventAcronym{CONCUR}
\EventYear{2026}
\EventDate{September 1--4, 2026}
\EventLocation{Liverpool, UK}
\EventLogo{}
\SeriesVolume{391}
\ArticleNo{23}

\begin{document}
	
\maketitle

\begin{abstract}
Families of deterministic finite automata (FDFA) have been introduced as a concise automaton model that characterizes $\omega$-regular languages by processing their ultimately periodic words. FDFA are known to enjoy many good properties and can be exponentially more succinct than deterministic $\omega$-automata with Rabin, Streett or parity acceptance.
This paper addresses two main questions: (1) Are FDFA suitable for probabilistic model checking purposes? and (2) Is it possible to obtain an even more compact representation of $\omega$-regular languages by allowing the components of an FDFA to be unambiguous instead of deterministic?
Question (1) is answered in the affirmative by presenting the first polynomial-time algorithm for computing the probability that a discrete-time Markov chain satisfies an $\omega$-regular property represented as an FDFA.
Question (2) is motivated by the fact that unambiguous finite automata may require exponentially fewer states than deterministic ones. This paper introduces a model of families of unambiguous finite automata (FUFA) that captures the class of $\omega$-regular languages. FUFA can be exponentially more succinct than both FDFA and unambiguous Büchi automata, and there is a single-exponential translation from linear temporal logic (LTL) to FUFA. This stands in contrast to a double-exponential lower bound for the translation from LTL to FDFA. Moreover, the polynomial-time probabilistic model checking algorithm for discrete-time Markov chains against FDFA-specifications is extended to the case where the property is represented by an FUFA with a deterministic leading automaton.
\end{abstract} 
\section{Introduction}
Regular languages over infinite words, also called \emph{$\omega$-regular languages}, are widely used in fields such as verification and synthesis \cite{SRM,ATAAPV,PoMC,CSRCP}. There are many ways to represent $\omega$-regular languages like, e.g., $\omega$-automata, $\omega$-regular expressions, or logics \cite{ALIG,OA}. 

\emph{Families of deterministic finite automata} (\emph{FDFA}) \cite{LROL,FDFAAORL} gained popularity in recent years as a concise description of $\omega$-regular languages that is well-suited for automata learning \cite{LROL,NLABAFDFACT,NFFARLORL,SPFA}. The central idea behind FDFA is to exploit the fact that $\omega$-regular languages are uniquely characterized by their \emph{ultimately periodic words}, i.e., the words in the language that are of the form $w = uv^\omega$ for finite words $u, v$, where $v$ is nonempty \cite{DMRSOA,UPWROL}. 

Formally, an FDFA $\fdfa = (\lfdfa, \{\pfdfa^q\}_{q \in \lfdfa})$ consists of a \emph{leading automaton} $\lfdfa$ and a set of \emph{progress automata} $\pfdfa^q$, one for each state $q$ of $\lfdfa$. Here, $\lfdfa$ and all of the $\pfdfa^q$ are \emph{deterministic finite automata} (DFA). The inputs of $\fdfa$ are \emph{decompositions} $(u,v)$ that represent ultimately periodic words $xy^\omega$, i.e., that satisfy $uv^\omega = xy^\omega$. Intuitively, $\fdfa$ accepts $(u,v)$ if (1) $v$ is accepted by the progress automaton $\pfdfa^q$, for $q$ the state reached in $\lfdfa$ after reading $u$, and (2) $v$ closes a loop on $q$ in $\lfdfa$. It was shown in \cite{LROL,FDFAAORL} that \emph{saturated} FDFA, which are FDFA in which either \emph{all} decompositions of an ultimately periodic word that satisfy (1) and (2) are accepted or \emph{none} of them is, characterize the set of $\omega$-regular languages and enjoy many nice properties like constant space complementation and efficiently implementable Boolean operations and decision problems. Moreover, saturated FDFA can be exponentially more succinct than deterministic Rabin, Streett or parity automata \cite{FDFAAORL}, making them a promising automaton model for applications in the field of (probabilistic) model checking.

To the best of our knowledge, the only available works that discuss the use of FDFA in the context of model checking are \cite{BALi,CRBA}. While the former gives a translation from the linear temporal logic (LTL) \cite{TLP} to FDFA and shows how to use FDFA for the verification of transition systems against LTL-formulas, the latter proposes to model both the system under consideration and the property of interest as FDFA, and check if the system satisfies the property by deciding language containment between the two FDFA, which can be done efficiently \cite{FDFAAORL}. Both \cite{BALi,CRBA} do, however, not consider \emph{probabilistic} model checking. 

Given an $\omega$-regular specification $E$ and a discrete-time Markov chain (DTMC) $\dtmc$, the \emph{probabilistic model checking problem} asks to compute the probability that $\dtmc$ generates a path that satisfies $E$ \cite{PoMC}. Classical solutions to this problem translate the $\omega$-regular specification into a deterministic $\omega$-automaton $\automaton$ and compute the probability of reaching specific bottom strongly connected components (BSCCs) in a product of $\dtmc$ and $\automaton$ \cite{ATAAPV,CPV,PoMC}. Because of the succinctness of saturated FDFA compared to types of deterministic $\omega$-automata typically used in probabilistic verification \cite{FDFAAORL}, it is of interest to investigate if saturated FDFA are a suitable model for the representation of $\omega$-regular properties in probabilistic model checking. 

We answer this question in the affirmative by providing the \emph{first} probabilistic model checking algorithm for DTMCs $\dtmc$ against $\omega$-regular specifications represented as saturated FDFA $\fdfa = (\lfdfa, \{\pfdfa^q\}_{q \in \lfdfa})$. The algorithm is based on a classification and subsequent reachability analysis of specific \emph{good} BSCCs in a product of $\dtmc$ and $\lfdfa$, and its runtime is polynomial in the size of $\dtmc$ and the sizes of the components of $\fdfa$. 

By definition, all components of an FDFA are deterministic. It is, however, well-known that nondeterministic and unambiguous finite automata can require exponentially fewer states than minimal DFA for the same language \cite{SDCFRFL,ECPURERGFA,DCNFADA}. Here, an automaton is unambiguous if every word has at most one accepting run \cite{UAT}. Replacing the deterministic finite automata in an FDFA with nondeterministic and/or unambiguous finite automata can therefore yield up to exponential savings in the individual components of the resulting family of automata. 

Following this observation, we introduce the novel notion of \emph{families of unambiguous finite automata} (\emph{FUFA}), which have the same structure as FDFA, but allow a nondeterministic leading automaton and unambiguous progress automata. The notion from the literature closest to FUFA are families of nondeterministic finite automata (FNFA) \cite{SPFA}, which require a deterministic leading automaton but allow  \emph{nondeterministic} progress automata. 

Similar to the case of FDFA, \emph{saturated} FUFA characterize the set of $\omega$-regular languages. Furthermore, FUFA can be exponentially more succinct than FDFA and unambiguous Büchi automata (UBA), and there is a single-exponential translation from LTL to FUFA. This stands in contrast to the case of saturated FDFA, where a double-exponential blow-up is inevitable \cite{DFpersCom,CLTLFE}. Additionally, it is shown how to combine the probabilistic model checking algorithm for DTMCs against  FDFA-specifications developed in this paper with methods from \cite{MCUA} for probabilistic model checking of DTMCs against UBA-specifications to obtain a polynomial-time algorithm that computes the probability of a DTMC to generate an infinite path that satisfies an $\omega$-regular specification represented as a \emph{d}FUFA, which is an FUFA with a \emph{deterministic} leading automaton. While the complexity of probabilistic model checking of DTMCs against general FUFA-specifications remains open, the corresponding polynomial-time algorithm for dFUFA is a first step in the direction of applying FUFA in probabilistic verification. 

\subparagraph*{Main Contributions.} 
In summary, the main contributions are 
\begin{itemize}
	\item a poly-time model checking algorithm for DTMC against FDFA-specifications (\Cref{sec:fdfa-mc}). 
	\item the introduction of FUFA (families of unambiguous finite automata) as a new model to represent $\omega$-regular languages that can be exponentially more succinct than FDFA and UBA, and for which there is a single-exponential translation from LTL (\Cref{sec:fufa}).
	\item a poly-time model checking algorithm for DTMC against dFUFA-specifications (\Cref{sec:dfufa-mc}).
\end{itemize}\vspace{0.1cm}
\noindent
Full proofs can be found in the Appendix. 
\section{Preliminaries}\label{sec:preliminaries}
\subparagraph*{Words and languages.} A finite set $\Sigma$ is called an \emph{alphabet}. Let $\varepsilon$ be the \emph{empty word}, $\Sigma^*$ the set of \emph{finite words} over $\Sigma$, $\Sigma^+ = \Sigma^* \setminus \{\varepsilon\}$, and $\Sigma^{*+} = \Sigma^* \times \Sigma^+$. The \emph{concatenation} of $w_1, w_2 \in \Sigma^*$ is $w_1 \cdot w_2$, where $w \cdot \varepsilon = w$, and we often write $w_1w_2$ instead of $w_1 \cdot w_2$. For $k \in \mathbb{N}$, $w^k$ is the $k$-fold repetition of $w$, where $w^0 = \varepsilon$. The \emph{length} of $w = w_1 \ldots w_n \in \Sigma^*$ is $\vert w \vert = 0$ if $w = \varepsilon$ and $\vert w \vert = n$ otherwise. 
$\Sigma^\omega$ is the set of \emph{infinite} or \emph{$\omega$-words} over $\Sigma$. A \emph{language} is a subset of $\Sigma^*$, and an \emph{$\omega$-language} is a subset of $\Sigma^\omega$. 

\subparagraph*{Automata over finite words.} A \emph{non-deterministic finite automaton} (\emph{NFA}) is a $5$-tuple $\automaton = (\astates, \Sigma, \ainit, \delta, \afinal)$ with $\astates$ a finite set of \emph{states}, $\Sigma$ an \emph{alphabet}, $\ainit \in \astates$ a unique \emph{initial state}, $\delta \colon \astates \times \Sigma \to 2^\astates$ a \emph{transition function} and $\afinal \subseteq \astates$ a set of \emph{final states}. To emphasize the automaton $\automaton$ under consideration we may write, e.g., $\afinal^\automaton$ instead of $\afinal$. By slight abuse of notation we extend $\delta$ to finite words via $\delta \colon \astates \times \Sigma^* \to 2^{\astates}$ with $\delta(q, \varepsilon) = q$ and $\delta(q, aw) = \bigcup_{q' \in \delta(q, a)} \delta(q', w)$ for $q \in \astates, a \in \Sigma$ and $w \in \Sigma^*$. A \emph{run} of $\automaton$ on a word $w = a_1 \ldots a_n \in \Sigma^*$ is a sequence of states $q_0q_1 \ldots q_{n}$ such that $q_0 = \ainit$ and $q_{i+1} \in \delta(q_i, a_{i+1})$ for all $0 \leq i < n$. $\automaton$ \emph{accepts} the word $w \in \Sigma^*$ if there is a run $r$ of $\automaton$ on $w$ that ends in a final state, and in this case $r$ is \emph{accepting}. $\accLanguage(\automaton) = \{w \in \Sigma^* \mid \automaton \text{ accepts} w\}$ is the \emph{language of $\automaton$}. 

Given $w \in \Sigma^*$, let $\automaton(w) = \delta(\ainit, w)$. The \emph{size of $\automaton$} is $\vert \automaton \vert = \vert \astates \vert$. We often, by slight abuse of notation, identify $\astates$ with $\automaton$ and write $q \in \automaton$ instead of $q \in \astates$. For $q \in \astates$ let $\automaton_q$ be like $\automaton$, but with initial state $q$. $\automaton$ is \emph{total} if $\vert \delta(q,a) \vert \geq 1$ for all $q \in \astates$ and $a \in \Sigma$. Every automaton $\automaton$ can be made total by adding a rejecting sink state reached via all transitions for which $\delta(q,a) = \emptyset$ in $\automaton$. W.l.o.g. we assume all states of $\automaton$ to be reachable from $\ainit$. 

A \emph{deterministic finite automaton} (\emph{DFA}) is an NFA that satisfies $\vert \delta(q, a) \vert \leq 1$ for every $q \in \astates$ and $a \in \Sigma$. 
An \emph{unambiguous finite automaton} (\emph{UFA}) is an NFA that has at most one accepting run on every $w \in \Sigma^*$. NFA, UFA and DFA all characterize the regular languages.

\subparagraph*{$\boldsymbol{\omega}$-Automata.}
A \emph{Nondeterministic Büchi automaton} (\emph{NBA}) has the same structure as an NFA, but operates on $\omega$-words. A \emph{run} of an NBA $\automaton$ on $w = a_1a_2 \ldots \in \Sigma^\omega$ is an infinite sequence $\rho = q_0q_1 \ldots \in \astates^\omega$ with $q_0 = \ainit$ and $q_{i+1} \in \delta(q_i, a_{i+1})$ for all $i \in \mathbb{N}$. The run $\rho$ is \emph{accepted} by $\automaton$ if $\mathrm{inf}(\rho) \cap \afinal \neq \emptyset$, where $\mathrm{inf}(\rho) = \{q \in \astates \mid q \text{ is visited infinitely often along } \rho \}$, and in this case $\rho$ is called \emph{accepting}. The \emph{language} of $\automaton$, denoted $\accLanguage(\automaton)$, is defined as for NFA. \emph{Unambiguous Büchi automata} (\emph{UBA}) are NBA in which every word $w \in \Sigma^\omega$ has at most one accepting run. Both NBA and UBA characterize the set of $\omega$-regular languages.

\emph{Deterministic Rabin automata} (\emph{DRA}) are deterministic $\omega$-automata that also characterize the set of $\omega$-regular languages. Instead of final states, a DRA $\automaton$ has $n$ so-called \emph{Rabin pairs} $(T_i, R_i)$ with $T_i, R_i \subseteq \astates$. A word $w \in \Sigma^\omega$ is \emph{accepted} by $\automaton$ if the unique run $\rho$ of $\automaton$ on $w$ satisfies $\mathrm{inf}(\rho) \cap T_i = \emptyset$ and $\mathrm{inf}(\rho) \cap R_i \neq \emptyset$ for some $1 \leq i \leq n$. The language $\accLanguage(\automaton)$ of a DRA $\automaton$ is defined as for NBA. W.l.o.g. we assume all DRA and DFA to be total. 

\subparagraph*{Ultimately periodic words and FDFA.} A word $w \in \Sigma^\omega$ is \emph{ultimately periodic} if $w = uv^\omega$ for a $(u,v) \in \Sigma^{*+}$. For every $\omega$-regular language $L$, $\UP(L) = \{uv^\omega \in L \mid (u,v) \in \Sigma^{*+}\} \neq \emptyset$~\cite{DMRSOA}. Moreover, for $\omega$-regular languages $L, L'$ it holds that $L = L'$ iff $\UP(L) = \UP(L')$ \cite{UPWROL}. 

A \emph{family of deterministic finite automata} (\emph{FDFA}) is a tuple $\fdfa = (\lfdfa, \{\pfdfa^q\}_{q \in \lfdfa})$. $\lfdfa$ is the \emph{leading automaton} and satisfies $\afinal^\lfdfa = \emptyset$, and each $\pfdfa^q$, $q \in \lfdfa$, is a \emph{progress automaton}. In an FDFA, $\lfdfa$ and all $\pfdfa^q$ are DFA. The \emph{size} of $\fdfa$ is $\vert \fdfa \vert = (\vert \lfdfa \vert, \max_{q \in \lfdfa} \vert \pfdfa^q \vert)$. FDFA process ultimately periodic words represented by \emph{decompositions} $(u,v) \in \Sigma^{*+}$. A decomposition $(u,v)$ is \emph{normalized} (w.r.t. $\lfdfa$) if $\lfdfa(u) = \lfdfa(uv)$, and $\fdfa$ \emph{accepts} $(u,v)$ if it is normalized and $v \in \accLanguage(\pfdfa^{q})$ for $q = \lfdfa(u)$. 
Let $\llbracket \fdfa \rrbracket = \{(u,v) \in \Sigma^{*+} \mid \fdfa \text{ accepts } (u,v)\}$ and $\accLanguage(\fdfa) = \{uv^\omega \mid \exists (x,y) \in \Sigma^{*+}\colon xy^\omega = uv^\omega \text{ and} (x,y) \in \llbracket \fdfa \rrbracket\}$. $\fdfa$ \emph{accepts} $w = uv^\omega$ if $w \in \accLanguage(\fdfa)$. 

An FDFA is \emph{saturated} if for all $w=uv^\omega$ it accepts either every or none of the normalized decompositions of $w$. Saturated FDFA characterize the $\omega$-regular languages \cite{FDFAAORL,SPFA}, i.e., for every saturated FDFA $\fdfa$ there is a unique $\omega$-regular language $\accLanguage_{\omega}(\fdfa)$ with $\UP(\accLanguage_{\omega}(\fdfa)) = \accLanguage(\fdfa)$, and for every $\omega$-regular language $L$ there is a saturated FDFA $\fdfa_L$ with $\accLanguage(\fdfa_L) = \UP(L)$. 

\subparagraph*{Markov chains.} We fix a countable set $AP$ of \emph{atomic propositions}. For a finite set $S$, let $\mathrm{Distr}(S) = \{\mu \colon S \to [0,1] \mid \sum_{s \in S} \mu(s) = 1\}$ contain the \emph{probability distributions} on $S$. 

A \emph{discrete-time Markov chain} (\emph{DTMC}) is a tuple $\dtmc = (\stateSpace, \prob, \initialState, \labelFunction)$ with $\stateSpace$ a finite set of \emph{states}, $\prob \colon \stateSpace \to \mathrm{Distr}(\stateSpace)$ a \emph{transition distribution function}, $\initialState \in \stateSpace$ a unique \emph{initial state}, and $\labelFunction\colon \stateSpace \to 2^{AP}$ a \emph{labeling function}. $\prob(s,s') = \prob(s)(s')$ denotes the probability to move from $s$ to $s'$ in one step, and $\mathrm{Succ}(s) = \{s' \in S \mid \prob(s,s') > 0\}$ is the set of \emph{successors} of $s$. For $s \in \stateSpace$, let $\dtmc_s$ be like $\dtmc$, but with initial state $s$. The \emph{size} of $\dtmc$ is $\vert \dtmc \vert = \vert \stateSpace \vert$.

A sequence $\pi = s_0s_1\ldots \in \stateSpace^{\omega}$ is a (\emph{infinite}) \emph{path} of $\dtmc$ if $s_0 = \initialState$ and $s_{i+1} \in \mathrm{Succ}(s_i)$ for all $i \in \mathbb{N}$.
$\pi[i]=s_i$ is the state at position $i$ of $\pi$, and  $\trace(\pi) = \labelFunction(s_0)\labelFunction(s_1)\ldots \in (2^{AP})^\omega$ is the \emph{trace} of $\pi$. \emph{Finite} paths $\pi = s_0s_1\ldots s_n \in \stateSpace^{n+1}$ and their traces are defined analogously, and if $\pi$ is finite we set $\mathrm{last}(\pi) = s_n$. $\Paths(\dtmc)$ is the set of all paths of $\dtmc$, and $\Paths^*(\dtmc)$ is the set of all finite paths of $\dtmc$. For $s \in \stateSpace$, let $\Paths(s) = \Paths(\dtmc_s)$ and $\Paths^*(s) = \Paths^*(\dtmc_s)$. 

$C \subseteq \stateSpace$ is a \emph{strongly connected component} (\emph{SCC}) of $\dtmc$ if all $s,s' \in C$ are reachable from one another via finite paths. An SCC $C$ is \emph{bottom} (\emph{BSCC}) if $\mathrm{Succ}(s) \subseteq C$ for all $s \in C$. 

For a DTMC $\dtmc$ and a deterministic automaton $\automaton$ with alphabet $2^{AP}$, the \emph{product} $\dtmc \otimes \automaton$ is a DTMC with state space $\stateSpace \times \astates$, initial state $\langle \initialState, \delta(\ainit, \labelFunction(\initialState))\rangle$, labeling function $\labelFunction'(\langle s, q \rangle) = \labelFunction(s)$ and transition probabilities $\prob'(\langle s,q \rangle, \langle s', q' \rangle) = \prob(s,s')$ if $q' = \delta(q, \labelFunction(s'))$ and $0$ otherwise. The DTMC $\dtmc \times \automaton$ is build similar to $\dtmc \otimes \automaton$, but with initial state $\langle \initialState, \ainit \rangle$. 

We consider the standard probability measure $\probMeasure^{\dtmc}$ on subsets of $\Paths(\dtmc)$, defined via \emph{cylinder sets} $\mathit{Cyl}(\pi) = \{ \sigma \in \Paths(\dtmc) \mid \pi \text{ is a prefix of } \sigma \}$ for $\pi = s_0s_{1} \dots s_n \in \Paths^*(\dtmc)$. We abbreviate $\probMeasure^{\dtmc}(\mathit{Cyl}(\pi))$ by $\probMeasure^{\dtmc}(\pi)$ and the measure yields $\probMeasure^{\dtmc}(\pi) = 0$ if $s_0 \neq \initialState$ and $\probMeasure^{\dtmc}(\pi)= \prod_{j = 0}^{n-1} \prob(s_j,s_{j+1})$ otherwise. Given an $\omega$-regular property $E$, $\probMeasure^{\dtmc}(E)$ is the probability of $\dtmc$ to generate a path with trace in $E$. For details on DTMCs see, e.g., \cite{FMC,PoMC}.

\subparagraph*{LTL.} For $a \in AP$, formulas of the \emph{linear temporal logic} (\emph{LTL}) are formed w.r.t. the grammar 
\begin{align*}
	\varphi \Coloneqq \mathit{true} \mid a \mid \lnot \varphi \mid \varphi_1 \lor \varphi_2 \mid \bigcirc \varphi \mid \varphi_1 \mathsf{U} \varphi_2.
\end{align*}
Here, $\bigcirc$ is the \emph{next}-operator, so $\sigma \in (2^{AP})^\omega$ satisfies $\bigcirc \varphi$ iff $\varphi$ holds in $\sigma[1]$, and $\mathsf{U}$ is the \emph{until}-operator, so $\sigma$ satisfies $\varphi_1 \mathsf{U} \varphi_2$ iff, along $\sigma$, $\varphi_1$ holds until $\varphi_2$ is true. As syntactic sugar we define $\lozenge \varphi \equiv \mathit{true} \mathsf{U} \varphi$ which is satisfied by $\sigma$ iff $\varphi$ holds \emph{eventually} along $\sigma$.  Let $\Words(\varphi) = \{w \in (2^{AP})^\omega \mid w \text{ satisfies } \varphi \}$. For a DTMC $\dtmc$, $\probMeasure^\dtmc(\varphi)$ is the probability that $\dtmc$ generates a path $\pi$ with $\trace(\pi) \in \Words(\varphi)$. 
For details on LTL see, e.g., \cite{TLP,PoMC}.
\section{Model Checking DTMCs against Saturated FDFA}\label{sec:fdfa-mc}
Let $\dtmc$ be a DTMC and $E$ an $\omega$-regular property.
The \emph{probabilistic model checking problem} for $\dtmc$ and $E$ asks to compute $\probMeasure^{\dtmc}(E)$, i.e., the probability that $\dtmc$ generates a path $\pi$ with $\trace(\pi) \in E$. In this section we show how to efficiently solve the probabilistic model checking problem in the case that $E$ is represented by a saturated FDFA $\fdfa = (\lfdfa, \{\pfdfa^q\}_{q \in \lfdfa})$.
\begin{theorem}\label{thm:poly-time-mc-for-fdfa}
	For a given DTMC $\dtmc$ and a given saturated FDFA $\fdfa$ with $\accLanguage_\omega(\fdfa) = E$, the probability $\probMeasure^{\dtmc}(E)$ can be computed in time polynomial in the sizes of $\dtmc$ and $\fdfa$. 
\end{theorem}

To prove \Cref{thm:poly-time-mc-for-fdfa} we present a polynomial-time probabilistic model checking algorithm for DTMCs $\dtmc$ against saturated FDFA $\fdfa$. The algorithm follows the standard approach for model checking $\dtmc$ against deterministic $\omega$-automata $\automaton$ \cite{PoMC}, which boils down to (1) constructing the product $\dtmc \otimes \automaton$, (2) computing the BSCCs $B_1, \ldots, B_k$ of $\dtmc \otimes \automaton$ that satisfy the acceptance condition of $\automaton$, and (3) computing the probability of reaching one of these BSCCs $B_i$ in the product. When dealing with (saturated) FDFA instead of classical deterministic $\omega$-automata, the challenge is that FDFA can only process the ultimately periodic words of a given $\omega$-regular language. Hence, the decision in (2) if a BSCC of the product of $\dtmc$ and (in this case) the leading automaton $\lfdfa$ of $\fdfa$ is accepting requires some extra care. 

Intuitively, our algorithm works as follows: First, the set of BSCCs of $\dtmc \otimes \lfdfa$ is computed. Afterwards, for each BSCC $B$ of $\dtmc \otimes \lfdfa$, an arbitrary state $\langle s, q \rangle \in B$ is picked, and the product $I_B = \dtmc_s \times \lfdfa_q \times \pfdfa^q$ is analyzed for the existence of a BSCC that does not contain a state $\langle s, q, f \rangle$ with $f \in \afinal^{\pfdfa^q}$. Importantly, the existence of such a \emph{rejecting} BSCC in $I_B$ does not depend on the specific choice of $\langle s, q \rangle \in B$ (\Cref{lem:single-inner-product-sufficient}), and checking if $I_B$ has a rejecting BSCC can be done by a simple graph analysis in time linear in $\vert I_B \vert$ (\Cref{lem:good-or-bad-via-graph-analysis}). If $I_B$ does not contain a rejecting BSCC then $B$ is classified as \emph{good}. In the end, $\probMeasure^{\dtmc}(E)$ equals the probability of reaching a good BSCC in $\dtmc \otimes \lfdfa$ (\Cref{prop:fdfa-mc-preservation-of-probabilities}).
The majority of this section is devoted to the proof of correctness of this algorithm, and hence of \Cref{thm:poly-time-mc-for-fdfa}.

\smallskip

The upcoming proofs make use of an auxiliary DRA $\automaton$ with $\accLanguage(\automaton) = E$ and exploit the fact that $\probMeasure^\dtmc(E)$ equals the probability to reach an \emph{accepting} BSCC in $\dtmc \otimes \automaton$ \cite{PoMC}. Here, a BSCC $B$ of $\dtmc \otimes \automaton$ is accepting if there is a Rabin pair $(T_i, R_i)$ of $\automaton$ with $B \vert_\automaton \cap T_i = \emptyset$ and $B\vert_\automaton \cap R_i \neq \emptyset$, where $B \vert_\automaton = \{q \in \automaton \mid \exists \ s \in \dtmc\colon \langle s, q \rangle \in B\}$. Otherwise, $B$ is \emph{rejecting}. 

The following lemma shows that if a BSCC $B$ of $\dtmc \otimes \automaton$ is accepting, all paths that reach $B$ and visit all its states infinitely often satisfy $E$, while if $B$ is rejecting all these paths violate $E$. Recall that for $\sigma \in \Paths(\dtmc \otimes \automaton)$, $\trace(\sigma)$ equals the trace of the path of $\dtmc$ obtained from the projection of $\sigma$ to its first components. 

\begin{restatable}{lemma}{LemAcceptingAndRejectingBSCCinDRAProduct}\label{lem:accepting-and-rejecting-bscc-dra-product}
	Let $u \in \Paths^*(\dtmc \otimes \automaton)$ end in $\langle s, p \rangle \in B$ for some BSCC $B$ of $\dtmc \otimes \automaton$, and let $\pi \in \Paths^{\dtmc \otimes \automaton}(\langle s, p \rangle)$ be such that, along $\pi$, every state of $B$ is visited infinitely often.  
	\begin{enumerate}
		\item If $B$ is accepting then $\trace(u \pi) \in E$. 
		\item If $B$ is rejecting then $\trace (u \pi) \notin E$.
	\end{enumerate}
\end{restatable}
\begin{proof}[Proof Sketch.]
	If $B$ is accepting then there is a Rabin pair $(T_i, R_i)$ with $B \vert_\automaton \cap T_i = \emptyset$ and $B \vert_\automaton \cap R_i \neq \emptyset$. Because all states of $B$ are visited infinitely often along $\pi$, so is a state in $R_i$, while no state in $T_i$ is visited at all. Thus, $\trace(u\pi) \in \accLanguage(\automaton)= E$. If $B$ is rejecting then for every $(T_i, R_i)$ we have $T_i \cap B\vert_\automaton \neq \emptyset$, so $\inf(u\pi) \cap T_i \neq \emptyset$, or $R_i \cap  B \vert_\automaton= \emptyset$. Hence, $\trace(u\pi) \notin E$. 
\end{proof}

Let $\langle s, q \rangle$ be a state in a BSCC of $\dtmc \otimes \lfdfa$. We define $I_{\langle s, q \rangle} = \dtmc_s \times \lfdfa_q \times \pfdfa^q$ and say that a state $\langle s, q, f \rangle$ in a BSCC of $I_{\langle s, q \rangle}$ is \emph{accepting} if $f \in \afinal^{\pfdfa^q}$. Note that any accepting state of $I_{\langle s, q \rangle}$ has $s$ in its first and $q$ in its second component. W.l.o.g. we assume that accepting states in $I_{\langle s, q \rangle}$ are labeled with a fresh atomic proposition $\mathit{accept}$. A BSCC $B$ of $I_{\langle s,q \rangle}$ is \emph{rejecting} if it does not contain an accepting state. 
It turns out that the specific choice of $\langle s, q \rangle$ in a BSCC of $\dtmc \otimes \lfdfa$ has no influence on the existence of rejecting BSCCs in $I_{\langle s, q \rangle}$. 
\begin{restatable}{lemma}{LemSingleInnerProduct}\label{lem:single-inner-product-sufficient}
	Let $C$ be a BSCC of $\dtmc \otimes \lfdfa$ and $\langle s,q \rangle, \langle s', q' \rangle \in C$. If $I_{\langle s, q \rangle}$ has a rejecting BSCC then $I_{\langle s', q' \rangle}$ does as well. 
\end{restatable}
\begin{proof}[Proof Sketch.]
	Let $u, v, w$ be such that $u$ reaches $\langle s', q' \rangle$ in $\dtmc \otimes \lfufa$, $v$ takes $\dtmc \otimes \lfufa$ from $\langle s', q' \rangle$ to $\langle s, q \rangle$, and $w$ reaches a rejecting BSCC in $I_{\langle s, q \rangle}$. Assume that $I_{\langle s', q' \rangle}$ does not have a rejecting BSCC. Then there is an $x$ such that $vwx$ reaches an accepting state $\langle s', q', f \rangle$ in $I_{\langle s', q' \rangle}$, and so $(\trace(u), \trace(vwx)) \in \llbracket \fdfa \rrbracket$. However, $\lfdfa(\trace(uv)) = q = \delta^\lfdfa(q, \trace(wxv))$ and $\trace(wxv) \notin \accLanguage(\pfdfa^q)$, where the latter follows since a rejecting BSCC of $I_{\langle s, q \rangle}$ is reached upon $w$. Hence, $(\trace(uv), \trace(wxv)) \notin \llbracket \fdfa \rrbracket$, contradicting the saturation of $\fdfa$. 
\end{proof}

Let $B$ be a BSCC of $\dtmc \otimes \lfdfa$ and let $\langle s, q \rangle$ be an arbitrary state in $B$. We set $I_B = I_{\langle s, q \rangle}$ and say that $B$ is \emph{good} iff $I_B$ does not have a rejecting BSCC. By \Cref{lem:single-inner-product-sufficient}, the classification of $B$ as \emph{good} is independent of the choice of $\langle s, q \rangle$. Deciding if $B$ is good can be done efficiently.\unskip 
\begin{restatable}{lemma}{LemGraphAnalysisForGoodBad}\label{lem:good-or-bad-via-graph-analysis}
	Deciding if a BSCC $B$ of $\dtmc \otimes \lfdfa$ is good is possible in time linear in $\vert I_B \vert$. 
\end{restatable}
\begin{proof}[Proof Sketch.]
	It holds that $I_B$ has a rejecting BSCC iff $\probMeasure^{I_B}(\lozenge \mathit{accept}) < 1$, which is decidable in time linear in $\vert I_B \vert$, see, e.g., \cite{PoMC}. 
\end{proof}

Recall that $\automaton$ is an auxiliary DRA with $\accLanguage(\automaton) = E = \accLanguage_\omega(\fdfa)$. 
The next lemma shows a direct connection between reaching good BSCCs of $\dtmc \otimes \lfdfa$ and accepting BSCCs of $\dtmc \otimes \automaton$.  
\begin{restatable}{lemma}{LemAccBSCCiffGoodBSCCFDFA}\label{lem:accepting-bscc-iff-good-bscc}
	Let $u\in \Paths^*(\dtmc)$ such that the lifting of $u$ to $\dtmc \otimes \automaton$ reaches a BSCC $B_\automaton$, and the lifting of $u$ to  $\dtmc \otimes \lfdfa$ reaches a BSCC $B_\lfdfa$. Then $B_\automaton$ is accepting iff $B_\lfdfa$ is good.
\end{restatable}
\begin{proof}[Proof Sketch.]
	Let $B_\automaton$ be accepting and assume that $I_{B_\lfdfa} = I_{\langle s, q \rangle}$ has a rejecting BSCC. Then there is a finite cycle $v$ in $B_\automaton$ visiting all states of $B_\automaton$, while if followed in $I_{B_\lfdfa}$ ends in a state $\langle s, q, r \rangle$ in the rejecting BSCC. Hence, $uv^\omega \in \accLanguage(\automaton) =  E$ by \Cref{lem:accepting-and-rejecting-bscc-dra-product}, but $uv^\omega \notin \accLanguage_\omega(\fdfa) = E$, a contradiction. 
	The reverse direction can be shown similarly. 
\end{proof}

The precondition of \Cref{lem:accepting-bscc-iff-good-bscc} requires that the finite path $u$ is such that its lifting to both $\dtmc \otimes \lfdfa$ and $\dtmc \otimes \automaton$ reaches a BSCC. However, the final model checking algorithm does not have access to the DRA $\automaton$, and hence has no information about the runs of $\dtmc \otimes \automaton$ on $u$. It turns out that this is not a problem, since after reaching a BSCC in $\dtmc \otimes \lfdfa$ on $u$ either all BSCCs of $\dtmc \otimes \automaton$ reachable after $u$ are accepting, or all of them are rejecting.  

This is the case because of the following technical lemma, which states that all finite words reaching the same state in the leading automaton $\lfdfa$ of an FDFA $\fdfa$ for an $\omega$-regular property $E$ have the same \emph{residual languages} w.r.t. $E$. For a given $x \in \Sigma^*$, the residual language of $x$ w.r.t. $E$ is $\resLang_x = \{z \in \Sigma^\omega \mid xz \in E\}$.  \begin{restatable}{lemma}{LemResidualLanguages}\label{lem:residual-languages-coincide}
	Let $x, y \in \Sigma^*$ and let $\lfdfa$ be the leading automaton of an FDFA $\fdfa$ for $E$. If $\lfdfa(x) = \lfdfa(y)$ then $\resLang_x = \resLang_y$.
\end{restatable}
\begin{proof}[Proof Sketch.]
	Let $Q(x) = Q(y)$.	Then $xuv^\omega \in E$ iff $yuv^\omega \in E$, and so $\UP(\resLang_x) = \UP(\resLang_y)$. Because $\resLang_x$ and $\resLang_y$ are $\omega$-regular, this implies  $\resLang_x = \resLang_y$ \cite{UPWROL}. 
\end{proof}

Let $u \in \Paths^*(\dtmc)$ be such that the lifting of $u$ to $\dtmc \otimes \lfdfa$ reaches a state $\langle s, q \rangle$ in a BSCC $B_\lfdfa$. Because $\dtmc \otimes \automaton$ is a DTMC, almost surely a finite extension $v$ of $u$ is generated by $\dtmc$ such that the lifting of $uv$ to $\dtmc \otimes \automaton$ ends in a BSCC $B_\automaton$. As $\dtmc \otimes \lfdfa$ is in $\langle s, q \rangle \in B_\lfdfa$ after $u$, almost surely $\dtmc$ generates an extension $w$ of $uv$ such that the lifting of $uvw$ to $\dtmc \otimes \lfdfa$ also ends in $\langle s, q \rangle$. But then $\resLang_u = \resLang_{uvw}$ by \Cref{lem:residual-languages-coincide}, and so either almost surely a path with prefix $u$ is in $E$ (if $B_\automaton$ is accepting), or violates $E$ (if $B_\automaton$ is rejecting). 

Hence, after entering a BSCC $B_\lfdfa$ of $\dtmc \otimes \lfdfa$ the probability of $\dtmc$ to generate a path in $E$ equals $1$ (if $B_\lfdfa$ is good), or it equals $0$. 
Let $\mathit{accBSCC}$ contain all accepting BSCCs of $\dtmc \otimes \automaton$, and $\mathit{goodBSCC}$ contain all good BSCCs of $\dtmc \otimes \lfufa$.
Because $\probMeasure^\dtmc(E) = \probMeasure^{\dtmc \otimes \automaton}(\lozenge \mathit{accBSCC})$~\cite{PoMC}, \Cref{lem:accepting-bscc-iff-good-bscc} implies the following result. 
\begin{restatable}{proposition}{PropComputationOfProbabilites}\label{prop:fdfa-mc-preservation-of-probabilities}
	$\probMeasure^{\dtmc}(E) = \probMeasure^{\dtmc \otimes \lfdfa} (\lozenge \mathit{goodBSCC})$.
\end{restatable}

The following algorithm efficiently computes the satisfaction probability of a DTMC $\dtmc$ against an $\omega$-regular property $E$ represented by a saturated FDFA $\fdfa = (\lfdfa, \{\pfdfa^q\}_{q \in \lfdfa})$: \vspace{0.1cm}
\begin{enumerate}
	\item Compute the set $\mathcal{B}$ of BSCCs of $\dtmc \otimes \lfdfa$. 
	\item For every $B \in \mathcal{B}$ pick an arbitrary $\langle s, q \rangle \in B$ and check if $I_B = I_{\langle s, q \rangle} = \dtmc_s \times \lfdfa_q \times \pfdfa^q$ contains a rejecting BSCC. If this is not the case, add $B$ to the set $\mathit{goodBSCC}$. 
	\item Compute $\probMeasure^{\dtmc}(E) = \probMeasure^{\dtmc \otimes \lfdfa}(\lozenge \mathit{goodBSCC})$.
\end{enumerate} \vspace{0.1cm}
The correctness of the algorithm follows from \Cref{prop:fdfa-mc-preservation-of-probabilities}.
Regarding its complexity, the procedure boils down to obtaining the set of good BSCCs of $\dtmc \otimes \lfdfa$ and a subsequent computation of reachability probabilities in $\dtmc \otimes \lfdfa$. Deciding if a BSCC $B$ of $\dtmc \otimes \lfdfa$ is good can be done in time linear in the size of $I_B$ (\Cref{lem:good-or-bad-via-graph-analysis}), while the reachability probability of $\mathit{goodBSCC}$ in $\dtmc \otimes \lfdfa$ can be computed in time polynomial in $\vert \dtmc \otimes \lfdfa \vert = \vert \dtmc \vert \cdot \vert \lfdfa \vert$ by solving a system of linear equations \cite{PoMC}. Hence, the overall runtime of the algorithm is polynomial in the sizes of $\dtmc$ and $\fdfa$. This finishes the proof of \Cref{thm:poly-time-mc-for-fdfa}.

\smallskip 

To conclude this section we discuss how to obtain an FDFA for a given $\omega$-regular property. In the literature, FDFA are mainly applied in learning algorithms for $\omega$-regular languages \cite{LROL,NLABAFDFACT,NFFARLORL,SPFA}. These algorithms naturally yield a way to compute FDFA representations for a given $\omega$-regular specification. An orthogonal approach is the direct construction of saturated FDFA from, e.g., LTL formulas. The only direct translation from LTL to FDFA that we are aware of is presented in \cite{BALi}. It takes inspiration from the works of Esparza et al. \cite{LTLLDBADPA,UTLTLOA} on the construction of $\omega$-automata from LTL formulas and, in the worst case, incurs a double-exponential blow-up. 
Unfortunately, this blow-up is inevitable. 
\begin{restatable}[\cite{DFpersCom,CLTLFE}]{proposition}{LemDoubleExponentialBlowUp}\label{prop:double-exponential-blow-up-ltl-to-fdfa}
	There is a family of LTL-formulas $\varphi_n$ of length polynomial in $n$ such that any saturated FDFA for $\varphi_n$ has a leading automaton of size at least $2^{2^n}$.
\end{restatable}
\begin{proof}[Proof Sketch.]
	Let $L_n' = \{x \# w \# y \$ w \mid x, y \in \{0, 1, \#\}^*, w \in \{0,1\}^n\}$ and $L_n = L_n' \cdot \#^\omega$.  
	For each $n$ there is an LTL formula $\varphi_n$ of length quadratic in $n$ with $\Words(\varphi_n) = L_n$ \cite{FLTBT}. The leading automaton of any saturated FDFA for $L_n$ must ``recognize'' if a prefix is in $L_n'$. But this requires a DFA with at least $2^{2^n}$ states \cite{Alternation}.
\end{proof}

Nevertheless, it was shown in \cite{FDFAAORL} that saturated FDFA can be exponentially more succinct than deterministic Rabin, Streett or parity automata, which are types of $\omega$-automata often used in probabilistic verification \cite{PoMC}. The probabilistic model checking procedure outlined in this section may therefore still reduce the total verification time significantly. 
\section{Families of Unambiguous Finite Automata}\label{sec:fufa}

It is well-known that an NFA or a UFA for a given regular language $L$ may require exponentially fewer states than a DFA for $L$ \cite{SDCFRFL}. Hence, allowing nondeterminism in the leading automaton and/or the progress automata of an FDFA might yield families of finite automata of significantly reduced size. Based on this observation, we define the novel notion of \emph{families of unambiguous finite automata} (\emph{FUFA}), which have leading NFA and progress UFA. 
\begin{definition}\label{def:fufa}
	A \emph{family of unambiguous finite automata} \emph{(FUFA)} is a tuple $\fufa = (\lfufa, \{\pfufa^q\}_{q \in \mstates})$ with $\mstates \subseteq \lfufa$, where $\lfufa$ is a leading NFA and $\pfufa^q$, $q \in \mstates$, are progress UFA. 
\end{definition}

Let $\fufa = (\lfufa, \{\pfufa^q\}_{q \in\mstates})$ be an FUFA. In contrast to FDFA, it is not required in an FUFA that every state $q \in \lfufa$ has a corresponding progress automaton $\pfufa^q$, but only those $q$ contained in $\mstates \subseteq \lfdfa$.  
As every DFA is both a UFA and an NFA, all FDFA can be interpreted as FUFA in which $A$ contains all states of $\lfufa$. On the other hand, not every FUFA is an FDFA, since the leading automaton and the progress automata of an FUFA can contain nondeterminism. 

We extend the notion of (normalized) acceptance of decompositions $(u,v) \in \Sigma^{*+}$ from FDFA to FUFA. The main difference between the acceptance of $(u,v)$ in an FDFA and an FUFA is that in the latter, acceptance of $v$ by a progress automaton can only be defined if a state in $\mstates$ is reachable on $u$, i.e., if $\mstates_u \coloneqq \lfufa(u) \cap \mstates$ is nonempty. Moreover, the leading automaton of an FUFA can be nondeterministic, so $\lfufa(u)$ may contain multiple states and it is  possible that $(u,v)$ is normalized only w.r.t. some of them. Hence, accepted normalized decompositions $(u,v)$ are required to satisfy the normalization condition w.r.t. at least one $q \in \mstates_u$ for which $v \in \accLanguage(\pfufa^q)$. These observations are combined in the following definition. \unskip
\begin{definition}\label{def:normalized-acceptance}
	Let $\fufa = (\lfufa, \{\pfufa^q\}_{q \in \mstates})$ be an FUFA. A decomposition $(u,v) \in \Sigma^{*+}$ is \emph{$q$-normalized} if $q \in \mstates_u$ and $q \in \delta^{\lfufa}(q, v)$, and $(u,v)$ is \emph{normalized} if it is $q$-normalized w.r.t. some $q \in \mstates$. Moreover, $\fufa$ \emph{accepts} the decomposition $(u,v)$ if there is a $q \in \mstates$ such that $(u,v)$ is $q$-normalized and $v \in \accLanguage(\pfufa^{q})$. 
	The sets $\llbracket \fufa \rrbracket$, $\accLanguage(\fufa)$ and $\accLanguage_{\omega}(\fufa)$ are defined as for FDFA. 
\end{definition}

Similar to FDFA, we also require FUFA to be \emph{saturated}. While in a saturated FDFA either every or no normalized decomposition of a word is accepted, in a saturated FUFA we again have to restrict the definition to those decompositions $(u,v) \in \Sigma^{*+}$ with $\mstates_u \neq \emptyset$. 
\begin{definition}\label{def:saturated}
	An FUFA $\fufa$ is saturated if for every ultimately periodic word $w$ either every normalized decomposition of $w$ is accepted by $\fufa$, or none of them is. 
\end{definition}

Let $\fufa = (\lfufa, \{\pfufa^q\}_{q \in \mstates})$ be a saturated FUFA, $w \in \Sigma^\omega$ an ultimately periodic word, and $(u,v)$ a normalized decomposition of $w$. The central idea behind FUFA is to obtain an \emph{unambiguous extension} of FDFA. The nondeterminism of $\lfufa$ and $\pfufa^q$, $q \in \mstates$, can, however, cause some undesired behavior when viewing FUFA under this premise. For instance, it is possible that there are $q_1, \ldots, q_n \in \mstates_u$ such that $(u,v)$ is $q_i$-normalized and $v \in \accLanguage(\pfufa^{q_i})$ for all $1 \leq i \leq n$. This is not in the spirit of unambiguous automata, which have at most one accepting run for every word. Adjusted to FUFA this means that for every decomposition $(u,v)$ of $w$ there should be a \emph{unique} $q \in \mstates_u$ for which $(u,v)$ is $q$-normalized and $v \in \accLanguage(\pfufa^q)$. As another example, a saturated FDFA $\fdfa = (\lfdfa_\fdfa, \{\pfdfa_\fdfa^q\}_{q \in \lfdfa_\fdfa})$ satisfies that if $(u,v) \in \llbracket \fdfa \rrbracket$ then also $(u, v^n) \in \llbracket \fdfa \rrbracket$ for all $n$ \cite{SPFA}. Hence, for $q = \lfdfa_\fdfa(u)$ it holds that $v^n \in \accLanguage(\pfdfa_\fdfa^q)$ for all $n \geq 1$. In saturated FUFA, however, decompositions $(u,v^i)$ and $(u,v^j)$ for $i \neq j$ can be normalized w.r.t. different $q_i$ and $q_j$ in $\mstates_u$, and accepted via different progress automata $\pfufa^{q_i}$ and $\pfufa^{q_j}$. This distinguishes saturated FUFA significantly from saturated FDFA. 

To bring saturated FUFA closer to saturated FDFA and unambiguous automata, we require saturated FUFA to be \emph{decomposition-unambiguous} and \emph{power-unambiguous}.\unskip
\begin{definition}\label{def:properties-fufa}
	Let $\fufa = (\lfufa, \{\pfufa^q\}_{q \in \mstates})$ be a saturated FUFA. 
	\begin{enumerate}
		\item $\fufa$ is \emph{decomposition-unambiguous} if for every $(u,v) \in \Sigma^{*+}$ there is at most one $q \in \mstates_u$ such that $(u,v)$ is $q$-normalized and $v \in \accLanguage(\pfufa^q)$.
		\item $\fufa$ is \emph{power-unambiguous} if for all $q$-normalized decompositions $(u,v) \in \llbracket \fufa \rrbracket$ with $v \in \accLanguage(\pfufa^q)$ and every $f \in \pfufa^q(v) \cap \afinal^{\pfufa^q}$ it holds that $\afinal^{\pfufa^q} \cap \delta^{\pfufa^q}(f, v^n) \neq \emptyset$ for all $n \geq 1$. 
	\end{enumerate}
\end{definition}

Decomposition- and power-unambiguity imply several useful properties that are guaranteed in FUFA $\fufa$ that satisfy (one of) these definitions. 
For instance, decomposition-unambiguity of $\fufa$ yields that for every $(u,v) \in \llbracket \fufa \rrbracket$ there is a unique $q \in \mstates_u$ such that $(u,v)$ is $q$-normalized and $v \in \accLanguage(\pfufa^q)$. If $\fufa$ is, on the other hand, power-unambiguous then no $f \in \afinal^{\pfufa^q}$ is terminal, and $v \in \accLanguage(\pfufa^q)$ for $v$ the period of a $q$-normalized decomposition $(u,v)$ implies $v^n \in \accLanguage(\pfufa^q)$ for all $n \geq 1$. In fact, power-unambiguity of $\fufa$ even implies that $f_{n+1} \in \delta^{\pfufa^q}(f_n,v)$ for all $n \geq 1$, where $f_{i}$ denotes the final state reached by the accepting run of $\pfufa^q$ on $v^i$. To see that this is indeed the case consider, for $n \geq 1$, the word $v^{n \cdot (n+1)}$. Then $\pfufa^q$ has an accepting run $\pi_1$ on $v^{n \cdot (n+1)}$ that is in $f_n$ after $v^n$. Such a $\pi_1$ exists since $f_n \in \pfufa^q(v^n) \cap \afinal^{\pfufa^q}$ and, as $\fufa$ is power-unambiguous, $\afinal^{\pfufa^q} \cap \delta^{\pfufa^q}(f_n, v^{n \cdot (n+1) - n}) = \afinal^{\pfufa^q} \cap \delta^{\pfufa^q}(f_n, (v^{n})^n)\neq \emptyset$. Similarly, $\pfufa^q$ has an accepting run $\pi_2$ on $v^{n \cdot (n+1)}$ that is in $f_{n+1}$ after $v^{n+1}$. Because $\pfufa^q$ is unambiguous, $\pi_1$ and $\pi_2$ coincide. Thus, the unique accepting run of $\pfufa^q$ on $v^{n \cdot (n+1)}$ must be in $f_n$ after $v^n$ and in $f_{n+1}$ after $v^{n+1}$, which is only possible if $f_{n+1} \in \delta^{\pfufa^q}(f_n, v)$. As a direct consequence, it follows that if $\fufa$ is power-unambiguous then for all $q$-normalized $(u,v) \in \llbracket \fufa \rrbracket$ the accepting run of $\pfufa^q$ on $v^j$ is a prefix of the accepting run of $\pfufa^q$ on $v^k$ for all $j < k$.  

The next lemma suggest that saturation as well as decomposition-unambiguity and power-unambiguity are natural requirements for unambiguous extensions of saturated FDFA. \unskip
\begin{restatable}{lemma}{LemFDFAareFUFA}\label{lem:fdfa-are-fufa}
	Let $\fdfa = (\lfdfa_\fdfa, \{\pfdfa_\fdfa^q\}_{q \in \lfdfa_\fdfa})$ be a saturated FDFA. Let $\fufa = (\lfufa_\fufa, \{\pfufa_\fufa^q\}_{q \in A})$ with $A = \lfufa_\fufa$ be the FUFA obtained by interpreting $\lfdfa_\fdfa$ as an NFA and the progress automata $\pfdfa_\fdfa^q$ as UFA. Then $\fufa$ is saturated, decomposition-unambiguous and power-unambiguous. 
\end{restatable}

If not mentioned otherwise, we from now on assume that all FUFA are saturated, decomposition-unambiguous and power-unambiguous.

\subsection{\texorpdfstring{Succinctness of FUFA compared to FDFA and $\omega$-Automata}{Succinctness of FUFA compared to FDFA and omega-Automata}}

Similar to FDFA, the \emph{size} of an FUFA $\fufa$ is $\vert \fufa \vert = (\vert \lfufa \vert, \max_{q \in \mstates}\vert \pfufa^q \vert)$.
We compare the succinctness of FUFA to other types of $\omega$-automata, starting with saturated FDFA. Like for the comparison of DFA and UFA \cite{SDCFRFL}, there are $\omega$-regular languages for which representations as saturated FDFA require exponentially more states than representations as FUFA. 
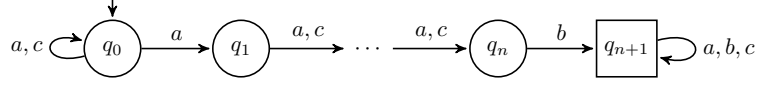
\begin{figure}[t]
	\centering
	\begin{tikzpicture}[->,>=stealth',shorten >=1pt,auto, semithick]
		\tikzstyle{every node} = [text = black, scale = 0.85]
		
		\node[state] (q0) {$q_0$}; 
		\node[above of = q0, node distance = 0.9cm] (init) {}; 
		\node[right of = q0, node distance = 2cm, state] (q1) {$q_1$}; 
		\node[right of = q1, node distance = 2cm] (dots) {$\cdots$}; 
		\node[right of = dots, node distance = 2cm, inner sep = 0pt, state] (qn-1) {$q_{n}$}; 
		\node[right of = qn-1, node distance = 2cm, state, draw, rectangle] (qn) {$q_{n+1}$};
		
		\path 
		(init) edge (q0)
		(q0) edge [loop left] node {$a, c$} (q0)
		(q0) edge node {$a$} (q1)
		(q1) edge node {$a,c$} (dots)
		(dots) edge node {$a,c$} (qn-1)
		(qn-1) edge node {$b$} (qn)
		(qn) edge [loop right] node {$a, b, c$} (qn)
		;
	\end{tikzpicture}
	\caption{The leading automaton of an FUFA for the family of languages $L_n$ used in the proof of \Cref{prop:succinctness-of-fufa}. The states in $\mstates$ are depicted as rectangles. Adapted from \cite{MCUA}.}
	\label{fig:leading-automaton-n-steps-before}
\end{figure}
\begin{restatable}{proposition}{PropSuccinctnessFUFA}\label{prop:succinctness-of-fufa}
	There is a family of $\omega$-regular languages $(L_n)_{n \geq 1}$ over a ternary alphabet such that every $L_n$ can be represented by a saturated, decomposition- and power-unambiguous FUFA of size $(n{+}2, 2)$, while every saturated FDFA for $L_n$ is of size at least $(2^n, 2)$. 
\end{restatable}
\begin{proof}[Proof Sketch.]
	Let $\Sigma = \{a,b,c\}$, $L_n' = (a+c)^*a(a+c)^{n-1}b$ and $L_n = L_n' \cdot \Sigma^\omega$. A saturated, decomposition- and power-unambiguous FUFA for $L_n$ of size $(n{+}2, 2)$ has a leading automaton as shown in \Cref{fig:leading-automaton-n-steps-before} with $\mstates = \{q_{n+1}\}$, and a progress automaton $\pfufa^{q_{n+1}}$ accepting $\Sigma^+$. Since an $\omega$-word is in $L_n$ iff it has a prefix in $L_n'$, the leading automaton of any saturated FDFA for $L_n$ must ``recognize'' $L_n'$. But any DFA for $L_n'$ requires at least $2^n$ states \cite{MCUA}. 
\end{proof}

In \cite{MCUA} it is argued that \emph{every} deterministic $\omega$-automaton accepting the language $L_n$ used in the proof of \Cref{prop:succinctness-of-fufa} has at least $2^n$ states. This yields the following corollary.
\begin{corollary}\label{cor:fufa-exp-smaller-than-deterministic-omega}
	There is a family of $\omega$-regular languages $(L_n)_{n \geq 1}$ over a ternary alphabet such that every $L_n$ can be represented by a saturated, decomposition- and power-unambiguous FUFA of size $(n{+}2, 2)$, while every deterministic $\omega$-automaton for $L_n$ needs at least $2^n$ states. 
\end{corollary}

Next, we compare the succinctness of FUFA to that of unambiguous Büchi automata. 

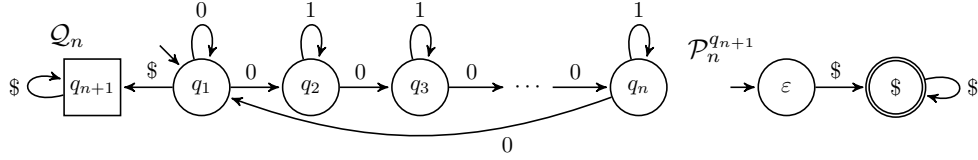
\begin{figure}[t]
	\centering
	\begin{tikzpicture}[->,>=stealth',shorten >=1pt,auto, semithick]
		\tikzstyle{every node} = [text = black, scale = 0.85]
		
		\node[state] (q1) {$q_1$}; 
		\node[state, right of = q1, node distance = 1.7cm] (q2) {$q_2$}; 
		\node[state, right of = q2, node distance = 1.7cm] (q3) {$q_3$}; 
		\node[right of = q3, node distance = 1.7cm] (temp) {$\cdots$}; 
		\node[state, right of = temp, node distance = 1.7cm] (qn) {$q_n$}; 
		\node[left of = q1, node distance = 1.7cm, state, draw, rectangle, inner sep = 0pt] (qn+1) {$q_{n+1}$}; 
		\node[above left of = q1, node distance = 1.1cm] (initq) {}; 
		\node[left of = qn+1, node distance = 0.4cm, yshift = 0.8cm, scale=1.2] {$\lfufa_n$}; 
		
		\path 
		(initq) edge (q1)
		(q1) edge [loop above] node {$0$} (q1)
		(q1) edge node [pos = 0.4] {$0$} (q2)
		(q1) edge node [above, pos = 0.4] {$\$$} (qn+1)
		(q2) edge [loop above] node {$1$} (q2)
		(q2) edge node [pos = 0.4] {$0$} (q3)
		(q3) edge [loop above] node {$1$} (q3)
		(q3) edge node [pos = 0.4] {$0$} (temp)
		(temp) edge node [pos = 0.4] {$0$} (qn)
		(qn) edge [loop above] node {$1$} (qn)
		(qn) edge [bend left = 20] node [pos = 0.3] {$0$} (q1)
		(qn+1) edge [loop left] node {$\$$} (qn+1)
		;
		
		\node[state, right of = qn, node distance = 2.3cm] (eps) {$\varepsilon$}; 
		\node[state, accepting, right of = eps, node distance = 1.7cm] (dollar) {$\$$}; 
		\node[left of = eps, node distance = 1cm] (initp) {}; 
		\node[above of = initp, node distance = 0.6cm, scale = 1.2] {$\pfufa_n^{q_{n+1}}$}; 
		
		\path 
		(initp) edge (eps)
		(eps) edge node [pos = 0.4] {$\$$} (dollar)
		(dollar) edge [loop right] node {$\$$} (dollar)
		;
		
	\end{tikzpicture}
	\caption{The FUFA $\fufa_n = (\lfufa_n, \pfufa_n^{q_{n+1}})$ as used in the proof of \Cref{thm:fufa-2-uba}.}
	\label{fig:leading-automaton-blow-up-fufa-2-uba}
\end{figure}
\begin{restatable}{theorem}{ThmFUFAToUBA}\label{thm:fufa-2-uba}
	There is a family of $\omega$-regular languages $(L_n)_{n \geq 1}$ over a ternary alphabet such that every $L_n$ can be represented by a saturated, decomposition- and power-unambiguous FUFA of size $(n{+}1, 2)$, while every UBA for $L_n$ has at least $2^n {-} 1$ states. 
\end{restatable}
\begin{proof}[Proof Sketch.]
	Let $L_n' = (0 + (01^*)^{n-1}0)^*$ and $L_n = L_n' \cdot \$^\omega$. The FUFA $\fufa_n$ depicted in \Cref{fig:leading-automaton-blow-up-fufa-2-uba} is saturated, decomposition- and power-unambiguous, and satisfies $\accLanguage_\omega(\fufa_n) = L_n$. Using techniques from \cite{FABA} we can show that every UBA $\automatonB_n$ for $L_n$ has a substructure $\automaton_n$ that, if interpreted as a UFA, accepts $L_n'$. But then $\automaton_n$ requires at least $2^n {-} 1$ states \cite{SEAFAPAFA}. 
\end{proof}

Because of the exponential lower bound for the translation from FUFA to UBA, the question arises if an exponential blow-up is also inevitable when translating saturated, decomposition- and power-unambiguous FUFA to NBA. It turns out that this is not the case. \unskip
\begin{restatable}{lemma}{LemFUFAtoNBA}\label{lem:fufa-to-nba-polynomial}
	Let $\fufa$ be a saturated, decomposition- and power-unambiguous FUFA of size $\vert \fufa \vert = (n,k)$. Then there is an NBA $\automaton$ of size $\mathcal{O}(n^2k^3)$ such that $\UP(\accLanguage(\automaton)) = \accLanguage(\fufa)$.
\end{restatable}
\begin{proof}[Proof Sketch.]
	We adjust the translation of saturated FDFA to NBA from \cite[Thm. 5.8]{FDFAAORL} to FUFA.
	For $q \in \mstates$ and $f \in \afinal^{\pfufa^q}$, let $M_q = \{v \in \Sigma^* \mid q \in \lfufa(v)\}$ and $N_{q,f} = \{v \in \Sigma^+ \mid q \in \delta^{\lfufa}(q, v) \land  f \in \pfufa^q(v) \land f \in \delta^{\pfufa^q}(f, v)\}$,	as well as $L = \bigcup_{\{(q,f) \mid q \in \mstates \land f \in \afinal^{\pfufa^q}\}} M_q \cdot N_{q,f}^\omega$. Then there is an NBA $\automaton$ for $L$ of size $\mathcal{O}(n^2k^3)$ with $\UP(\accLanguage(\automaton)) = \accLanguage(\fufa)$. 
\end{proof}

As every saturated, decomposition- and power-unambiguous FUFA $\fufa$ can be translated into an NBA $\automaton$ with $\UP(\accLanguage(\automaton))= \accLanguage(\fufa)$, any such FUFA represents an $\omega$-regular language. 
To show that also every $\omega$-regular language can be represented by a saturated, decomposition- and power-unambiguous FUFA, we provide a polynomial translation from UBA to FUFA.

Given a UBA $\automaton$, let $\fufa = (\lfufa, \{\pfufa^q\}_{q \in \mstates})$ be an FUFA whose leading automaton $\lfufa$ has the same structure as $\automaton$, but with $\afinal^\lfufa = \emptyset$, and let $\mstates = \afinal^{\automaton}$. Moreover, for $q \in \mstates$, let $\pfufa^q$ be such that $\accLanguage(\pfufa^q) = \{v \in \Sigma^+ \mid q \in \delta^\automaton(q, v)\}$. The automaton $\pfufa^q$ is obtained from $\automaton_q$ by setting $\afinal^{\pfufa^q} = \{q'\}$ for a fresh copy $q'$ of $q$ and defining $\delta^{\pfufa^{q}}$ 
such that if $p \neq q'$ then $\delta^{\pfufa^{q}}(p, a)$ is like $\delta^{\automaton}(p, a)$, but with transitions to $q$ replaced by transitions to $q'$, while $\delta^{\pfufa^q}(q', a)$ is like $\delta^{\automaton}(q,a)$, but with transitions to $q$ replaced by self-loops on $q'$.  
The resulting FUFA $\fufa$ is saturated, decomposition- and power-unambiguous, and satisfies $\accLanguage(\automaton) = \accLanguage_\omega(\fufa)$.
\begin{restatable}{lemma}{LemUBAtoFUFA}\label{lem:uba-to-fufa}
	Let $\automaton$ be a UBA. Then there is a saturated, decomposition- and power-unambiguous FUFA $\fufa = (\lfufa, \{\pfufa^q\}_{q \in \mstates})$ of size $\vert \fufa \vert = (\vert \automaton \vert, \mathcal{O}(\vert \automaton \vert))$ such that $\accLanguage(\automaton) = \accLanguage_\omega(\fufa)$. 	
\end{restatable}

The above translation from UBA $\automaton$ to FUFA $\fufa = (\lfufa, \{\pfufa^q\}_{q \in \mstates})$ does not guarantee that the automaton obtained from $\lfufa$ by setting $\afinal^\lfufa = \mstates$ is a UFA. To see this let $\automaton$ have three states $q_1, q_2, q_3$, with $q_1$ the initial state, $\afinal^{\automaton} = \{q_2, q_3\}$, and the only transitions $\delta^\automaton(q_1, a) = \{q_2, q_3\}, \delta^\automaton(q_2, b) = \{q_2\}$ and $\delta^\automaton(q_3, c) = \{q_3\}$. Then the FUFA constructed from $\automaton$ can move to both $q_2, q_3 \in \mstates$ on $a$, i.e., $\lfufa$ is not unambiguous w.r.t. $\mstates$. 

\begin{remark}
	By \Cref{thm:fufa-2-uba} there are $\omega$-regular languages $L_n$ for which saturated, decomposition- and power-unambiguous FUFA of size $(n{+}2, 2)$ exist, while every UBA for $L_n$ requires at least $2^n{-}1$ states. The FUFA obtained from such a UBA then also has a leading automaton of size at least $2^n{-}1$. Thus, the above translation is not optimal and can produce FUFA that are significantly larger than a minimal FUFA for the same language.  
\end{remark}

It is well-known that there is a single-exponential translation from LTL to UBA \cite{ATAAPV}. By \Cref{lem:uba-to-fufa}, translating LTL to saturated, decomposition- and power-unambiguous FUFA is of similar complexity. This is in contrast to the double-exponential lower bound for the translation from LTL to saturated FDFA \cite{DFpersCom,CLTLFE} shown in \Cref{prop:double-exponential-blow-up-ltl-to-fdfa}.
\begin{restatable}{corollary}{CorSingleExponentialLTLtoFUFA}\label{cor:single-exponential-translation-ltl-to-fufa}
	There is a single-exponential translation from LTL to saturated, power- and decomposition-unambiguous FUFA. 
\end{restatable}

The combination of \Cref{lem:fufa-to-nba-polynomial,lem:uba-to-fufa} yields the following result.  
\begin{restatable}{theorem}{ThmExpressivePowerFUFA}\label{thm:expressive-power-fufa}
	Saturated, decomposition- and power-unambiguous 
	FUFA are as expressive as the set of $\omega$-regular languages.
\end{restatable}

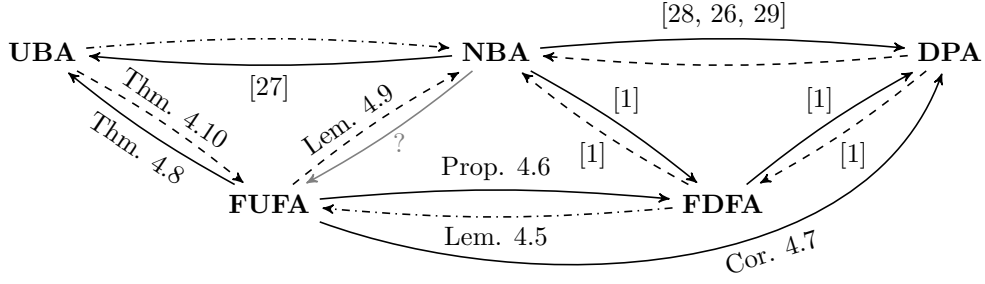
\begin{figure}[t]
	\centering
	\begin{tikzpicture}[->,>=stealth',shorten >=1pt,auto, semithick]
		\tikzstyle{every node} = [text = black, scale = 1]
		
		\node (nba) {\textbf{NBA}}; 
		\node [left of = nba, node distance = 6cm] (uba) {\textbf{UBA}};
		\node [right of = nba, node distance = 6cm] (dpa) {\textbf{DPA}}; 
		\node [below of = nba, node distance = 2cm] (temp) {}; 
		\node[right of = temp, node distance = 3cm] (fdfa) {\textbf{FDFA}}; 
		\node[left of = temp, node distance = 3cm] (fufa) {\textbf{FUFA}}; 
		
		\path[dash dot] 
		(fdfa) edge [bend left = 5] node {Lem. \ref{lem:fdfa-are-fufa}}(fufa)
		(uba) edge [bend left = 5] (nba)
		; 
		
		\path [dashed]
		(fdfa) edge [bend left = 5] node [xshift = 0.2cm] {\cite{FDFAAORL}} (nba)
		(uba) edge [bend left = 5] node [sloped] {Thm. \ref{lem:uba-to-fufa}} (fufa)
		(fufa) edge [bend left = 5] node [sloped, pos = 0.4] {Lem. \ref{lem:fufa-to-nba-polynomial}} (nba)
		(dpa) edge [bend left = 5] node [xshift = -0.2cm] {\cite{FDFAAORL}} (fdfa)
		(dpa) edge [bend left = 5] (nba)
		;
		
		\path 
		(fufa) edge [bend left = 5] node [sloped, below] {Thm. \ref{thm:fufa-2-uba}} (uba)
		(fufa) edge [bend left = 5] node {Prop. \ref{prop:succinctness-of-fufa}} (fdfa)
		(nba) edge [bend left = 5] node {\cite{FABA}} (uba)
		(nba) edge [bend left = 5] node [xshift = -0.2cm] {\cite{FDFAAORL}} (fdfa)
		(nba) edge [bend left = 5] node {\cite{CMDAIW,OBTOA,NBSADPA}} (dpa)
		(fdfa) edge [bend left = 5] node [xshift = 0.2cm] {\cite{FDFAAORL}} (dpa)
		(fufa) edge [out = 340, in = 250, looseness = 0.9] node[below, pos = 0.6, xshift = 0.2cm, sloped] {Cor. \ref{cor:fufa-exp-smaller-than-deterministic-omega}} (dpa)
		;
		
		\path[gray]
		(nba) edge [bend left = 5] node[xshift = -0.1cm, yshift = 0.1cm] {\textcolor{gray}{?}} (fufa)
		; 
	\end{tikzpicture}
	\vspace{-0.3cm}
	\caption{Overview of lower bounds for the state complexity of translations between FUFA, FDFA and different types of $\omega$-automata. Solid arrows indicate an inevitable exponential blow-up, dashed arrows denote polynomial translations, and dash-dotted arrows are used for constant translations.}
	\label{fig:succinctness-overview}
\end{figure}

\Cref{fig:succinctness-overview} summarizes the  succinctness results for FUFA compared to FDFA and other types of $\omega$-automata. In the figure, \textbf{FUFA} stands for the class of saturated, decomposition- and power-unambiguous FUFA, and \textbf{FDFA} stands for the class of saturated FDFA. \textbf{DPA} stands for the class of \emph{deterministic parity automata} (\emph{DPA}), which is another type of deterministic $\omega$-automata that characterizes the set of $\omega$-regular languages. By first using a polynomial translation from DPA to saturated FDFA \cite{FDFAAORL} and afterwards interpreting the resulting FDFA as an FUFA (\Cref{lem:fdfa-are-fufa}), we obtain a polynomial translation from DPA to FUFA. Moreover, the exponential lower bound for translating FUFA to DPA does, by \Cref{cor:fufa-exp-smaller-than-deterministic-omega}, also hold for all other types of deterministic $\omega$-automata. The only missing lower bound is that for translating NBA to FUFA. We expect this translation to come with an inevitable exponential blow-up, but leave a formal proof of this claim for future work. 

\subsection{Boolean Operations and Decision Procedures for FUFA}
Let $\fufa = (\lfufa, \{\pfufa^q\}_{q \in \mstates})$ be an FUFA of size $(n,k)$. The \emph{membership problem} for $\fufa$ and a given $(u,v) \in \Sigma^{*+}$ asks if $(u,v) \in \llbracket \fufa \rrbracket$. An algorithm that solves this problem first computes $\mstates_u$, then checks for all $q \in \mstates_u$ if $(u,v)$ is $q$-normalized, and afterwards decides if $v \in \accLanguage(\pfufa^q)$ for any such $q$. All of these operations can be realized by running (slightly modified versions of) $\lfufa$ on $u$ and some of the progress automata of $\fufa$ on $v$, yielding the following result. 
\begin{restatable}{lemma}{LemMembershipFUFA}
	Given an FUFA $\fufa$ and a decomposition $(u,v) \in \Sigma^{*+}$, the membership problem for $\fufa$ and $(u,v)$ can be solved in time polynomial in the sizes of $\fufa$, $\vert u \vert$ and $\vert v \vert$. 
\end{restatable}

Next, we consider the intersection operation for FUFA, i.e., we discuss how to obtain, from given FUFA $\fufa_1$ and $\fufa_2$, an FUFA $\fufaH$ with $\llbracket \fufaH \rrbracket = \llbracket \fufa_1 \rrbracket \cap \llbracket \fufa_2 \rrbracket$. Such an $\fufaH$ can be constructed by taking as leading automaton the parallel product of $\lfufa_1$ and $\lfufa_2$, setting $\mstates = \mstates_1 \times \mstates_2$, and defining $ \pfufa^{\langle q_1, q_2\rangle}$ as the parallel product of $\pfufa^{q_1}$ and $\pfufa^{q_2}$ for all $\langle q_1, q_2 \rangle\in \mstates$. \unskip
\begin{restatable}{lemma}{LemIntersectionFUFAViaDecompositionSets}\label{lem:intersection-fufa-decompositions}
	Let $\fufa_i = (\lfufa_i, \{\pfufa_i^q\}_{q \in \mstates_i})$ for $i \in \{1,2\}$ be two FUFA. Then there is an FUFA $\fufaH$ of size $(\vert \lfufa_1 \vert \cdot \vert \lfufa_2 \vert, \max_{\langle q_1,q_2\rangle \in \mstates_1 \times \mstates_2} \vert \pfufa_1^{q_1} \vert \cdot \vert \pfufa_2^{q_2} \vert)$ such that $\llbracket \fufaH \rrbracket = \llbracket \fufa_1 \rrbracket \cap \llbracket \fufa_2 \rrbracket$. Moreover, $\fufaH$ can be constructed such that saturation, decomposition-unambiguity and power-unambiguity are preserved, i.e., such that if both $\fufa_1$ and $\fufa_2$ have any of these properties, $\fufaH$ does as well.
\end{restatable}

Both the membership problem and the intersection operation described above consider (accepted) \emph{decompositions}. Extending the results to ultimately periodic words $w \in \Sigma^\omega$, i.e., to deciding if $w \in \accLanguage(\fufa)$ or constructing an FUFA $\fufaH$ such that $\accLanguage(\fufaH) = \accLanguage(\fufa_1) \cap \accLanguage(\fufa_2)$, is not straightforward. In the case of membership, an algorithm that decides if $w \in \accLanguage(\fufa)$ has to find a normalized decomposition $(u,v)$ of $w$ that is accepted by $\fufa$. Without any information on, e.g., the form of the period of $w$ or the length of its prefix, it is not clear how such a decomposition can be found algorithmically. 
Regarding intersection, situations in which $\llbracket \fufa_1 \rrbracket \cap \llbracket \fufa_2 \rrbracket = \emptyset$ but $\accLanguage(\fufa_1) \cap \accLanguage(\fufa_2) \neq \emptyset$ are challenging. An example for this is given in \Cref{fig:intersection-language-level-difficult}, where $\accLanguage(\fufa_1) = \accLanguage(\fufa_2) = \{a^\omega\}$, but $\llbracket \fufa_1 \rrbracket \cap \llbracket \fufa_2 \rrbracket = \emptyset$. Because acceptance of FUFA is defined w.r.t. decompositions $(u,v)$, it is not clear how to obtain an FUFA for $\accLanguage(\fufa_1) \cap \accLanguage(\fufa_2)$ from $\fufa_1$ and $\fufa_2$ in this case. 
We are not aware of any available methods that solve the corresponding problems for the easier case of (saturated) FDFA.

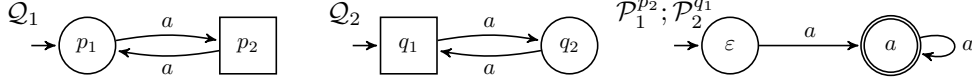
\begin{figure}[t]
	\centering
	\begin{tikzpicture}[->,>=stealth',shorten >=1pt,auto, semithick]
		\tikzstyle{every node} = [text = black, scale = 0.85]
		
		\node[state] (epsl) [] {$p_1$}; 
		\node[state, draw,rectangle] (al) [right of = epsl, node distance = 2.5cm] {$p_2$}; 
		\node (initl) [left of = epsl, node distance = 1cm] {}; 
		
		\path
		(initl) edge (epsl)
		(epsl) edge [bend left = 10] node {$a$} (al)
		(al) edge [bend left = 10] node {$a$} (epsl)
		;
		
		\node[state, draw,rectangle] (epsm) [right of = al, node distance = 2.5cm] {$q_1$}; 
		\node[state] (am) [right of = epsm, node distance = 2.5cm] {$q_2$}; 
		\node (initm) [left of = epsm, node distance = 1cm] {}; 
		
		\path
		(initm) edge (epsm)
		(epsm) edge [bend left = 10] node {$a$} (am)
		(am) edge [bend left = 10] node {$a$} (epsm)
		;
		
		\node[state] (epsr) [right of = am, node distance = 2.5cm] {$\varepsilon$}; 
		\node[state, accepting] (ar) [right of = epsr, node distance = 2.5cm] {$a$}; 
		\node (initr) [left of = epsr, node distance = 1cm] {}; 
		
		\path 
		(initr) edge (epsr)
		(epsr) edge node {$a$} (ar)
		(ar) edge [loop right] node {$a$}  (ar)
		;
		
		\node[above of = initl, node distance = 0.5cm, scale=1.2] {$\lfufa_1$}; 
		\node[above of = initm, node distance = 0.5cm, scale=1.2] {$\lfufa_2$}; 
		\node[above of = initr, node distance = 0.5cm, scale=1.2] {$\pfufa_1^{p_2} ; \pfufa_2^{q_1}$};

	\end{tikzpicture}
	\caption{Two FUFA for $\{a^\omega\}$ without a decomposition of $a^\omega$ accepted by both.}
	\label{fig:intersection-language-level-difficult}
\end{figure}

To conclude this section we discuss the challenges that arise in the search for the precise complexities of other decision problems and Boolean operations for FUFA.
\begin{remark}\label{rem:complexity-union-complementation-universality-fufa}
	For saturated FDFA the operations of union and complementation can be done efficiently \cite{FDFAAORL}. Since the precise complexity of these operations is, to the best of our knowledge, not yet known for UFA \cite{SLBSNDCUA,LBUACC}, we leave the analysis of the complexity of the corresponding operations on FUFA for future work. 
	Furthermore, deciding universality, language containment and language equality can be done efficiently for saturated FDFA \cite{FDFAAORL}. While these problems are in $\mathsf{P}$ for UFA \cite{ECPURERGFA,UAT}, their precise complexities in the case of UBA are long-standing open problems \cite{MCUA}. Because of the polynomial-time translation from UBA to saturated, decomposition- and power-unambiguous FUFA (\Cref{lem:uba-to-fufa}), the corresponding problems for (this type of) FUFA are at least as hard as for UBA.
\end{remark}
\section{Model Checking DTMCs against dFUFA-Specifications}\label{sec:dfufa-mc}
In this section we take the first step towards an efficient probabilistic model checking algorithm that uses saturated, decomposition- and power-unambiguous FUFA for the representation of $\omega$-regular properties. More precisely, we prove the existence of such an algorithm for the subclass of \emph{deterministic} FUFA (\emph{dFUFA}), which are FUFA with a deterministic leading automaton. The notion of dFUFA lies between FDFA and the families of nondeterministic finite automata (FNFA) defined in \cite{SPFA}, which require deterministic leading automata but allow progress NFA. Since every FDFA can be interpreted as a dFUFA, the set of saturated, decomposition- and power-unambiguous dFUFA characterizes the set of $\omega$-regular languages. Moreover, the determinism of the leading automaton $\lfufa$ of a dFUFA yields that $\lfufa(u)$ is a singleton for every $u \in \Sigma^*$. Therefore, dFUFA are always decomposition-unambiguous, and we usually do not explicitly mention decomposition-unambiguity in the context of dFUFA. 

\smallskip

Let $\dtmc$ be a DTMC, $\fufa = (\lfufa, \{\pfufa^q\}_{q \in \mstates})$ a saturated and power-unambiguous dFUFA, and $\automaton$ a DRA with $\accLanguage_\omega(\fufa) = E = \accLanguage(\automaton)$. W.l.o.g. we assume that the alphabet of all automata considered in this section equals the set of states $\stateSpace$ of $\dtmc$, i.e., that $\Sigma = \stateSpace$. This assumption is common in the literature \cite{MCMCAUBA,MCUA} and ensures that the product of $\dtmc$ with any progress UFA $\pfufa^q$ remains unambiguous if interpreted as an automaton with final states in $\stateSpace \times \afinal^{\pfufa^q}$. 

Our goal is to compute $\probMeasure^{\dtmc}(E)$ in time polynomial in the sizes of $\dtmc$ and $\fufa$. We start by observing that, because the leading automaton $\lfufa$ of $\fufa$ is deterministic, the residual languages of words reaching the same state in $\lfufa$ coincide. This is similar to the case of saturated FDFA considered in \Cref{lem:residual-languages-coincide} and can be proved analogously.  
\begin{corollary}\label{cor:residual-language-for-dfufa}
	Let $x, y \in \Sigma^*$ and let $\lfufa$ be the leading automaton of a dFUFA $\fufa$ for $E$. If $\lfufa(x) = \lfufa(y)$ then $\resLang_x = \resLang_y$.
\end{corollary}

Because $\lfufa$ is deterministic, the product $\dtmc \otimes \lfufa$ is again a DTMC, and so $\dtmc$ almost surely generates a path whose lifting to $\dtmc \otimes \lfufa$ reaches a BSCC  \cite{PoMC}. As in the FDFA-setting, the computation of $\probMeasure^\dtmc(E)$ boils down to the computation of reachability probabilities of specific BSCCs in $\dtmc \otimes \lfufa$. Recall that, in contrast to FDFA, a (d)FUFA requires the existence of progress automata $\pfufa^q$ only for states $q \in \mstates$. A given BSCC of $\dtmc \otimes \lfufa$ is, however, not guaranteed to contain any $\langle s, q \rangle$ with $q \in \mstates$. Since the acceptance condition of (d)FUFA is only defined for decompositions $(u,v) \in \Sigma^{*+}$ with $\mstates_u \neq \emptyset$, it is not immediately clear how to handle these BSCCs in a model checking algorithm. Fortunately, however, once a BSCC without a states $\langle s, q \rangle$ with $q \in \mstates$ is entered, the probability of $\dtmc$ to generate a path in $E$ equals $0$. Thus, these BSCCs do not play a role in the computation of $\probMeasure^{\dtmc}(E)$.\unskip 
\begin{restatable}{lemma}{LemBSCCWithoutAState}\label{lem:no-A-state-prob-0}
	Let $u \in \Paths^*(\dtmc)$ such that the lifting of $u$ to $\dtmc \otimes \lfufa$ reaches a BSCC $B_\lfufa$ that does not contain a $\langle s, q \rangle \in B_\lfufa$ with $q \in\mstates$. Then the probability of $\dtmc$ to generate a path that satisfies $E$ and has prefix $u$ equals $0$. 
\end{restatable}
\begin{proof}[Proof Sketch.]
	Let $\automaton$ be a DRA with $\accLanguage(\automaton) = E$ and assume that an accepting BSCC $B_\automaton$ is reachable in $\dtmc \otimes \automaton$ on some extension $uv$ of $u$. Because $B_\automaton$ is a BSCC there is a $z$ that extends $uv$ and closes a cycle that visits all states of $B_\automaton$. \Cref{lem:accepting-and-rejecting-bscc-dra-product} implies $uvz^\omega \in \accLanguage(\automaton)$, so also $uvz^\omega \in \accLanguage(\fufa)$. But then there must be a $\langle s', q' \rangle \in B_\lfufa$ with $q' \in \mstates$, a contradiction. 
	Hence, every BSCC reachable in $\dtmc \otimes \automaton$ after $u$ is rejecting, implying the claim.
\end{proof}

Let $\langle s, q \rangle$ be a state in a BSCC of $\dtmc \otimes \lfufa$ with $q \in \mstates$ and let again $I_{\langle s, q \rangle} = \dtmc_s \times \lfufa_q \times \pfufa^q$. Since $\pfufa^q$ is a UFA, $I_{\langle s, q \rangle}$ is not guaranteed to be a DTMC. Additionally, let $\automatonB(s,q)$ be a UBA with the same structure as $I_{\langle s, q \rangle}$ and $\afinal^{\automatonB(s,q)} = \{\langle s, q, f \rangle \mid f \in \afinal^{\pfufa^q}\}$. 
The UBA $\automatonB(s,q)$ can be used to decide if BSCCs in $\dtmc \otimes \automaton$ are accepting. 
\begin{restatable}{lemma}{LemCharacterizationGoodBSCCinDFUFA}\label{lem:characterization-bscc-dfufa}
	Let $u\in \Paths^*(\dtmc)$ such that the lifting of $u$ to $\dtmc \otimes \automaton$ reaches a BSCC $B_\automaton$, and the lifting of $u$ to  $\dtmc \otimes \lfdfa$ reaches a BSCC $B_\lfdfa$. Let $\langle s, q \rangle$ be the state reached on $u$ in $B_\lfufa$, and assume that $q \in \mstates$. Then $B_\automaton$ is accepting iff $\probMeasure^{\dtmc_s}(\accLanguage(\automatonB(s,q))) = 1$.
\end{restatable}
\begin{proof}[Proof Sketch.]
	If $\probMeasure^{\dtmc_s}(\accLanguage(\automatonB(s,q))) < 1$ there is $v \in \Paths^*(\dtmc_s)$ such that (1) $(u,v)$ is $q$-normalized, (2) $v \notin \accLanguage(\pfufa^q)$ and (3) $v$ closes a cycle in $B_\automaton$ that visits all states of $B_\automaton$. Combining (1)+(2) yields $uv^\omega \notin E = \accLanguage(\automaton)$, so (3) and  \Cref{lem:accepting-and-rejecting-bscc-dra-product} imply that $B_\automaton$ is rejecting.

	Now let $\probMeasure^{\dtmc_s}(\accLanguage(\automatonB(s,q))) = 1$. There is $x \in \Paths^*(\dtmc_s)$ closing a loop on $\langle s, q \rangle$ and such that the run of $\dtmc \otimes \automaton$ on $x$ starting in any $\langle s, a \rangle \in B_\automaton$ visits all states of $B_\automaton$. Because $\probMeasure^{\dtmc_s}(\accLanguage(\automatonB(s,q))) = 1$ there is a $y$ such that a $\langle s, q, f \rangle$ with $f \in \afinal^{\pfufa^q}$ is reachable in $I_{\langle s, q \rangle}$ upon $xy$. Thus, $(u,xy) \in \llbracket \fufa \rrbracket$ and $u(xy)^\omega \in E = \accLanguage(\automaton)$, so $B_\automaton$ is accepting by \Cref{lem:accepting-and-rejecting-bscc-dra-product}.
\end{proof}

Similar to \Cref{lem:single-inner-product-sufficient}, the specific choice of $\langle s, q \rangle \in B_\lfufa$ with $q \in \mstates$ does not matter when checking if $\probMeasure^{\dtmc_s}(\accLanguage(\automatonB(s,q))) = 1$, as this is either the case for all such $\langle s, q \rangle$ or none of them. \unskip
\begin{restatable}{lemma}{LemAlmostSurelyForAllOrNoneFUFA}\label{lem:probability-infinitely-many-accepting-same-in-bscc}
	Let $B_\lfufa$ be a BSCC of $\dtmc \otimes \lfufa$ and let $\langle s, q \rangle, \langle s', q' \rangle \in B_\lfufa$ with $q,q' \in \mstates$. If $\probMeasure^{\dtmc_s}(\accLanguage(\automatonB(s,q))) = 1$ then $\probMeasure^{\dtmc_{s'}}(\accLanguage(\automatonB(s',q'))) = 1$. 
\end{restatable}
\begin{proof}[Proof Sketch.]
	Let $u$ be such that its lifting to $\dtmc \otimes \lfufa$ reaches $\langle s', q' \rangle$, while its lifting to $\dtmc \otimes \automaton$ reaches a BSCC $B_\automaton$. Let $v$ be an extension of $u$ such that, upon $uv$, $\dtmc \otimes \lfufa$ reaches $\langle s, q \rangle$. Since $\probMeasure^{\dtmc_s}(\accLanguage(\automatonB(s,q))) = 1$, $B_\automaton$ is accepting by \Cref{lem:characterization-bscc-dfufa}, and because $B_\automaton$ was already reached after $u$, the same lemma yields that $\probMeasure^{\dtmc_{s'}}(\accLanguage(\automatonB(s',q'))) = 1$.
\end{proof}

A BSCC $B$ of $\dtmc \otimes \lfufa$ is \emph{good} if it contains a state $\langle s, q \rangle$ with $q \in \mstates$ and such that $\probMeasure^{\dtmc_s}(\accLanguage(\automatonB(s,q))) = 1$. Otherwise, $B$ is \emph{bad}. 
The classification of BSCCs of $\dtmc \otimes \lfufa$ into good and bad BSCCs is the key step of the algorithm. This is due to the following lemma. 
\begin{restatable}{lemma}{LemZeroOnePropertyFUFA}\label{lem:zero-one-property-dfufa}
	The paths of $\dtmc$ whose lifting to $\dtmc \otimes \lfufa$ reach a good BSCC almost surely satisfy $E$, and those whose lifting to $\dtmc \otimes \lfufa$ reach a bad BSCC almost surely violate $E$. 
\end{restatable}
\begin{proof}[Proof Sketch.]
	Let the lifting of $u \in \Paths^*(\dtmc)$ to $\dtmc \otimes \lfufa$ reach a BSCC $B_{\lfufa}$. If $B_\lfufa$ is good then it follows from \Cref{lem:characterization-bscc-dfufa} that all BSCCs reachable in $\dtmc \otimes \automaton$ after $u$ are accepting. Otherwise, if $B_\lfufa$ is bad, then the combination of \Cref{lem:no-A-state-prob-0,lem:characterization-bscc-dfufa} implies that all BSCCs reachable in $\dtmc \otimes \automaton$ after $u$ are rejecting.
\end{proof}

Since $\dtmc \otimes \lfufa$ is a DTMC, almost surely $\dtmc$ generates a path whose lifting to $\dtmc \otimes \lfufa$ enters a BSCC $B_\lfufa$. It follows from \Cref{lem:zero-one-property-dfufa} that if $B_\lfufa$ is good then almost surely a word in $E$ will be generated, while if $B_\lfufa$ is bad this will almost surely not be the case. Hence,  $\probMeasure^{\dtmc}(E)$ equals the probability of reaching a good BSCC in $\dtmc \otimes \lfufa$. This is formalized in the following proposition, where $\mathit{goodBSCC}$ contains all good BSCCs of $\dtmc \otimes \lfufa$. 
\begin{restatable}{proposition}{PropCorrectnessDFUFAMC}\label{prop:correctness-dfufa-mc}
	$\probMeasure^{\dtmc}(E) = \probMeasure^{\dtmc \otimes \lfufa}(\lozenge \mathit{goodBSCC})$. 
\end{restatable}

The probability of a DTMC $\dtmc$ to generate a path satisfying an $\omega$-regular property $E$ represented by a saturated and power-unambiguous dFUFA can thus be computed as follows: \unskip
\vspace{0.1cm}
\begin{enumerate}
	\item Compute the set of BSCCs of $\dtmc \otimes \lfufa$ that contain a state $\langle s, q \rangle$ with $q \in \mstates$. 
	\item For each such BSCC $B$ decide if $B$ is good by picking any $\langle s, q \rangle \in B$ with $q \in \mstates$ and checking if $\probMeasure^{\dtmc_s}(\accLanguage(\automatonB(s,q))) = 1$. If this is the case, add $B$ to $\mathit{goodBSCC}$. 
	\item Compute $\probMeasure^{\dtmc}(E) = \probMeasure^{\dtmc\otimes \lfufa}(\lozenge \mathit{goodBSCC})$. 
\end{enumerate}
\vspace{0.1cm}
Correctness follows from \Cref{prop:correctness-dfufa-mc} and the fact that either $\probMeasure^{\dtmc_s}(\accLanguage(\automatonB(s,q))) = 1$ for all $\langle s, q \rangle$ with $q \in \mstates$ in a BSCC of $\dtmc \otimes \lfufa$, or none of them (\Cref{lem:probability-infinitely-many-accepting-same-in-bscc}). 
Furthermore, computing the BSCCs of $\dtmc\otimes \lfufa$ and the value $\probMeasure^{\dtmc\otimes \lfufa}(\lozenge \mathit{goodBSCC})$ is possible in time polynomial in the sizes of $\dtmc$ and $\lfufa$ \cite{PoMC}, and checking if $\probMeasure^{\dtmc_s}(\accLanguage(\automatonB(s,q))) = 1$ can be done in time polynomial in the sizes of $\dtmc$ and $\automatonB(s,q)$ \cite{MCUA}. Since the latter is necessary for at most $(\vert \dtmc \vert \cdot \vert \lfufa \vert)$-many BSCCs, the overall runtime of the algorithm is polynomial in $\vert \dtmc \vert$ and $\vert \fufa \vert$.

This yields an extension of \Cref{thm:poly-time-mc-for-fdfa} to saturated and power-unambiguous dFUFA. 
\begin{restatable}{theorem}{ThmDFUFAMC}\label{thm:dfufa-mc}
	Let $\dtmc$ be a DTMC and $\fufa$ a saturated and power-unambiguous dFUFA with $\accLanguage_\omega(\fufa) = E$. Then $\probMeasure^{\dtmc}(E)$ can be computed in time polynomial in the sizes of $\dtmc$ and $\fufa$. 
\end{restatable}

We conclude with a discussion on why replacing the UBA-methods used in the algorithm with approaches for model checking DTMCs against UFA-specifications is challenging. 

The probabilistic model checking problem for a DTMC $\dtmc$ against a UFA $\automaton$ asks to compute $\probMeasure^{\dtmc}(\accLanguage(\automaton) \cdot \Sigma^\omega)$. This probability can be computed in polynomial time by solving a specific system of linear equations \cite{MCMCAUBA}, provided that (1) the language of the UFA $\automaton$ is prefix-free or, more generally, that (2) for all $x, xy \in \accLanguage(\automaton)$ the accepting run of $\automaton$ on $x$ is a prefix of the accepting run of $\automaton$ on $xy$. For more details, see Appendix \ref{app:cex-ufa-mc}. However, neither (1) nor (2) are guaranteed by the translation from UBA to FUFA given in \Cref{sec:fufa}, and so it is not clear how to construct an FUFA for a given $\omega$-regular specification whose progress automata are suitable for the application of available UFA-based probabilistic model checking algorithms. To the best of our knowledge, it is not known whether model checking DTMCs against general UFA-specifications is possible in polynomial time. 
Moreover, the characterization of BSCCs of $\dtmc \otimes \lfufa$ used in the model checking procedure proposed in this section depends on the almost sure generation of certain \emph{infinite} paths of $\dtmc$. Therefore, even if polynomial-time model checking algorithms for DTMCs against general UFA-specifications would be at our disposal, adjusting the algorithm such that it uses UFA- instead of UBA-techniques is non-trivial.

Lastly, we observe that the probabilistic model checking problem for DTMCs against FUFA-specifications is at least as hard as the corresponding problem for UFA-specifications.
To see this, let $\automaton$ be a UFA over $\Sigma$. We can construct an FUFA $\fufa$ that accepts $\accLanguage(\automaton) \cdot \Sigma^\omega$ by choosing $\lfufa$ to be like $\automaton$, but with a fresh state $f$ that is reachable on the full alphabet from every final state of $\automaton$ and closes a self-loop on $\Sigma$, setting $\mstates = \{f\}$, and choosing $\pfufa^f$ such that $\accLanguage(\pfufa^f)= \Sigma^+$. Hence, a probabilistic model checking algorithm for DTMCs against FUFA-specifications would also yield a procedure for the verification of DTMCs against general UFA-specifications. 
\section{Conclusion}\label{sec:conclusion}
This paper presented the first polynomial-time probabilistic model checking algorithm for DTMCs against $\omega$-regular specifications represented as saturated FDFA. 
Furthermore, FUFA were introduced as a variant of FDFA that allows nondeterministic leading automata and unambiguous progress automata. Saturated, decomposition- and power-unambiguous FUFA characterize the $\omega$-regular languages and can be exponentially more succinct than saturated FDFA and UBA. Additionally, there is a single-exponential translation from LTL to FUFA, which is in contrast to a double-exponential lower bound for the translation from LTL to FDFA. An extension of the probabilistic model checking algorithm for DTMCs against FDFA to dFUFA, i.e., to FUFA with a deterministic leading automaton, was also provided.

Finding a polynomial-time algorithm that solves the probabilistic model checking problem for DTMCs against general FUFA-specifications is challenging. In particular, the fact that the leading automaton of an FUFA can be ambiguous is both a blessing and a curse. While this ambiguity is key in the proof of \Cref{thm:fufa-2-uba} to obtain an exponential lower bound for the translation from FUFA to UBA, it also increases the complexity of algorithmic questions like probabilistic model checking for FUFA. In particular, the known PSPACE-hardness results for model checking DTMCs against NBA-specifications \cite{AVPCFSP} suggest that there might not be an efficient probabilistic model checking algorithm against FUFA-specifications at all. 

An analysis of the complexity of the probabilistic model checking problem against general FUFA-specifications is important future work. Moreover, the exact cost of translating NBA into FUFA and the question if dFUFA can be exponentially more succinct than saturated FDFA are open. 
An experimental evaluation of the proposed algorithms is also in order. 

\bibliography{lit}

@book{PoMC,
  author       = {Christel Baier and Joost{-}Pieter Katoen},
  title        = {Principles of Model Checking},
  publisher    = {{MIT} Press},
  year         = {2008},
  isbn         = {978-0-262-02649-9},
}

@InProceedings{SPFA,
	author =	{Bohn, Le\'{o}n and Li, Yong and L\"{o}ding, Christof and Schewe, Sven},
	title =	{Saturation Problems for Families of Automata},
	booktitle =	{52nd International Colloquium on Automata, Languages, and Programming (ICALP 2025)},
	pages =	{146:1--146:19},
	series =	{Leibniz International Proceedings in Informatics (LIPIcs)},
	year =	{2025},
	volume =	{334},
	editor =	{Censor-Hillel, Keren and Grandoni, Fabrizio and Ouaknine, Jo\"{e}l and Puppis, Gabriele},
	publisher =	{Schloss Dagstuhl -- Leibniz-Zentrum f{\"u}r Informatik},
	address =	{Dagstuhl, Germany},
	doi =		{10.4230/LIPIcs.ICALP.2025.146}
}

@article{FDFAAORL,
	title      = {Families of {DFA}s as Acceptors of $\omega$-Regular Languages},
	author     = {Dana Angluin and Udi Boker and Dana Fisman},
	doi        = {10.23638/LMCS-14(1:15)2018},
	journal    = {Logical Methods in Computer Science},
	issn       = {1860-5974},
	volume     = {Volume 14, Issue 1},
	eid        = 15,
	year       = {2018},
	month      = {Feb}
}

@inproceedings{UPWROL,
	address = {Berlin, Heidelberg},
	title = {Ultimately periodic words of rational $\omega$-languages},
	doi = {10.1007/3-540-58027-1_27},
	booktitle = {Mathematical {Foundations} of {Programming} {Semantics} (MFPS 1993)},
	publisher = {Springer},
	author = {Calbrix, Hugues and Nivat, Maurice and Podelski, Andreas},
	editor = {Brookes, Stephen and Main, Michael and Melton, Austin and Mislove, Michael and Schmidt, David},
	year = {1994},
	pages = {554--566},
	series = {Lecture Notes in Computer Science (LNCS)},
	volume = {802}
}

@article{LROL,
	series = {Algorithmic {Learning} {Theory}},
	title = {Learning regular omega languages},
	volume = {650},
	doi = {10.1016/j.tcs.2016.07.031},
	journal = {Theoretical Computer Science},
	author = {Angluin, Dana and Fisman, Dana},
	month = oct,
	year = {2016},
	pages = {57--72},
}

@InProceedings{NLABAFDFACT,
	author="Li, Yong
	and Chen, Yu-Fang
	and Zhang, Lijun
	and Liu, Depeng",
	editor="Legay, Axel
	and Margaria, Tiziana",
	title="A Novel Learning Algorithm for {B\"uchi} Automata based on Family of {DFA}s and Classification Trees",
	booktitle="Tools and Algorithms for the Construction and Analysis of Systems (TACAS 2017)",
	year="2017",
	publisher="Springer",
	address="Berlin, Heidelberg",
	pages="208--226",
	doi="10.1007/978-3-662-54577-5_12", 
	series="Lecture Notes in Computer Science (LNCS)", 
	volume="10205"
}

@InProceedings{NFFARLORL,
	author="Li, Yong
	and Schewe, Sven
	and Tang, Qiyi",
	editor="Andr{\'e}, {\'E}tienne
	and Sun, Jun",
	title="A Novel Family of Finite Automata for Recognizing and Learning $\omega$-Regular Languages",
	booktitle="Automated Technology for Verification and Analysis",
	year="2023",
	publisher="Springer Nature Switzerland",
	address="Cham",
	pages="53--73",
	doi="10.1007/978-3-031-45329-8_3", 
	series="Lecture Notes in Computer Science (LNCS)", 
	volume="14215"
}

@book{FMC,
	title={Finite {Markov} Chains},
	author={Kemeny, John G. and Snell, J. Laurie},
	series={Undergraduate Texts in Mathematics (UTM)},
	year={1976},
	publisher={Springer New York}
}

@InProceedings{LBUACC,
	author =	{G\"{o}\"{o}s, Mika and Kiefer, Stefan and Yuan, Weiqiang},
	title =	{Lower Bounds for Unambiguous Automata via Communication Complexity},
	booktitle =	{49th International Colloquium on Automata, Languages, and Programming (ICALP 2022)},
	pages =	{126:1--126:13},
	series =	{Leibniz International Proceedings in Informatics (LIPIcs)},
	year =	{2022},
	volume =	{229},
	editor =	{Boja\'{n}czyk, Miko{\l}aj and Merelli, Emanuela and Woodruff, David P.},
	publisher =	{Schloss Dagstuhl -- Leibniz-Zentrum f{\"u}r Informatik},
	address =	{Dagstuhl, Germany},
	doi =		{10.4230/LIPIcs.ICALP.2022.126}
}

@incollection{DMRSOA,
	title = {Symposium on Decision Problems: On a Decision Method in Restricted Second Order Arithmetic},
	editor = {Ernest Nagel and Patrick Suppes and Alfred Tarski},
	series = {Studies in Logic and the Foundations of Mathematics},
	publisher = {Elsevier},
	volume = {44},
	pages = {1-11},
	year = {1966},
	booktitle = {Logic, Methodology and Philosophy of Science},
	doi = {https://doi.org/10.1016/S0049-237X(09)70564-6},
	author = {{B}üchi, {J}. {R}ichard},
}

@InProceedings{SLBSNDCUA,
	author =	{Raskin, Mikhail},
	title =	{A Superpolynomial Lower Bound for the Size of Non-Deterministic Complement of an Unambiguous Automaton},
	booktitle =	{45th International Colloquium on Automata, Languages, and Programming (ICALP 2018)},
	pages =	{138:1--138:11},
	series =	{Leibniz International Proceedings in Informatics (LIPIcs)},
	year =	{2018},
	volume =	{107},
	editor =	{Chatzigiannakis, Ioannis and Kaklamanis, Christos and Marx, D\'{a}niel and Sannella, Donald},
	publisher =	{Schloss Dagstuhl -- Leibniz-Zentrum f{\"u}r Informatik},
	address =	{Dagstuhl, Germany},
	doi =		{10.4230/LIPIcs.ICALP.2018.138},
}

@incollection{UAT,
	address = {Cham},
	title = {Unambiguity in Automata Theory},
	booktitle = {Descriptional {Complexity} of {Formal} {Systems} (DCFS 2015)},
	publisher = {Springer International Publishing},
	author = {Colcombet, Thomas},
	editor = {Shallit, Jeffrey and Okhotin, Alexander},
	year = {2015},
	doi = {10.1007/978-3-319-19225-3_1},
	series = {Lecture Notes in Computer Science (LNCS)},
	volume = {9118},
	pages = {3--18},
}

@article{ECPURERGFA,
	author = {Stearns, R. E. and Hunt III, H. B.},
	title = {On the Equivalence and Containment Problems for Unambiguous Regular Expressions, Regular Grammars and Finite Automata},
	journal = {SIAM Journal on Computing},
	volume = {14},
	number = {3},
	pages = {598-611},
	year = {1985},
	doi = {10.1137/0214044}
}

@article{MCUA,
	title = {Markov chains and unambiguous automata},
	volume = {136},
	doi = {10.1016/j.jcss.2023.03.005},
	journal = {Journal of Computer and System Sciences},
	author = {Baier, Christel and Kiefer, Stefan and Klein, Joachim and Müller, David and Worrell, James},
	month = sep,
	year = {2023},
	pages = {113--134},
}

@unpublished{BALi,
	title        = {Translating {LTL} to {FDFA} and {FDFA} Model Checking},
	author       = {Li, Jianlin},
	note         = {Bachelor's thesis. Nanjing University of Aeronautics and Astronautics},
	year         = {2018}
}

@misc{DFpersCom,
	author       = {Fisman, Dana},
	title        = {Personal communication},
	year         = {2024}
}

@article{Alternation,
	author = {Chandra, Ashok K. and Kozen, Dexter C. and Stockmeyer, Larry J.},
	title = {Alternation},
	year = {1981},
	issue_date = {Jan. 1981},
	publisher = {Association for Computing Machinery},
	address = {New York, NY, USA},
	volume = {28},
	number = {1},
	doi = {10.1145/322234.322243},
	journal = {J. ACM},
	month = jan,
	pages = {114–133},
	numpages = {20}
}

@article{UTLTLOA,
	author = {Esparza, Javier and K\v{r}et\'{\i}nsk\'{y}, Jan and Sickert, Salomon},
	title = {A Unified Translation of Linear Temporal Logic to $\omega$-Automata},
	year = {2020},
	issue_date = {December 2020},
	publisher = {Association for Computing Machinery},
	address = {New York, NY, USA},
	volume = {67},
	number = {6},
	doi = {10.1145/3417995},
	journal = {J. ACM},
	month = oct,
	articleno = {33},
	numpages = {61}
}

@InProceedings{LTLLDBADPA,
	author="Esparza, Javier
	and K{\v{r}}et{\'i}nsk{\'y}, Jan
	and Raskin, Jean-Fran{\c{c}}ois
	and Sickert, Salomon",
	editor="Legay, Axel
	and Margaria, Tiziana",
	title="From {LTL} and Limit-Deterministic {B{\"u}chi} Automata to Deterministic Parity Automata",
	booktitle="Tools and Algorithms for the Construction and Analysis of Systems (TACAS 2017)",
	year="2017",
	publisher="Springer",
	address="Berlin, Heidelberg",
	pages="426--442",
	series="Lecture Notes in Computer Science (LNCS)", 
	volume="10205", 
	doi="10.1007/978-3-662-54577-5_25"
}

@inproceedings{TLP,
	author = {Pnueli, Amir},
	title = {The temporal logic of programs},
	year = {1977},
	publisher = {IEEE Computer Society},
	address = {USA},
	doi = {10.1109/SFCS.1977.32},
	booktitle = {Proceedings of the 18th Annual Symposium on Foundations of Computer Science (SFCS 1977)},
	pages = {46–57},
	numpages = {12}
}

@InProceedings{CRBA,
	author="Li, Yong
	and Tsay, Yih-Kuen
	and Turrini, Andrea
	and Vardi, Moshe Y.
	and Zhang, Lijun",
	editor="Huisman, Marieke
	and P{\u{a}}s{\u{a}}reanu, Corina
	and Zhan, Naijun",
	title="Congruence Relations for {B{\"u}chi} Automata",
	booktitle="Formal Methods (FM 2021)",
	year="2021",
	publisher="Springer International Publishing",
	address="Cham",
	pages="465--482",
	series="Lecture Notes in Computer Science (LNCS)", 
	volume="13047", 
	doi="10.1007/978-3-030-90870-6_25"
}

@article{DCNFADA,
	title = {Descriptional Complexity of {NFA} of different ambiguity},
	volume = {16},
	doi = {10.1142/S0129054105003418},
	journal = {International Journal of Foundations of Computer Science},
	author = {Leung, Hing},
	month = jan,
	year = {2004},
	pages = {975--984}
}

@phdthesis{SDCFRFL,
	title  = {Succinctness of Descriptions of Context-Free, Regular, and Finite Languages},
	author = {Schmidt, Erik Meineche},
	year   = {1978},
	school = {University of Aarhus}
}

@article{CPV,
	title = {The complexity of probabilistic verification},
	volume = {42},
	doi = {10.1145/210332.210339},
	number = {4},
	journal = {J. ACM},
	author = {Courcoubetis, Costas and Yannakakis, Mihalis},
	month = jul,
	year = {1995},
	pages = {857--907},
}

@InProceedings{CSRCP,
	author="Thomas, Wolfgang",
	editor="de Alfaro, Luca",
	title="Facets of Synthesis: Revisiting {Church}'s Problem",
	booktitle="Foundations of Software Science and Computational Structures (FoSSaCS 2009)",
	year="2009",
	publisher="Springer",
	address="Berlin, Heidelberg",
	pages="1--14",
	series="Lecture Notes in Computer Science (LNCS)", 
	volume="5504", 
	doi="10.1007/978-3-642-00596-1_1"
}

@inproceedings{SRM,
	author = {Pnueli, A. and Rosner, R.},
	title = {On the synthesis of a reactive module},
	year = {1989},
	publisher = {Association for Computing Machinery},
	address = {New York, NY, USA},
	doi = {10.1145/75277.75293},
	booktitle = {Proceedings of the 16th ACM SIGPLAN-SIGACT Symposium on Principles of Programming Languages (POPL 1989)},
	pages = {179–190},
	numpages = {12},
	location = {Austin, Texas, USA}
}

@InProceedings{ATAAPV,
	author = 	 {Moshe Y. Vardi and Pierre Wolper},
	title = 	 {An Automata-Theoretic Approach to Automatic Program Verification (Preliminary Report)},
	booktitle =  {Proceedings of the First Annual IEEE Symposium on Logic in Computer Science (LICS 1986)},
	year =	 {1986},
	month =	 {June}, 
	pages =      {332--344 },
	location =   {Cambridge, MA, USA}, 
	publisher =	 {IEEE Computer Society Press}
}

@article{OA,
	author       = {Thomas Wilke},
	title        = {$\omega$-{A}utomata},
	journal      = {CoRR},
	volume       = {abs/1609.03062},
	year         = {2016},
	doi =  	{10.48550/arXiv.1609.03062}
}

@book{ALIG,
	title     = {Automata, Logics, and Infinite Games},
	subtitle  = {A Guide to Current Research},
	editor    = {Gr{\"a}del, Erich and Thomas, Wolfgang and Wilke, Thomas},
	series    = {Lecture Notes in Computer Science (LNCS)},
	volume    = {2500},
	publisher = {Springer},
	address   = {Berlin, Heidelberg},
	year      = {2002},
	doi       = {10.1007/3-540-36387-4}
}

@article{FLTBT,
	author = {Kupferman, Orna and Vardi, Moshe Y.},
	title = {From linear time to branching time},
	year = {2005},
	issue_date = {April 2005},
	publisher = {Association for Computing Machinery},
	address = {New York, NY, USA},
	volume = {6},
	number = {2},
	doi = {10.1145/1055686.1055689},
	journal = {ACM Trans. Comput. Logic},
	month = apr,
	pages = {273–294},
	numpages = {22},
}

@INPROCEEDINGS{AVPCFSP, 
	author={Vardi, Moshe Y.}, 
	booktitle={26th Annual Symposium on Foundations of Computer Science (FOCS 1985)}, 
	title={Automatic verification of probabilistic concurrent finite state programs},
	year={1985}, 
	volume={}, 
	number={}, 
	pages={327-338},
	doi={10.1109/SFCS.1985.12}}

@article{MCMCAUBA,
	author       = {Michael Benedikt and
	Rastislav Lenhardt and
	James Worrell},
	title        = {Model Checking {Markov} Chains Against Unambiguous {Büchi} Automata},
	journal      = {CoRR},
	volume       = {abs/1405.4560},
	year         = {2014}, 
	doi = {10.48550/arXiv.1405.4560}
}

@article{SEAFAPAFA,
	title = {Separating Exponentially Ambiguous Finite Automata from Polynomially Ambiguous Finite Automata},
	volume = {27},
	doi = {10.1137/S0097539793252092},
	number = {4},
	journal = {SIAM Journal on Computing},
	author = {Leung, Hing},
	month = aug,
	year = {1998},
	pages = {1073--1082},
}

@inproceedings{FABA,
	address = {Cham},
	title = {On Finitely Ambiguous {Büchi} Automata},
	doi = {10.1007/978-3-319-98654-8_41},
	booktitle = {Developments in {Language} {Theory} (DLT 2018)},
	publisher = {Springer International Publishing},
	author = {Löding, Christof and Pirogov, Anton},
	editor = {Hoshi, Mizuho and Seki, Shinnosuke},
	year = {2018},
	pages = {503--515},
}

@article{CMDAIW,
	title={Complementation is more difficult with automata on infinite words},
	author={Michel, Max},
	journal={CNET, Paris},
	volume={15},
	year={1988}
}

@inproceedings{OBTOA,
	address = {Berlin, Heidelberg},
	title = {Optimal Bounds for Transformations of $\omega$-Automata},
	doi = {10.1007/3-540-46691-6_8},
	booktitle = {Foundations of {Software} {Technology} and {Theoretical} {Computer} {Science} (FSTTCS 1999)},
	publisher = {Springer},
	author = {Löding, Christof},
	editor = {Rangan, C. Pandu and Raman, V. and Ramanujam, R.},
	year = {1999},
	pages = {97--109}
}

@INPROCEEDINGS{NBSADPA,
	author={Piterman, N.},
	booktitle={21st Annual IEEE Symposium on Logic in Computer Science (LICS 2006)}, 
	title={From Nondeterministic {Büchi} and {Streett} Automata to Deterministic Parity Automata}, 
	year={2006},
	pages={255-264},
	doi={10.1109/LICS.2006.28}}

@misc{CLTLFE,
	title = {Characterizing {LTL} {Formulas} by {Examples}},
	doi = {10.48550/arXiv.2604.22097},
	publisher = {arXiv},
	author = {Cate, Balder ten and Fisman, Dana and Ohayon, Roi and Sestic, Patrik},
	month = apr,
	year = {2026},
}

\newpage 
\appendix
\section{\texorpdfstring{Proofs of \Cref{sec:fdfa-mc}}{Proofs of Section 3}}\label{app:proofs-fdfa-mc}
\LemAcceptingAndRejectingBSCCinDRAProduct*
\begin{proof}
	Let $B$ be accepting. Then there is a Rabin pair $(T_i, R_i)$ of $\automaton$ such that $B \vert_{\automaton} \cap T_i = \emptyset$ and $B \vert_{\automaton} \cap R_i \neq \emptyset$. Now let $\pi$ be as in the claim. By construction, the projection of $u \pi \in \Paths(\dtmc \otimes \automaton)$ onto its second component yields the unique run of $\automaton$ on the word $w = \trace(u \pi)$. This run visits infinitely many states in $R_i$ and only finitely many states in $T_i$ (since, after reading $u$, no state with a second component in $T_i$ will be encountered anymore along $\pi$), so $w \in \accLanguage(\automaton)$ and hence $w = \trace(u \pi) \in E$. This proves the first claim.
	
	\smallskip
	To show the second claim let $B$ be rejecting. Then for every Rabin pair $(T_i, R_i)$ we have $B \vert_\automaton \cap T_i \neq \emptyset$ or $B \vert_\automaton \cap R_i = \emptyset$. Thus, for any $\pi$ as in the claim and all $i$ the projection $r$ of $u\pi$ onto its second component does not satisfy $\mathrm{inf}(r) \cap T_i = \emptyset$ and $\mathrm{inf}(r) \cap R_i \neq \emptyset$. It follows that $\automaton$ does not accept $w = \trace(u\pi)$, and so $w \notin E$. 
\end{proof}

\LemSingleInnerProduct*
\begin{proof}
	Assume that $I_{\langle s, q \rangle}$ has a rejecting BSCC $B$, and that $I_{\langle s', q' \rangle}$ does not. 
	
	Let $u$ be a finite path in $\dtmc$ whose lifting to $\dtmc \otimes \lfdfa$ ends in $\langle s', q' \rangle$. Since $\langle s, q \rangle, \langle s', q' \rangle \in C$, there is a finite path $v \in \Paths^*(s')$ whose lifting to $\dtmc \otimes \lfdfa$ satisfies $\langle s' , q' \rangle \overset{v}{\longrightarrow} \langle s, q \rangle$.
	Moreover, let $w \in \Paths^*(s)$ such that the lifting of $w$ to $I_{\langle s, q \rangle}$ ends in some $\langle s'', q'', r \rangle \in B$. Lastly, let $\langle s'', q'', p \rangle$ be the state of $I_{\langle s', q' \rangle}$ reached after following (the lifting of) $vw$. 
	
	Because $I_{\langle s', q' \rangle}$ does not contain a rejecting BSCC, there must be an accepting state $\langle s', q', f \rangle$ reachable from $\langle s'', q'', p \rangle$ via some finite $x$, i.e., $\langle s'', q'', p \rangle \overset{x}{\longrightarrow} \langle s', q', f \rangle$. In particular, there is a path in $\pfdfa^{q'}$ from $p$ to $f$ on $\trace(x)$. 
	
	Now let $y = \trace(u)$ and $z = \trace(vwx)$. 
	Since $\lfdfa(y) = q'$ and 
	\begin{center}
		$\delta^{\lfdfa}(q', z) = \delta^{\lfdfa}(q', \trace(vwx)) = \delta^{\lfdfa}(q, \trace(wx)) = \delta^{\lfdfa}(q'', \trace(x)) = q'$,
	\end{center} 
	the decomposition $(y,z)$ is normalized w.r.t. $\lfdfa$. Furthermore, $z \in \accLanguage(\pfdfa^{q'})$ as 
	\begin{center}
		$\pfdfa^{q'}(z) = \pfdfa^{q'}(\trace(vwx)) = \delta^{\pfdfa^{q'}}(p, \trace(x)) = f \in \afinal^{\pfdfa^{q'}}. $
	\end{center}
	Thus, $(y,z) \in \llbracket \fdfa \rrbracket$. 
	
	On the other hand, let $y' = \trace(uv)$ and $z' = \trace(wxv)$, i.e., $(y', z')$ is obtained from shifting $(y,z)$ by $v$ ``to the left''. Then $\lfdfa(y') = q$ and $(y', z')$ is normalized w.r.t. $\lfdfa$. However
	\begin{center}
		$\pfdfa^{q}(z') = \pfdfa^{q}(\trace(wxv)) = \pfdfa^{q}(q'', \trace(xv)) \notin \afinal^{\pfdfa^{q}}$
	\end{center}
	since after following $w$ in $I_{\langle s, q \rangle}$ the rejecting BSCC $B$ is entered, and so the state reached after following $wxv$ in $I_{\langle s, q \rangle}$ must be of the form $\langle s, q, r \rangle$ for some $r \notin \afinal^{\pfdfa^q}$. 
	Thus, $z' \notin \accLanguage(\pfdfa^q)$ and so $(y', z') \notin \llbracket \fdfa \rrbracket$.
	As $yz^\omega = y'(z')^\omega$, and because both $(y,z)$ and $(y',z')$ are normalized, this contradicts the saturation of $\fdfa$. 
\end{proof}

\LemGraphAnalysisForGoodBad*
\begin{proof}
	Let $\langle s, q \rangle \in B$ and $I_B = I_{\langle s, q \rangle}$. We start by proving that 
	\begin{align}
		I_B \text{ has a rejecting BSCC \qquad iff \qquad } \probMeasure^{I_B}(\lozenge \mathit{accept}) < 1. \label{eq:linear-time-computation-good-bscc}
	\end{align}
	For the direction from left to right, let $I_B$ have a rejecting BSCC $C$. Then there is a $u \in \Paths^*(I_B)$ with $\probMeasure^{I_B}(u)  > 0$ that reaches $C$ while not visiting any state labeled with $\mathit{accept}$ (as only accepting states in BSCCs of $I_B$ can have label $\mathit{accept}$). Since $C$ is rejecting, no state in $C$ is of the form $\langle s, q, f \rangle$ for $f \in \afinal^{\pfdfa^q}$, and so in particular no state labeled with $\mathit{accept}$ is reachable after entering $C$. Hence, for every $\pi \in Cyl(u)$ we have that $\pi$ violates $\lozenge \mathit{accept}$. Thus, $\probMeasure^{I_B}(\lozenge \mathit{accept}) \leq 1 - \probMeasure^{I_B}(u) < 1$. 
	
	Regarding the reverse direction, let $\probMeasure^{I_B}(\lozenge \mathit{accept}) < 1$ and assume towards a contradiction that $I_B$ does not have a rejecting BSCC. Then every BSCC of $I_B$ must contain at least one state $\langle s, q, f \rangle$ for some $f \in \afinal^{\pfdfa^q}$. Because $I_B$ is a DTMC, almost surely a path is generated that reaches a BSCC of $I_B$ and visits all its states infinitely often \cite{PoMC}. But then $ \probMeasure^{I_B}(\lozenge \mathit{accept}) = 1$, a contradiction. 
	
	\smallskip 
	
	Following \Cref{eq:linear-time-computation-good-bscc,lem:single-inner-product-sufficient}, $B$ is good iff $\probMeasure^{I_B}(\lozenge \mathit{accept}) = 1$. As almost sure reachability in $I_B$ can be decided in time linear in the size of $I_B$ \cite{PoMC}, the claim follows.
\end{proof}

\LemAccBSCCiffGoodBSCCFDFA*
\begin{proof}
	The proof follows by combining \Cref{lem:acc-bscc-implies-good-bscc,lem:good-bscc-implies-acc-bscc}. 
\end{proof}
\begin{lemma}\label{lem:acc-bscc-implies-good-bscc}
	Let $u\in \Paths^*(\dtmc)$ such that the lifting of $u$ to $\dtmc \otimes \automaton$ reaches a BSCC $B_\automaton$, and the lifting of $u$ to  $\dtmc \otimes \lfdfa$ reaches a BSCC $B_\lfdfa$. If $B_\automaton$ is accepting then $B_\lfdfa$ is good.
\end{lemma}
\begin{proof}
	Let $B_\automaton$ be accepting, let the lifting of $u$ to $\dtmc \otimes \lfdfa$ end in $\langle s, q \rangle $, and let $I_B = I_{B_\lfdfa} = \dtmc_s \times \lfdfa_q \times \pfdfa^q$. Towards a contradiction, assume that $I _B$ contains a rejecting BSCC $C$. 
	
	Let $v$ be a finite path in $I_B$ that ends in $C$. 
	Because the first two components of any state in $I_B$ correspond to states in $\dtmc \otimes \lfdfa$ and since all considered automata are total, $C$ must contain for each $\langle s', q' \rangle \in B_{\lfufa}$ at least one state $\langle s', q', r' \rangle$. W.l.o.g. we thus assume that $v$ ends in a state $\langle s, q, r \rangle \in C$. Then the lifting of $uv$ to $\dtmc \otimes \lfdfa$ ends in $\langle s, q \rangle$ and, as $C$ is rejecting, $r \notin \afinal^{\pfdfa^q}$.
	
	Let $\langle s, p \rangle \in B_\automaton$ be the state of $\dtmc \otimes \automaton$ reached after moving according to $u$, and let $\langle s, p_1 \rangle$ be the state in $B_\automaton$ reached after $uv$. 
	Because $B_\automaton$ is a BSCC of $\dtmc \otimes \automaton$ there is a finite path $w_1 \in \Paths^*(s)$ of $\dtmc$ that, when lifted to $\dtmc \otimes \automaton$ and followed from $\langle s, p_1 \rangle$, brings $B_\automaton$ back to $\langle s, p \rangle$ while visiting all states of $B_\automaton$ at least once. Because $B_\automaton$ is a BSCC, we can choose $w_1 = vz$ for $v$ as above and some finite $z$. Let $\langle s, q_1 \rangle$ be the state in $\dtmc \otimes \lfdfa$ reached from $\langle s, q \rangle$ on $w_1$. Because $\langle s, q_1 \rangle \in B_\lfufa$, there is a finite path $r_1$ taking $\langle s, q_1 \rangle$ back to $\langle s, q \rangle$ in $\dtmc \otimes \lfdfa$. Let $\langle s, p_2 \rangle$ be the state in $B_\automaton$ reached when following $r_1$ from $\langle s, p \rangle$. Now, from $\langle s, p_2 \rangle$ there is again a $w_2 = vz_2$ that takes $B_\automaton$ back to $\langle s, p \rangle$ while visiting all states of $B_\automaton$.
	
	By iterating this argument we obtain a sequence of finite words $w_i$ and $r_i$ such that, for all $i$, $w_i$ starts with $v$ and takes the BSCC $B_\automaton$ from the state $\langle s, p_i \rangle$ back to $\langle s, p \rangle$ while visiting all states of $B_\automaton$, and $r_i$ takes $\dtmc \otimes \lfdfa$ from the state $\langle s, q_i \rangle$ back to $\langle s, q \rangle$. As $B_\lfufa$ and $B_\automaton$ are finite, there must be $j, k \in \mathbb{N}$ with $k > j$ such that, after $r_{j}$ and $r_{k}$, $\dtmc \otimes \automaton$ is in the same state $\langle s, p^* \rangle$. Pictorially, the construction looks as follows: 
	\begin{alignat*}{8}
		& \dtmc {\otimes} \lfdfa\colon \overset{u}{\longrightarrow} \langle s, q \rangle &&\overset{v}{\longrightarrow} \langle s, q \rangle &&\overset{w_1}{\longrightarrow} \langle s, q_1 \rangle &&\overset{r_1}{\longrightarrow} \langle s, q \rangle &&\overset{w_2}{\longrightarrow} \ldots &&\overset{r_{j}}{\longrightarrow} \langle s, q \rangle &&\overset{w_{{j}+1}}{\longrightarrow} \ldots &&\overset{r_{k}}{\longrightarrow} \langle s, q \rangle \\
		& \dtmc {\otimes} \automaton\colon \overset{u}{\longrightarrow} \langle s, p \rangle &&\overset{v}{\longrightarrow} \langle s, p_1 \rangle &&\overset{w_1}{\longrightarrow} \langle s, p \rangle &&\overset{r_1}{\longrightarrow} \langle s, p_2 \rangle &&\overset{w_2}{\longrightarrow} \ldots &&\overset{r_{j}}{\longrightarrow} \langle s, p^* \rangle &&\overset{w_{{j}+1}}{\longrightarrow} \ldots &&\overset{r_{k}}{\longrightarrow} \langle s, p^* \rangle 
	\end{alignat*}
	Now let $x = uvw_1r_1 \ldots w_{j-1}r_{j}$ and $y = w_{j+1} \ldots r_{k}$. Then the lifting of $x$ to $\dtmc \otimes \automaton$ ends in some state $\langle s, p^* \rangle \in B_\automaton$, and the lifting of $y$ to $\dtmc \otimes \automaton$ starting in $\langle s, p^* \rangle$ closes a loop on $\langle s, p^* \rangle$ that visits all state of $B_{\automaton}$ at least once. Moreover, the lifting of $x$ to $\dtmc \otimes \lfdfa$ ends in the state $\langle s, q \rangle$, and the lifting of $y$ to $\dtmc \otimes \lfdfa$ closes a loop on $\langle s, q \rangle$. 
	
	As $B_{\automaton}$ is accepting and because upon $xy^\omega$ all states of $B_\automaton$ are visited infinitely often, it follows from \Cref{lem:accepting-and-rejecting-bscc-dra-product} that $xy^\omega \in E$. On the other hand, $(x, y) \in \Sigma^{*+}$ is normalized w.r.t $\lfdfa$ (since $\lfdfa(x) = q = \delta^{\lfdfa}(q, y)$), and $y \notin \accLanguage(\pfdfa^q)$. The latter holds since $y$ starts with $w_{j+1} = vz_{j+1}$, and so the lifting of $y$ to $I_B$ reaches the rejecting BSCC $C$ after $v$ and cannot leave $C$ anymore. As $y$ closes a loop on $\langle s, q \rangle$ in $\dtmc \otimes \lfdfa$, the state reached in $I_B$ on $y$ must be of the form $\langle s, q, r \rangle$ for some $r \in \pfdfa^q$. Because $C$ is rejecting, it follows that $r \notin \afinal^{\pfdfa^q}$, so $y \notin \accLanguage(\pfdfa^q)$.  Thus, $(x, y) \notin \llbracket \fdfa \rrbracket$, and the saturation of $\fdfa$ implies that $xy^\omega \notin \accLanguage_\omega(\fdfa) = E$, a contradiction.
\end{proof}

To complete the proof of \Cref{lem:accepting-bscc-iff-good-bscc} it remains to show that if $u$ as in the lemma reaches a good BSCC of $\dtmc \otimes \lfdfa$ then $B_\automaton$ is accepting. Before showing this claim in \Cref{lem:good-bscc-implies-acc-bscc}, we state some additional lemmas used in the proof.

Let $\langle s, q \rangle \in B_\lfdfa$ be the state reached in $\dtmc \otimes \lfdfa$ after reading $u$, and let $\langle s, p \rangle \in B_\automaton$ be the corresponding state in $\dtmc \otimes \automaton$. Given $I_B =  I_{B_\lfdfa} = \dtmc_s \times \lfdfa_q \times \pfdfa^q$, we set $G = G_{B, p} = I_B \times \automaton_p$. Intuitively, $G$ annotates the states of $I_B$ with information on the current state of $\automaton$. 
\begin{restatable}{lemma}{LemAccBSCCinG}\label{lem:bscc-in-full-product}
	If every BSCC of $I_B = I_{\langle s, q \rangle}$ contains an accepting state then there is a BSCC in $G = I_B \times \automaton_p$ that contains a state $\langle s, q, f, a \rangle$ for an $f \in \afinal^{\pfdfa^q}$.
\end{restatable}
\begin{proof}
	Assume that every BSCC of $I_B$ contains a state $\langle s, q, f \rangle$. Let $z_1$ be a finite run of $I_B$ that ends in such a state in some BSCC $C$ of $I_B$, and let $\langle s, q, f, a_1 \rangle$ be the state of $G$ reached when lifting $z_1$ to $G$. If $\langle s, q, f, a_1 \rangle$ is in a BSCC of $G$ then we are done. 
	
	Otherwise, there must be a $\langle s', q', r', a_2 \rangle$ reachable from $\langle s, q, f, a_1 \rangle$ via a finite $z_2$ such that from $\langle s', q', r', a_2 \rangle$ the state $\langle s, q, f, a_1 \rangle$ is not reachable anymore. 
	
	Since, along $z_1$, $I_B$ enters the BSCC $C$, the state $\langle s', q', r' \rangle$ must be in $C$ as well. Hence, there is a finite $z_3$ from $\langle s', q', r' \rangle$ back to $\langle s, q, f \rangle$. Lifted to $G$ this means that the run of $G$ on $z_3$ takes the model from $\langle s', q', r', a_2 \rangle$ to a state $\langle s, q, f, a_3 \rangle$ with $a_3 \neq a_1$ (as otherwise, $\langle s, q, f, a_1 \rangle$ would be reachable from $\langle s', q', r', a_2 \rangle$). 
	
	When repeating this procedure we obtain a sequence of finite $z_i$ such that for all odd $i$ the product $G$ is in a state of the form $\langle s, q, f, a_i \rangle$, and such that $\langle s, q, f, a_i \rangle$ is not reachable anymore from the state $\langle s_{i+1}, q_{i+1}, r_{i+1}, a_{i+1} \rangle$ reached when following $z_{i+1}$ from $\langle s, q, f, a_i \rangle$.
	
	Because $\automaton$ is finite there are only finitely many choices for the fourth component of any state in $G$. Hence, there must be finite numbers $j, k$ with $k > j$ such that $G$ is in the same state $\langle s, q, f, a^* \rangle$ after following $z_1 \ldots z_{j}$ and $z_1 \ldots z_{j} \ldots z_{k}$. But then $\langle s, q, f, a^* \rangle$ is reachable from the state $\langle s_{j+1}, q_{j+1}, r_{j+1}, a_{j+1} \rangle$ reached on $z_{j+1}$ from $\langle s, q, f, a^* \rangle$, which contradicts the assumption on the sequence $(z_i)_{i \geq 1}$. It follows that the set of states reachable from $\langle s, q, f, a^* \rangle$ is strongly connected and cannot be left anymore, i.e., forms a BSCC of $G$.
\end{proof}

\begin{lemma}\label{lem:duo-normalized-word-reaches-bscc}
	If every BSCC of $I_B = I_{\langle s, q \rangle}$ contains an accepting state then there is a $v \in \Paths^*(\dtmc_s)$ such that the lifting of $v$ to $I_B$ ends in an accepting state $\langle s, q, f \rangle$ in a BSCC of $I_B$ on which $v$ also closes a loop, i.e., it holds that $I_B(v) = \langle s, q, f \rangle = I_B(v^2)$. 
\end{lemma}
\begin{proof}
	Let $C$ be a BSCC of $I_B$ and let $\langle s, q, f \rangle \in C$ for an $f \in \afinal^{\pfdfa^q}$, which must exist since we assume all BSCCs of $I_B$ to contain an accepting state. Let $\rho$ be a finite path in $I_B$ that ends in $\langle s, q, f \rangle$. Then the projection of $\rho$ onto $\dtmc \otimes \lfdfa$ induces a loop on $\langle s, q \rangle$, and for the trace $x$ of the projection of $\rho$ to $\dtmc$ it holds that $(u, x) \in \llbracket \fdfa \rrbracket$ for all $u \in \Sigma^*$ with $\lfdfa(u) = q$. 
	
	The saturation of $\fdfa$ implies that $(u, x^i) \in \llbracket \fdfa \rrbracket$ for all $i \geq 1$ \cite{SPFA}. Hence, the run of $\pfdfa^q$ on any $x^i$ must end in a final state $f_i$. Since $\vert \afinal^{\pfdfa^q} \vert < \infty$ there are minimal $j, k$ with $k > j$ and $\pfdfa^q(x^{j}) = f^* = \pfdfa^q(x^{k})$ for an $f^* \in \afinal^{\pfdfa^q}$. Let $n = k - j$. 
	
	Because $\pfdfa^q$ is deterministic it follows that, after reading $x^{j}$, the runs of $\pfdfa^q$ on repetitions of $x$ follow a cycle of length $x^n$ \cite{FDFAAORL}. Now let $m = n \cdot j$. Then $m \geq j$ and so the final state $f = \pfdfa^q(x^m)$ is on the aforementioned cycle. But then it follows that also $\delta^{\pfdfa^q}(f, x^m) = f$, because the run of $\pfdfa^q$ on $x^m$ starting in $f$ traverses this cycles exactly $j$ times. 
	
	All in all, the run of $I_B$ on $v=x^m$ reaches the state $\langle s, q, f \rangle \in C$ for $C$ a BSCC of $I_B$,
	and the run on $v$ starting in $\langle s, q, f \rangle$ closes a loop. Hence, $I_B(v) = \langle s, q, f \rangle = I_B(v^2)$. 
\end{proof}

\begin{corollary}\label{cor:duo-normalized-word-into-bscc-of-full-product}
	If every BSCC of $I_B = \dtmc_s \times \lfdfa_q \times \pfdfa^q$ contains an accepting state then there is a $v \in \Paths^*(\dtmc_s)$ such that the lifting of $v$ to $G$ reaches a state $\langle s, q, f, a \rangle$ in a BSCC of $G$, and additionally $v$ closes a loop on $\langle s, q, f, a \rangle$ in $G$. 
\end{corollary}
\begin{proof}
	Let $x$ be a finite word reaching a state $\langle s, q, f, a \rangle$ in a BSCC $C$ of $G$, which exists by \Cref{lem:bscc-in-full-product}. Similar to the proof of \Cref{lem:duo-normalized-word-reaches-bscc}, the finiteness of $G$ together with the saturation of $\fdfa$ yields that there must be $j, k$ with $j < k$ and such that, after following $x^{j}$ resp. $x^{k}$, the model is in the same state $\langle s, q, f^*, a^* \rangle$ for some $f^* \in \afinal^{\pfdfa^q}$. Taking $v = x^{j \cdot n}$ for $n = k - j$ yields the claim. 
\end{proof}

We now have all the pieces in place to prove the reverse of \Cref{lem:acc-bscc-implies-good-bscc}. 
\begin{lemma}\label{lem:good-bscc-implies-acc-bscc}
	Let $u\in \Paths^*(\dtmc)$ such that the lifting of $u$ to $\dtmc \otimes \automaton$ reaches a BSCC $B_\automaton$, and the lifting of $u$ to  $\dtmc \otimes \lfdfa$ reaches a BSCC $B_\lfdfa$. If $B_\lfdfa$ is good then $B_\automaton$ is accepting.
\end{lemma}
\begin{proof}
	We show the claim by contraposition, i.e., we show that if  $B_\automaton$ is rejecting then $B = B_\lfdfa$ cannot be good. Let $\langle s, q \rangle$ be the state of $\dtmc \otimes \lfdfa$ reached via $u$, $\langle s, p \rangle$ the corresponding state of $\dtmc \otimes \automaton$, and $I_B = \dtmc_s \times \lfdfa_q \times \pfdfa^q$. Towards a contradiction, assume that $B$ is good, i.e., that every BSCC of $I_B$ contains a state $\langle s, q, f \rangle$ for some $f \in \afinal^{\pfufa^q}$. 
	
	For $G = I_B \times \automaton_p$ we get from \Cref{cor:duo-normalized-word-into-bscc-of-full-product} that there is a $v \in \Paths^*(\dtmc)$ whose lifting to $G$ ends in a state $\langle s, q, f, a \rangle$ of a BSCC $C$ of $G$, and closes a loop on $\langle s, q, f, a \rangle$. 
	
	Because $B_\automaton$ is a BSCC there is a path $w = vx$ for $v$ as above and some $x$ that, when lifted to $\dtmc \otimes \automaton$ and followed from $\langle s, a \rangle$, visits all states of $B_\automaton$ at least once. Let $\langle s', q', p' \rangle$ be the state reached in $I_B$ when following (the lifting of) $w$ starting in $\langle s, q, f \rangle$, and let $\langle s', a' \rangle$ be the corresponding state reached in $\dtmc \otimes \automaton$ from $\langle s, a \rangle$ on $w$. Hence, when moving according to the lifting of $w$ starting in $\langle s, q, f, a\rangle$, the state $\langle s', q', p', a' \rangle$ of $G$ is reached. Because $\langle s, q, f, a \rangle, \langle s', q', p', a' \rangle \in C$ there is a $y$ that takes $G$ from $\langle s', q', p', a' \rangle$ back to $\langle s, q, f, a \rangle$. 
	
	From the construction above it follows that the (lifting of the) finite word $wy$ closes a loop in $B_\automaton$ on the state $\langle s, a \rangle$ reached after moving according to $uv$ in $\dtmc \otimes \automaton$, and that along this loop all states in $B_\automaton$ are visited at least once. Because $B_\automaton$ is rejecting, \Cref{lem:accepting-and-rejecting-bscc-dra-product} implies that $uv(wy)^\omega \notin E$. 
	
	On the other hand, the lifting of $vwy = vvxy$ to $I_B$ yields the sequence 
	\begin{center}
		$\langle s, q, \ainit^{\pfdfa^q} \rangle \overset{v}{\longrightarrow} \langle s, q, f \rangle \overset{v}{\longrightarrow} \langle s, q, f \rangle \overset{x}{\longrightarrow} \langle s', q', p' \rangle \overset{y}{\longrightarrow} \langle s, q, f \rangle$. 
	\end{center}
	
	As $v$ closes a loop on $\langle s, q, f \rangle$, and because all components are deterministic, the lifting of $wy = vxy$ to $I_B$ is then of the form 
	\begin{center}
		$\langle s, q, \ainit^{\pfdfa^q} \rangle \overset{v}{\longrightarrow} \langle s, q, f \rangle \overset{x}{\longrightarrow} \langle s', q', p' \rangle \overset{y}{\longrightarrow} \langle s, q, f \rangle$. 
	\end{center}
	
	It follows that $\pfdfa^q(vxy) = \pfdfa^q(wx) = f \in \afinal^{\pfdfa^q}$ and $\lfdfa(uv) = q = \delta^\lfdfa(q, wy)$. Therefore, $(uv, wy)$ is normalized w.r.t. $\lfdfa$ and $wy \in \accLanguage(\pfdfa^q)$, so $(uv, wy) \in \llbracket \fdfa \rrbracket$. But then $uv(wy)^\omega \in E$, a contradiction. 
	
	Thus, $I_B$ must contain at least one rejecting BSCC, and so $B$ cannot be good. It follows by contraposition that if $B_\lfdfa$ is good then $B_\automaton$ is accepting, as was to be shown.  
\end{proof}

\LemResidualLanguages*
\begin{proof}
	Let $p = \lfdfa(x) = \lfdfa(y)$ and let $\fdfa_p$ be like $\fdfa$ but with $\lfdfa_p$ as leading automaton. 
	
	We start by proving that $\fdfa_p$ is saturated. To see this, let $w = uv^\omega$ and let $(x_1, y_1), (x_2, y_2)$ be two normalized (w.r.t. $\lfdfa_p$) decompositions of $w$. Moreover,  let $(x_1, y_1) \in \llbracket \fdfa_p \rrbracket$, so $y_1 \in \accLanguage(\pfdfa^{q_1})$ for $q_1 = \lfdfa_p(x_1)$. 
	Now let $r \in \Sigma^*$ be such that $\lfdfa(r) = p$. Then $(rx_1, y_1)$ and $(rx_2, y_2)$ are normalized decompositions of $ruv^\omega$. Moreover, $\lfdfa(rx_1) = q_1$ and $y_1 \in \accLanguage(\pfdfa^{q_1})$, so $(rx_1, y_1) \in \llbracket \fdfa \rrbracket$. It follows from the saturation of $\fdfa$ and the fact that $rx_1y_1^\omega = rx_2y_2^\omega$ that $(rx_2, y_2) \in \llbracket \fdfa \rrbracket$, and so $(x_2, y_2) \in \llbracket \fdfa_p \rrbracket$. Hence, $\fdfa_p$ is saturated. 
	
	\smallskip 
	To prove the claim of the lemma, let $u \in \Sigma^*$, $v \in \Sigma^+$. Because $\accLanguage(\fdfa) = \UP(E)$ we have
	\begin{align}
		xuv^\omega \in E \quad &\text{ iff } \quad \text{ there are } i, j \in \mathbb{N} \text{ such that } (xuv^i, v^j) \in \llbracket \fdfa \rrbracket \nonumber \\ &\text{ iff } \quad \text{ there are } i, j \in \mathbb{N} \text{ such that } (uv^i, v^j) \in \llbracket \fdfa_p \rrbracket. \label{eq:strong-properties-bscc-eq}
	\end{align}
	
	The direction from right to left in the first equivalence is clear. For the reverse direction, observe that $(xu, v) \in \Sigma^{*+}$ is a decomposition for $xuv^\omega$. Using a normalization procedure described in \cite{FDFAAORL} we obtain values $i,j \in \mathbb{N}$ such that $(xuv^i, v^j)$ is a normalized decomposition of $xuv^\omega$. Because $xuv^\omega \in \UP(E) = \accLanguage(\fdfa)$ and the saturation of $\fdfa$ we have $(xuv^i, v^j) \in \llbracket \fdfa \rrbracket$. 
	
	Regarding the second equivalence, let first $(xuv^i, v^j) \in \llbracket \fdfa \rrbracket$. Then
	\begin{itemize}
		\item $\lfdfa_p(uv^i) = \delta^\lfdfa(p, uv^i) =\lfdfa(xuv^i) \eqcolon q$, 
		\item $\delta^{\lfdfa_p}(q, v^j) = \delta^{\lfdfa}(q, v^j) = q$ since $(xuv^i, v^j)$ is normalized w.r.t. $\lfdfa$, and 
		\item $v^j \in \accLanguage(\pfdfa^{q})$ since $(xuv^i, v^j) \in \llbracket \fdfa \rrbracket$ and the progress automata of $\fdfa$ and $\fdfa_p$ coincide.
	\end{itemize} 
	Thus, $(uv^i, v^j)$ is normalized w.r.t. $\lfdfa_p$ and accepted by $\fdfa_p$, so $(uv^i, v^j) \in \llbracket \fdfa_p \rrbracket$. 
	
	Similarly, for the reverse direction let 
	$(uv^i, v^j) \in \llbracket \fdfa_p \rrbracket$. Then
	\begin{itemize}
		\item $\lfdfa(xuv^i) = \delta^\lfdfa(p, uv^i) = \lfdfa_p(uv^i) \eqcolon q'$, 
		\item $\delta^\lfdfa(q', v^j) = \delta^{\lfdfa_p}(q', v^j) = q'$ since $(uv^i, v^j)$ is normalized w.r.t. $\lfdfa_p$, and 
		\item $v^j \in \accLanguage(\pfdfa^{q'}) = \accLanguage(\pfdfa^{\lfdfa_p(uv^i)}) = \accLanguage(\pfdfa^{\lfdfa(xuv^i)})$. 
	\end{itemize}
	It follows that $(xuv^i, v^j) \in \llbracket \fdfa \rrbracket$. In particular, we can take the same values for $i$ and $j$ in the second and third statement of \Cref{eq:strong-properties-bscc-eq}. 
	
	\smallskip
	Analogously to \Cref{eq:strong-properties-bscc-eq} we can show that $yuv^\omega \in E$ iff there are $k,l \in \mathbb{N}$ such that $(uv^k, v^l) \in \llbracket \fdfa_p \rrbracket$. Hence,
	\begin{align*}
		xuv^\omega \in E \quad \text{ iff } \quad yuv^\omega \in E
	\end{align*}
	and so for all $w = uv^\omega$ we have $w \in \UP(\resLang_x)$ iff $w \in \UP(\resLang_y)$, i.e., $\UP(\resLang_x) = \UP(\resLang_y)$. Since $\resLang_x$ and $\resLang_y$  are $\omega$-regular (an automaton accepting, e.g., $\resLang_x$ can be obtained from a DRA $\automaton$ for $E$ by changing its initial state to $\automaton(x)$) it follows from the fact that two $\omega$-regular languages coincide iff they contain exactly the same ultimately periodic words \cite{UPWROL} that $\resLang_x = \resLang_y$ 
\end{proof}

\PropComputationOfProbabilites*
\begin{proof}
	From \Cref{lem:residual-languages-coincide} it follows that $\dtmc$ almost surely generates a path whose lifting to $\dtmc \otimes \lfdfa$ reaches a BSCC $B_\lfdfa$, and whose lifting to $\dtmc \otimes \automaton$ reaches a BSCC $B_\automaton$. We have proven in \Cref{lem:accepting-bscc-iff-good-bscc} that in this case $B_\lfdfa$ is good iff $B_\automaton$ is accepting. Therefore, $\probMeasure^{\dtmc \otimes \lfdfa}(\lozenge \mathit{goodBSCC}) = \probMeasure^{\dtmc \otimes \automaton}(\lozenge \mathit{accBSCC})$, where $\mathit{accBSCC}$ is the set of accepting BSCCs of $\dtmc \otimes \automaton$. Since $\probMeasure^{\dtmc}(E) = \probMeasure^{\dtmc \otimes \automaton}(\lozenge \mathrm{acc BSCC})$ \cite{PoMC}, this implies the proposition. 
\end{proof}

\LemDoubleExponentialBlowUp*
\begin{proof}
	Let $\Sigma = \{0, 1, \$, \#\}$ and let 
	\begin{center}
		$L_n = \{x \# w \# y \$ w \#^\omega \mid x, y \in \{0, 1, \#\}^*, w \in \{0,1\}^n\}$
	\end{center}
	be the language that contains all words that contain $\$$ exactly once and are such that the word $w \in \{0,1\}^n$ occurring between $\$$ and the next $\#$ already occurred somewhere before between two $\#$. This language was also used in \cite{FLTBT} to prove that there is a double-exponential lower bound for the translation from LTL to deterministic $\omega$-automata. It was shown in the same paper that for each $n$ there is an LTL formula $\varphi_n$ of length quadratic in $n$ with $\Words(\varphi_n) = L_n$. 
	
	Thus, the only thing left to prove is that every saturated FDFA for $L_n$ has a leading automaton with at least $2 ^{2 ^n}$ states. To see this, observe that when removing the padding $\#^\omega$ at the end of words in $L_n$, we obtain the finite word language
	\begin{center}
		$L_n' = \{x \# w \# y \$ w \mid x, y \in \{0, 1, \#\}^*, w \in \{0,1\}^n\}$.
	\end{center}
	A similar language (over a ternary alphabet) was considered in \cite{Alternation}, and with arguments analogous to those used in \cite{Alternation} it follows that every DFA for $L_n'$ needs at least $2^{2^n}$ states. But then the leading automaton of any saturated FDFA $\fdfa$ accepting $L_n$ also requires at least $2^{2^n}$ states, since the non-periodic part $u$ of any decomposition $(u,v) \in \llbracket \fdfa \rrbracket$ must correspond to a word in $L_n'$. This finishes the proof.
\end{proof}

\section{\texorpdfstring{Proofs of \Cref{sec:fufa}}{Proofs of Section 4}}

\LemFDFAareFUFA*
\begin{proof}
	Because $\lfufa_\fufa$ has the same structure as $\lfdfa_\fdfa$, it is deterministic. This implies that $\vert \lfufa_\fufa(u) \vert = 1$ for every $u \in \Sigma^*$ and so it is clear that $\fufa$ is decomposition-unambiguous. 
	
	Let now $v \in \accLanguage(\pfufa_\fufa^q)$ for $v$ the period of some $q$-normalized decomposition $(u,v) \in \llbracket \fufa \rrbracket$. By construction we also have $(u,v) \in \llbracket \fdfa \rrbracket$. Because $\fdfa$ is saturated and thus \emph{power-stable} \cite{SPFA}, $v^i \in \accLanguage(\pfdfa_\fdfa^q)$ for all $i \geq 1$. As $\pfdfa_\fdfa^q$ is deterministic, the unique run of $\pfdfa_\fdfa^q$ on $v$ ends in some $f \in \afinal^{\pfufa_\fdfa^q}$, and consequently the run of $\pfdfa_\fdfa^q$ on $v^i, i \geq 1$, must be in $f$ after reading $v$ as well. Combining this with the fact that $v^i \in \accLanguage(\pfdfa_\fdfa^q)$ for all $i$, it follows that $\afinal^{\pfdfa_\fdfa^q} \cap \delta^{\pfdfa_\fdfa^q}(f, v^j) \neq \emptyset$ for all $j \geq 1$. Since $\pfdfa^q_\fdfa$ and $\pfufa^q_\fufa$ have the same structure, also $\afinal^{\pfufa_\fufa^q} \cap \delta^{\pfufa_{\fufa}^q}(f, v^j) \neq \emptyset$ for all $j \geq 1$. Therefore, $\fufa$ is power-unambiguous. 
	
	Lastly, let $w$ be an ultimately periodic word and let $(u_1, v_1), (u_2, v_2)$ be two normalized decompositions for $w$. Because $\fdfa$ is saturated, either both or none of them are accepted by $\fdfa$. Since our construction ensures that $(u,v) \in \llbracket \fufa \rrbracket$ iff $(u,v) \in \llbracket \fdfa \rrbracket$, $\fufa$ also either accepts both $(u_1, v_1), (u_2, v_2)$, or rejects both of them. Hence, $\fufa$ is saturated. 
\end{proof}

\PropSuccinctnessFUFA*
\begin{proof}
	Let $\Sigma = \{a,b,c\}$ and consider, for $n \in \mathbb{N}$, the family of regular languages
	\begin{center}
		$L_n' = (a+c)^*a(a+c)^{n-1}b$
	\end{center} 
	which contains all words having a $b$ only in the last position and an $a$ exactly $n$ steps before. 
	$L_n'$ is accepted by a UFA with $n{+}2$ states, which is similar to the automaton in \Cref{fig:leading-automaton-n-steps-before}, but with unique final state $q_{n+1}$ and no self-loop in this state. On the other hand, any DFA with language $L_n'$ needs at least $2^n$ states since it must keep track of the occurrences of $a$ in the last $n$ positions \cite{MCUA}. 
	We can extend $L_n'$ to an $\omega$-regular language $L_n$ by adding padding, i.e., by setting $L_n = L_n' \cdot \Sigma^\omega$. 
	
	$L_n$ is a co-safety property, and an $\omega$-word is in $L_n$ iff it has a prefix in $L_n'$. Hence, any family of finite automata representing $L_n$ must have a leading automaton ``recognizing'' $L_n'$. A saturated FDFA for $L_n$ therefore has a leading automaton with at least $2^n$ states, while a saturated, decomposition- and power-unambiguous FUFA of size $(n{+}2, 2)$ for $L_n$ is obtained by choosing a leading automaton like in \Cref{fig:leading-automaton-n-steps-before}, setting $\mstates = \{q_{n+1}\}$, and ensuring $\accLanguage(\pfufa^{q_{n+1}}) = \Sigma^+$. 
\end{proof}

\ThmFUFAToUBA*
\begin{proof}
	Let $n \in \mathbb{N}$ and $L_n' = (0 + (01^*)^{n-1}0)^*$. This language is also considered in \cite{SEAFAPAFA}, where it is shown that every UFA accepting $L_n'$ needs at least $2^n{-}1$ states. Let $L_n = L_n' \cdot \$^\omega$. 
	
	We start by proving that the FUFA $\fufa_n = (\lfufa_n, \{\pfufa_n^{q}\}_{q \in \{q_{n+1}\}})$ depicted in  \Cref{fig:leading-automaton-blow-up-fufa-2-uba}, which is of size $(n{+}1, 2)$,  is a saturated, decomposition- and power-unambiguous FUFA for $L_n$. 
	
	To this end, we first observe that the leading automaton $\lfufa_n$ is obtained by adding the state $q_{n+1}$ (together with its incoming and outgoing transitions) to an NFA presented in \cite[Fig. 1]{SEAFAPAFA} which accepts $L_n'$. It immediately follows that $\accLanguage(\fufa_n) = L_n$. 
	Furthermore, $\vert \mstates \vert = 1$ and so it trivially holds that $\fufa_n$ is decomposition-unambiguous. Because of the structure of $\lfufa_n$, any $(u,v) \in \llbracket \fufa_n \rrbracket$ must be $q_{n+1}$-normalized, implying that in any such decomposition $v = \$^k$ for some $k \geq 1$. But then also $v^j = \$^{k \cdot j}$ for every $j \geq 1$, yielding $\afinal^{\pfufa_n^{q_{n+1}}} \cap \delta^{\pfufa_n^{q_{n+1}}}(\$, v^j) \neq \emptyset$. Thus, $\fufa_n$ is power-unambiguous. Lastly, given a word $w$ and two $q_{n+1}$-normalized decompositions $(u_1, v_1)$ and $(u_2, v_2)$ of $w$, it follows that $q_{n+1}$ must be reachable in $\lfufa_n$ on $u_1$ and $u_2$, and both $v_1$ and $v_2$ have to be of the form $\$^{k_i}$ for some $k_1, k_2 \geq 1$. But then $v_1, v_2 \in \accLanguage(\pfufa_n^{q_{n+1}})$ and so $(u_1, v_1), (u_2, v_2) \in\llbracket \fufa_n \rrbracket$. Thus, $\fufa_n$ is saturated. 
	
	What remains to be proven is that every UBA $\automatonB_n$ for $L_n$ has at least $2^{n}{-}1$ states. Our argument takes inspiration from the proof of \cite[Thm. 6]{FABA}, where Löding and Pirogov show an exponential blow-up when translating NBA to \emph{finitely-ambiguous Büchi automata}, which are Büchi automata in which every $\omega$-word has at most finitely many accepting runs.\footnote{The proof of \cite[Thm. 6]{FABA} also uses an $\omega$-extension of the languages $L_n'$ from \cite{SEAFAPAFA}, but in contrast to the languages $L_n = L_n' \cdot \$^\omega$ the languages $L_n''$ in \cite{FABA} are of the form $L_n'' = (L_n' \cdot \$)^\omega$.} 
	
	Let $\automatonB_n$ be a UBA that accepts $L_n$. Without loss of generality we assume that $\automatonB_n$ is \emph{trim}, i.e, that for every state $q$ of $\automatonB_n$ the automaton obtained from $\automatonB_n$ by making $q$ initial has a nonempty language. This assumption can be made w.l.o.g. because any non-trim UBA can be transformed into a language equivalent trim one by removing all states from which no final state can be visited infinitely often. In particular, this transformation does not increase the number of required states. 
	Additionally, let $\automaton_n$ be the NFA over $\Sigma = \{0,1\}$ obtained from $\automatonB_n$ by making all states in $\automatonB_n$ with an outgoing $\$$-transition final, and afterwards removing all $\$$-transitions (and states that have become unreachable) from the automaton. 
	
	We show that $\automaton_n$ is a UFA. Towards a contradiction assume that this is not the case, i.e., that there is a word $w \in \{0,1\}^*$ that has two accepting runs $x_1$ and $x_2$ in $\automaton_n$. Let $f_1, f_2 \in \afinal^{\automaton_n}$ be the final states reached upon $x_1, x_2$ (where $f_1 = f_2$ is possible). Because $f_1, f_2 \in \afinal^{\automaton_n}$, there are $\$$-transitions in the states $q_1, q_2$ of $\automatonB_n$ corresponding to $f_1$ and $f_2$, and states $p_1, p_2$ reachable by these $\$$-transitions. Because $\automatonB_n$ is trim, there are $v_i \in \{0,1,\$\}^\omega$ accepted by the automata obtained from $\automatonB_n$ when making $p_i$, $i \in \{1, 2\}$, initial. Since $\accLanguage(\automatonB_n) = L_n$, and because the states $p_i$ are reachable via $\$$-transitions in $\automatonB_n$, it follows that $v_1$ and $v_2$ must be of the form $\$^\omega$. 
	Hence, both $wv_1$ and $wv_2$ equal $w\$^\omega \in L_n = \accLanguage(\automatonB_n)$. But then there are two different accepting runs of $\automatonB_n$ on $w\$^\omega$, which contradicts the unambiguity of $\automatonB_n$. Thus, any word $w$ accepted by $\automaton_n$ has exactly one accepting run in the automaton, i.e., $\automaton_n$ is unambiguous and hence a UFA. 
	
	Next, we show that $\accLanguage(\automaton_n) = L_n'$. Recall that the alphabet of $\automaton_n$ is $\{0,1\}$, and let $w \in \{0,1\}^*$. 
	If $w \in L_n'$ then $w \cdot \$^\omega \in L_n$, so there must be a state $q$ reachable on $w$ in $\automatonB_n$ that has a $\$$-transition. But then the state in $\automaton_n$ corresponding to $q$ is final and reachable on $w$, so $w \in \accLanguage(\automaton_n)$. 
	If $w \notin L_n'$ no state reachable in $\automatonB_n$ on $w$ can have a $\$$-transition. This holds since otherwise the fact that $\automatonB_n$ is trim would, as above, imply that $w \cdot \$^\omega \in \accLanguage(\automatonB_n)$, while at the same time $w \cdot \$^\omega \notin L_n' \cdot \$^\omega$, a contradiction. Hence, no final state of $\automaton_n$ is reachable on $w$, and so $w$ is not accepted by $\automaton_n$. All in all, $w \in L_n'$ iff $w \in \accLanguage(\automaton_n)$, i.e., $\accLanguage(\automaton_n) = L_n'$.
	
	Since $\automaton_n$ is a UFA that accepts $L_n'$, it has at least $2^{n}{-}1$ states \cite{SEAFAPAFA}. Furthermore, $\automaton_n$ is obtained from $\automatonB_n$ by (potentially) removing states, implying that $\automatonB_n$ has at least as many states as $\automaton_n$. Therefore, every UBA that accepts $L_n$ must have at least $2^{n}{-}1$ states.
\end{proof}

The following lemma of Calbrix \textit{et al.} \cite{UPWROL} is used in the proof of \Cref{lem:fufa-to-nba-polynomial}. \unskip
\begin{lemma}[\cite{UPWROL}]\label{lem:calbrix-lemma}
	Let $M, N \subseteq \Sigma^*$ such that $M = M \cdot N^*$ and $N = N^+$. Then for each $w \in \Sigma^\omega$ it holds that $w \in \UP(M \cdot N^\omega)$ iff there are $u \in M$ and $v \in N$ such that $w = uv^\omega$. 
\end{lemma}

\LemFUFAtoNBA*
\begin{proof}
	We adjust the proof of \cite[Thm. 5.8]{FDFAAORL}---where it is shown that every saturated FDFA $\fdfa$ can be translated into an NBA of size $\mathcal{O}(n^2k^3)$ whose accepted language contains precisely the ultimately periodic words in $\accLanguage(\fdfa)$---to the setting of FUFA. 
	
	Given $q \in \mstates$ and $f \in \afinal^{\pfufa^q}$, define the regular languages
	\begin{align*}
		M_q &= \{v \in \Sigma^* \mid q \in \lfufa(v)\} \\
		N_{q,f} &= \{v \in \Sigma^+ \mid q \in \delta^{\lfufa}(q, v) \land  f \in \pfufa^q(v) \land f \in \delta^{\pfufa^q}(f, v)\}, 
	\end{align*}
	as well as the $\omega$-regular language 
	\begin{align*}
		L = \bigcup_{\{(q,f) \mid q \in \mstates \land f \in \afinal^{\pfufa^q}\}} M_q \cdot N_{q,f}^\omega. 
	\end{align*}
	
	Setting $q$ as the unique final state of an automaton which otherwise equals $\lfufa$ yields an NFA accepting $M_q$. An NFA for $N_{q,f}$ can be obtained as follows. Let $\automaton_{q,q}$ be like $\lfufa$, but with initial and unique final state $q$. Moreover, let $\automaton_{f}$ be like $\pfufa^q$, but with unique final state $f$, and let $\automaton_{f,f}$ be like $\pfufa^q$, but with unique initial and final state $f$. Then an NFA accepting the nonempty words in the intersection of $\accLanguage(\automaton_{q,q}), \accLanguage(\automaton_{f})$ and $\accLanguage(\automaton_{f,f})$ accepts precisely $N_{q,f}$, and is of size $\mathcal{O}(nk^2)$. Taking the $\omega$-closure of this NFA, which requires at most one additional state \cite[Lemma 4.34]{PoMC}, results in an NBA for $N_{q,f}^\omega$. Constructing an automaton for the concatenation of $M_q$ and $N_{q,f}^\omega$ yields an NBA $\automaton_{q,f}$ for $M_q \cdot N_{q,f}^\omega$ of size $\mathcal{O}(n + nk^2)$. By taking the union of all $\automaton_{q,f}$, of which there are at most $n k$, we obtain an NBA $\automaton$ of size $\mathcal{O}(n k \cdot (n + nk^2)) = \mathcal{O}(n^2 k^3)$ for $L$. 
	
	To show the claim, we prove that $\accLanguage(\fufa) = \UP(\accLanguage(\automaton))$. 
	
	\smallskip 
	\noindent
	\textbf{Regarding ``$\boldsymbol{\subseteq}$''.} 
	Let $w = uv^\omega \in \accLanguage(\fufa)$. Then there is a $q \in \mstates$ and a $q$-normalized decomposition $(x,y) \in \llbracket \fufa \rrbracket$ with $w = xy^\omega$. Obviously, $x \in M_q$ and $y \in \accLanguage(\pfufa^q)$. We show the existence of an $f \in \afinal^{\pfufa^q}$ and an $l \in \mathbb{N}$ such that $y^l \in N_{q,f}$.
	
	Because $y \in \accLanguage(\pfufa^q)$ there is a unique $f_1 \in \afinal^{\pfufa^q}$ with $f_1 \in \pfufa^q(y)$. As $\fufa$ is power-unambiguous we have for every $n \geq 1$ that $\afinal^{\pfufa^q} \cap \delta^{\pfufa^q}(f_1, y^n) \neq \emptyset$. In particular, there is (again a unique) $f_2 \in \afinal^{\pfufa^q}$ with $f_2 \in \delta^{\pfufa^q}(f_1, y) \cap \pfufa^q(y^2)$. 
	
	Repeating this argument yields a sequence of final states $f_1, f_2, \ldots$ such that the unique accepting run of $\pfufa^q$ on $y^n$ ends in $f_n$, and is in $f_m$ after $y^m$ for $m<n$. Because $\pfufa^q$ is finite, there must be $j, k \in \mathbb{N}$ with $j < k$ for which $f_j = f_k = f^*$. W.l.o.g. let $j$ and $k$ be minimal among all indices with these properties. It follows that the accepting runs of $\pfufa^q$ on $y^m$, $m \geq j$, repeat a cycle of length $\vert y \vert \cdot n$ \cite{FDFAAORL}.
	
	Let $l = j \cdot n$. Because $l \geq j$ the unique accepting run of $\pfufa^q$ on $y^l$ ends in some $f \in \{f_j, f_{j+1}, \ldots, f_{k-1}\}$. As argued above, the power-unambiguity of $\fufa$  implies that the unique accepting run of $\pfufa^q$ on $y^{2l}$ is in $f$ after $y^l$, and afterwards closes the aforementioned cycle a total of $j$ times (as $l = j \cdot n$), ending in $f$. Therefore, $f \in \pfufa^q(y^l)$ and $f \in \delta^{\pfufa^q}(f, y^l)$. Moreover, since $q \in \delta^\lfufa(q, y)$ as $(x,y)$ is $q$-normalized, also $q \in \delta^{\lfufa}(q, y^l)$. Hence, $y^l \in N_{q,f}$. 
	
	All in all, we have shown that $x \in M_q$ and $y^l \in N_{q,f}$, and so $xy^\omega = x(y^l)^\omega \in M_q \cdot N_{q,f}^\omega$. This implies $w = xy^\omega \in \UP(\accLanguage(\automaton))$, and thus $\accLanguage(\fufa) \subseteq \UP(\accLanguage(\automaton))$.

	\smallskip
	\noindent
	\textbf{Regarding ``$\boldsymbol{\supseteq}$''.} Let $w = uv^\omega \in \UP(\accLanguage(\automaton))$. Then there is a $q \in \mstates$ and $f \in \afinal^{\pfufa^q}$ such that $w \in M_q \cdot N_{q,f}^\omega$. We want to apply \Cref{lem:calbrix-lemma} to show that there are $x \in M_q$ and $y \in N_{q,f}$ with $w = xy^\omega$. For this, it must hold that (i) $M_q \cdot N_{q,f}^* = M_q$ and (ii) $N_{q,f}^+ = N_{q,f}$. 
	
	To prove (i), let $v = uv_1v_2 \ldots v_n$ for $u \in M_q$ and $v_1, \ldots, v_n \in N_{q,f}$, $n \geq 0$. Then 
	\begin{center}
		$\lfufa(v) = \lfufa(uv_1v_2 \ldots v_n) \supseteq \delta^{\lfufa}(q, v_1 \ldots v_n) \supseteq \delta^{\lfufa}(q, v_2 \ldots v_n) \supseteq \ldots \supseteq \delta^{\lfufa}(q, v_n) \ni q$
	\end{center}
	since $u \in M_q$ and every $v_i \in N_{q,f}$ can close a loop on $q$. Thus,  $v \in M_q$ and so $M_q \cdot N_{q,f}^* \subseteq M_q$. Since $M_q \subseteq M_q \cdot N_{q,f}^*$ holds trivially, (i) follows.
	
	To prove (ii), let $v = v_1v_2 \ldots v_n$ for $v_i \in N_{q,f}$, $n \geq 1$. Then 
	\begin{center}
		$\delta^{\lfufa}(q, v_1 \ldots v_n) \supseteq \delta^{\lfufa}(q, v_2 \ldots v_n) \supseteq \ldots \supseteq \delta^{\lfufa}(q, v_n) \ni q$
	\end{center}
	because every $v_i \in N_{q,f}$ can close a loop on $q$. Moreover,
	\begin{center}
		$\pfufa^q(v_1 \ldots v_n) \supseteq \delta^{\pfufa^q}(f, v_2 \ldots v_n) \supseteq \ldots \supseteq \delta^{\pfufa^q}(f, v_n) \ni f$
	\end{center}
	as well as 
	\begin{center}
		$ \delta^{\pfufa^q}(f, v_1 \ldots v_n) \supseteq \delta^{\pfufa^q}(f, v_2 \ldots v_n) \supseteq \ldots \supseteq \delta^{\pfufa^q}(f, v_n) \ni f$
	\end{center}
	since for any $v \in N_{q,f}$ both $f \in \pfufa^q(v)$ and $f \in \delta^{\pfufa^q}(f, v)$.  
	Thus, $v \in N_{q,f}$ and so $N_{q,f} = N_{q,f}^+$. 
	
	This proves that the requirements for \Cref{lem:calbrix-lemma} are satisfied. Hence, there are $x \in M_q$ and $y \in N_{q,f}$ with $w = xy^\omega$. By the definitions of $M_q$ and $N_{q,f}$ it holds that $q \in \mstates_x$, $q \in \delta^{\lfufa}(q, y)$ and $y \in \accLanguage(\pfufa^q)$, so $(x,y) \in \llbracket \fufa \rrbracket$. Therefore $w \in \accLanguage(\fufa)$ and $\UP(\accLanguage(\automaton)) \subseteq \accLanguage(\fufa)$. 
\end{proof}

\LemUBAtoFUFA*
\begin{proof}
	Let $\automaton$ be a UBA.
	W.l.o.g. we assume that $\automaton$ is diamond-free, trim, and that all accepting runs of $\automaton$ visit only a single final state.  
	
	We construct an FUFA $\fufa = (\lfufa, \{\pfufa^q\}_{q \in \mstates})$ from $\automaton$ as follows. 
	The leading automaton of $\fufa$ is like $\automaton$, but without final states, and $\mstates = \afinal^{\automaton}$. For every $q \in \mstates$, the progress automaton $\pfufa^q$ satisfies $\accLanguage(\pfufa^q) = \{v \in \Sigma^+ \mid q \in \delta^\automaton(q, v)\}$, i.e., $\pfufa^q$ accepts all non-empty words on which $\automaton$ can close a loop on $q$. The automaton $\pfufa^q$ can be obtained from $\automaton_q$ by adding a copy $q'$ of $q$ to the automaton and making it the single final state of $\pfufa^q$, as well as setting
	\begin{center}
		$\delta^{\pfufa^q}(p,a) = \begin{cases}
			\delta^\automaton(p,a), & \text{if} p \neq q' \text{ and } q \notin \delta^\automaton(p, a) \\
			(\delta^\automaton(p,a) \setminus \{q\}) \cup \{q'\}, & \text{if} p \neq q' \text{ and } q \in \delta^{\automaton}(p,a) \\
			\delta^\automaton(q, a), &\text{if } p = q' \text{ and } q \notin \delta^{\automaton}(q,a) \\
			(\delta^\automaton(q,a) \setminus \{q\}) \cup \{q'\}, & \text{if} p = q' \text{ and } q \in \delta^{\automaton}(q,a)
		\end{cases}$
	\end{center} 
	As any such $\pfufa^q$ has at most $\vert \automaton \vert + 1$ states, $\fufa$ is of size $\vert \fufa \vert  = (\vert \automaton \vert, \mathcal{O}(\vert \automaton \vert))$. 
	
	\smallskip
	
	We prove that $\fufa$ is a saturated, decomposition- and power-unambiguous FUFA with $\accLanguage(\automaton) = \accLanguage_\omega(\fufa)$. 
	
	\smallskip
	\noindent
	\textbf{$\boldsymbol{\fufa}$ is an FUFA.} Let $q \in \mstates$ and assume that $\pfufa^q$ is not unambiguous. Then there is $v \in \accLanguage(\pfufa^q)$ such that $\pfufa^q$ has two different accepting runs $r_1, r_2$ on $v$. Because $\accLanguage(\pfufa^q) = \{z \in \Sigma^+ \mid q \in \delta^{\automaton}(q, z)\}$ this implies that $r_1$ and $r_2$ are two different ways of $\automaton$ to close a loop on the final state $q$ when reading $v$. 
	Now let $u \in \Sigma^*$ such that $q \in \automaton(u)$ and let $r$ be a run of $\automaton$ on $u$ that ends in $q$. Then $rr_1^\omega$ and $rr_2^\omega$ are two different accepting runs of $\automaton$ on $uv^\omega$. This contradicts the unambiguity of $\automaton$. 
	
	\smallskip
	\noindent
	\textbf{$\boldsymbol{\fufa}$ is saturated.} Let $w = uv^\omega$ be an ultimately periodic word, and let $(x, y)$ be a normalized decomposition of $w$. Then there is a $q \in \mstates_x$ such that $(x,y)$ is $q$-normalized. 
	By construction, $\mstates = \afinal^{\automaton}$ and so $q$ corresponds to a final state of $\automaton$. Because $(x,y)$ is $q$-normalized, $q \in \delta^\lfufa(q, y) = \delta^{\automaton}(q, y)$. Thus, $y \in \{z \in \Sigma^+ \mid q \in \delta^{\automaton}(q, z)\} = \accLanguage(\pfufa^{q})$. It follows that $(x,y) \in \llbracket \fufa \rrbracket$.
	Hence, every normalized decomposition of $w$ is accepted by $\fufa$, so $\fufa$ is saturated. 
	
	\smallskip 
	\noindent
	\textbf{$\boldsymbol{\fufa}$ is decomposition-unambiguous.} Let $(u,v) \in \Sigma^{*+}$ and assume that there are $q_1, q_2 \in \mstates_u$ with $q_1 \neq q_2$ such that $(u,v)$ is $q_i$-normalized and $v \in \accLanguage(\pfufa^{q_i})$ for $i \in \{1, 2\}$. Then $q_1, q_2 \in \afinal^{\automaton}$ and there are finite runs $r_1, r_2$ of $\automaton$ on $u$ that end in $q_1$ and $q_2$, respectively. As $v \in \accLanguage(\pfufa^{q_1}) \cap \accLanguage(\pfufa^{q_2})$ it follows from the definition of $\pfufa^{q_i}$, $i \in \{1,2\}$, that $q_1 \in \delta^\automaton(q_1, v)$ and $q_2 \in \delta^\automaton(q_2, v)$. Hence, there are (finite) runs $r_1'$ and $r_2'$ of $\automaton$ on $v$ such that $r_i'$ closes a loop on $q_i$. But then $\pi_1 = r_1(r_1')^\omega$ and $\pi_2 = r_2(r_2')^\omega$ are two different accepting runs of $\automaton$ on $uv^\omega$.
	This contradicts the unambiguity of $\automaton$. 
	
	\smallskip 
	\noindent
	\textbf{$\fufa$ is power-unambiguous.} 
	Let $(u,v) \in \llbracket \fufa \rrbracket$ with $q$ the unique (since $\fufa$ is decomposition-unambiguous) state in $\mstates_u$ with $q \in \delta^\lfufa(q, v)$ and $v \in \accLanguage(\pfufa^q)$. As $v \in \accLanguage(\pfufa^q)$ there is a run of $\pfufa^q$ on $v$ that ends in the unique final state $q'$ of $\pfufa^q$. Because $\accLanguage(\pfufa^q) = \{z \in \Sigma^+ \mid q \in \delta^\automaton(q, z)\}$ we have $q \in \delta^\automaton(q, v)$, so also $q \in \delta^\automaton(q, v^n)$ and $v^n \in \accLanguage(\pfufa^q)$ for all $n \geq 1$. By construction, for every run of $\automaton$ starting in $q$ there is a run of $\pfufa^q$ starting in $q'$ visiting the same states, with the only difference being that every $q$ is replaced by $q'$. Since $v$ closes a loop on $q$ in $\automaton$, it hence also closes a loop on $q'$ in $\pfufa^q$. It follows that $q' \in \delta^{\pfufa^q}(q', v^n)$ and thus $\delta^{\pfufa^q}(q', v^n) \cap \afinal^{\pfufa^q} \neq \emptyset$ for all $n \geq 1$. Therefore, $\fufa$ is power-unambiguous. 
	
	\smallskip 
	\noindent
	\textbf{$\boldsymbol{\accLanguage_\omega(\fufa) = \accLanguage(\automaton)}$.}  We show $\UP(\accLanguage(\automaton))= \UP(\fufa)$. As $\accLanguage(\automaton)$ is $\omega$-regular this implies the claim.  
	
	Let $w = uv^\omega \in \UP(\accLanguage(\automaton))$. Then there is a unique accepting run $\pi$ of $\automaton$ on $w$. Since $\automaton$ is finite, there is some $f \in \afinal^\automaton$ visited infinitely often along $\pi$. Because $w$ is ultimately periodic and $\vert v \vert < \infty$, there are $x \in \Sigma^*$ and $y \in \Sigma^+$ with $xy^\omega = w$ and such that, along $\pi$, after reading $x$ the final state $f$ is reached, and a loop on $f$ is closed on $y$. But then $(x,y) \in \Sigma^{*+}$ satisfies $f \in \mstates_x$, $f \in \delta^{\lfufa}(f, y) = \delta^{\automaton}(f,y)$ and $y \in \accLanguage(\pfufa^q)$, so $(x,y) \in \llbracket \fufa \rrbracket$ and $w \in \accLanguage(\fufa)$.

	\smallskip 
	Now let $w = uv^\omega \in \accLanguage(\fufa)$. Then there is some normalized decomposition $(x,y)$ of $w$ accepted by $\fufa$. Hence, there is a $q \in \mstates_x$ such that $q \in \delta^{\lfufa}(q, y)$ and $y \in \accLanguage(\pfufa^q)$. By construction, $q \in \mstates$ corresponds to a final state of $\automaton$, and since $\accLanguage(\pfufa^q) = \{z \in \Sigma^+ \mid q \in \delta^{\automaton}(q, z)\}$ it follows that $\automaton$ can close a loop $r_2$ on $q$ when reading $y$. Hence, for $r_1$ a finite run of $\automaton$ on $x$ that ends in $q$, the run $\pi = r_1 r_2^\omega$ of $\automaton$ on $xy^\omega$ visits $q \in \afinal^{\automaton}$ infinitely often, i.e., $w = xy^\omega \in \UP(\accLanguage(\automaton))$.
\end{proof}

\CorSingleExponentialLTLtoFUFA*
\begin{proof}
	Given $\varphi \in \mathrm{LTL}$ one can construct a UBA $\automaton_\varphi$ with number of states exponential in the length of $\varphi$ and such that $\accLanguage(\automaton_\varphi) = \Words(\varphi)$ \cite{ATAAPV}. By \Cref{lem:uba-to-fufa}, $\automaton_\varphi$ can be translated into a saturated, decomposition- and power-unambiguous FUFA $\fufa$ of size polynomial in $\vert \automaton_\varphi\vert$ and such that $\accLanguage_\omega(\fufa) = \accLanguage(\automaton_\varphi) = \Words(\varphi)$. Hence, there is also a single-exponential translation from LTL to saturated, decomposition- and power-unambiguous FUFA. 
\end{proof}

\ThmExpressivePowerFUFA*
\begin{proof}
	The claim follows immediately when combining \Cref{lem:uba-to-fufa,lem:fufa-to-nba-polynomial}.
\end{proof}

 \LemMembershipFUFA*
\begin{proof}
	To decide if $(u,v) \in \llbracket \fufa \rrbracket$ one has to check if there is a $q \in \mstates_u$ such that $(u,v)$ is $q$-normalized and $v \in \accLanguage(\pfufa^q)$. Computing the set $\mstates_u$ can be done by running $\lfufa$ on $u$ and returning the intersection of $\lfufa(u)$ with $\mstates$, which is possible in time polynomial in $\vert \lfufa \vert$ and $\vert u \vert$. Afterwards, checking for every $q_i \in \mstates_u$ if the word $v$ is accepted by the automaton obtained from $\lfufa$ by making $q_i$ the unique initial and final state yields those $q \in \mstates_u$ for which $(u,v)$ is $q$-normalized. Each of these checks is possible in time polynomial in $\vert \lfufa \vert$ and $\vert v \vert$, and at most $\vert \lfufa \vert$ many such checks are necessary. Lastly, for the states $q \in \mstates_u$ for which $(u,v)$ is $q$-normalized, one can decide if $v \in \accLanguage(\pfufa^q)$ by checking membership of $v$ in the language of $\pfufa^q$. This is possible in time polynomial in $\vert \pfufa^q \vert$ and $\vert v \vert$. Overall, this yields an algorithm to check if a given $(u,v)$ is in $\llbracket \fufa \rrbracket$ that runs in time polynomial in $\vert \fufa \vert, \vert u \vert$ and $\vert v \vert$. 
\end{proof}

\LemIntersectionFUFAViaDecompositionSets*
\begin{proof}
	Let $\fufa_1$ and $\fufa_2$ be as in the lemma, and let $\mstates$ be the Cartesian product of $\mstates_1$ and $\mstates_2$. For $(q_1, q_2) \in \mstates$, the \emph{parallel product} of $\pfufa_1^{q_1}$ and $\pfufa_2^{q_2}$ is defined as
	\begin{center}
		$\pfufa^{q_1, q_2} = \pfufa_1^{q_1} \times \pfufa_2^{q_2} = (\astates^{\pfufa^{q_1}_1} \times \astates^{\pfufa^{q_2}_2}, \Sigma, \delta^{q_1, q_2}, \langle \ainit^{\pfufa_1^{q_1}}, \ainit^{\pfufa_2^{q_2}} \rangle, \afinal^{\pfufa_1^{q_1}} \times \afinal^{\pfufa_2^{q_2}} )$,
	\end{center} 
	where for all $\langle p_1, p_2 \rangle \in \pfufa^{q_1, q_2}$ and $a \in \Sigma$ it holds that
	\begin{center}
		$\delta^{q_1, q_2}(\langle p_1, p_2 \rangle, a) = \{ \langle p_1', p_2' \rangle \mid p_1' \in \delta^{\pfufa_1^{q_1}}(p_1, a) \text{ and} p_2' \in \delta^{\pfufa_2^{q_2}}(p_2, a) \}$.
	\end{center}
	In particular, we have that $\accLanguage(\pfufa^{q_1, q_2})= \accLanguage(\pfufa_1^{q_1}) \cap \accLanguage(\pfufa_2^{q_2})$. 
	
	Let $\lfufa = \lfufa_1 \times \lfufa_2$ be the parallel product of $\lfufa_1$ and $\lfufa_2$ (which is build similar to $\pfufa^{q_1, q_2}$). We define $\fufaH = (\lfufa, \{\pfufa^q\}_{q \in \mstates})$. Then (1) $\fufaH$ is an FUFA, i.e., all progress automata $\pfufa^{q_1, q_2}$ are UFA, and (2) $\llbracket \fufaH \rrbracket = \llbracket \fufa_1 \rrbracket \cap \llbracket \fufa_2 \rrbracket$. Moreover, if both $\fufa_1$ and $\fufa_2$ are (3) decomposition-unambiguous, (4) power-unambiguous or (5) saturated, then $\fufaH$ is as well. 
	
	\smallskip
	\noindent
	\textbf{Regarding (1).} Let $\langle q_1, q_2 \rangle \in \mstates$ and $w \in \Sigma^*$ such that $\pfufa^{q_1, q_2}$ has two accepting runs $\pi_1, \pi_2$ on $w$. By construction, a run of $\pfufa^{q_1, q_2}$ is accepting iff the projections to its first and second component are accepting runs in $\pfufa_1^{q_1}$ and $\pfufa_2^{q_2}$, respectively. Hence, both $\pi_1$ and $\pi_2$ induce accepting runs of $\pfufa_1^{q_1}$ and $\pfufa_2^{q_2}$ on $w$. Since $\pfufa_1^{q_1}$ and $\pfufa_2^{q_2}$ are unambiguous, it follows that the projections of $\pi_1$ and $\pi_2$ to their first resp. second component must be equal. Therefore, $\pi_1 = \pi_2$ and it follows that $\pfufa^{q_1, q_2}$ is unambiguous. 
	
	\smallskip
	\noindent
	\textbf{Regarding (2).} We start by showing ``$\subseteq$''. Let $(u,v) \in \llbracket \fufaH \rrbracket$. Then there is a $\langle q_1, q_2 \rangle \in \mstates$ reachable on $u$ in $\lfdfa$ such that $(u,v)$ is $\langle q_1, q_2 \rangle$-normalized and $v \in \accLanguage(\pfufa^{q_1, q_2})$.
	
	By the latter there is a run $\pi$ of $\pfufa^{q_1, q_2}$ on $v$ that ends in a state $\langle f_1, f_2 \rangle \in \afinal^{\pfufa_1^{q_1}} \times \afinal^{\pfufa_2^{q_2}}$. The projections of $\pi$ onto its first and second component are, respectively, runs $\pi_1$ and $\pi_2$ of $\pfufa_1^{q_1}$ and $\pfufa_2^{q_2}$ on $v$ that end in $f_1$ and $f_2$. Since $\langle q_1, q_2 \rangle \in \mstates$ iff $q_1 \in \mstates_1$ and $q_2 \in \mstates_2$, it follows that $v \in \accLanguage(\pfufa_1^{q_1})$ for a $q_1 \in \lfufa_1(u) \cap \mstates_1$, and that $v \in \accLanguage(\pfufa_2^{q_2})$ for a $q_2 \in \lfufa_2(u) \cap \mstates_2$. Lastly, because $(u,v)$ is normalized w.r.t. $\langle q_1, q_2 \rangle$ it holds that $\langle q_1, q_2 \rangle \in \delta^{\lfufa}(\langle q_1, q_2 \rangle, v)$, which implies $q_1 \in \delta^{\lfufa_1}(q_1, v)$ and $q_2 \in \delta^{\lfufa_2}(q_2, v)$. All in all, $(u,v)$ is $q_i$-normalized and $v \in \accLanguage(\pfufa_i^{q_i})$ for $i \in \{1,2\}$, so $(u,v) \in \llbracket \fufa_1 \rrbracket$ and $(u,v) \in \llbracket \fufa_2 \rrbracket$, i.e., $(u,v) \in \llbracket \fufa_1 \rrbracket \cap \llbracket \fufa_2 \rrbracket$. 
	
	\smallskip
	To show ``$\supseteq$'', let $(u,v) \in \llbracket \fufa_1 \rrbracket \cap \llbracket \fufa_2 \rrbracket$. Then, for $i \in \{1,2\}$, there are $q_i \in \lfufa_i(u) \cap \mstates_i$ such that $(u,v)$ is $q_i$-normalized and $v \in \accLanguage(\pfufa_i^{q_i})$. Hence, there are runs of $\pfufa_i^{q_i}$ on $v$ that end in states $f_i \in \afinal^{\pfufa_i^{q_i}}$. It follows that $\langle q_1, q_2 \rangle \in \lfufa(u) \cap \mstates$ and $\langle f_1, f_2 \rangle \in \pfufa^{q_1, q_2}(v) \cap \afinal^{\pfufa^{q_1, q_2}}$. As $(u,v)$ is $q_i$-normalized for $i \in \{1,2\}$, we also have $\langle q_1, q_2 \rangle \in \delta^{\lfufa}(\langle q_1, q_2 \rangle, v)$. All in all, $(u,v)$ is $\langle q_1, q_2 \rangle$-normalized for $\langle q_1, q_2 \rangle \in \mstates_u$ and $v \in \accLanguage(\pfufa^{q_1, q_2})$, so $(u,v) \in \llbracket \fufaH \rrbracket$. 
	
	\smallskip
	\noindent
	\textbf{Regarding (3).} Let $(u,v) \in \Sigma^{*+}$ and assume that there are $\langle p_1, p_2 \rangle, \langle q_1, q_2 \rangle \in \mstates_u$ for which $(u,v)$ is normalized and such that $v \in \accLanguage(\pfufa^{p_1, p_2}) \cap \accLanguage(\pfufa^{q_1, q_2})$. By definition, $v$ is accepted by $\pfufa^{x,y}$ for some $\langle x, y \rangle \in \mstates$ iff $v \in \accLanguage (\pfufa_1^{x}) \cap \accLanguage(\pfufa_2^{y})$. Moreover, $\langle x,y \rangle$ is reachable on $u$ in $\lfufa$ iff $x \in \lfufa_1(u)$ and $y \in \lfufa_2(u)$, $(u,v)$ is normalized w.r.t. $\langle x,y \rangle$ if $(u,v)$ is $x$-normalized in $\lfufa_1$ and $y$-normalized in $\lfufa_2$, and $\langle x, y \rangle \in \mstates$ iff $x \in \mstates_1$ and $y \in \mstates_2$. Applied to our setting it follows that there are two states in $\lfufa_i(u) \cap \mstates_i$, $i \in \{1,2\}$, for which $(u,v)$ is normalized and whose corresponding progress automata accept $v$, namely $p_i$ and $q_i$. This contradicts that $\fufa_1$ and $\fufa_2$ are decomposition-unambiguous. 
	
	\smallskip
	\noindent
	\textbf{Regarding (4).} Let $(u,v) \in \llbracket \fufaH \rrbracket$ and let $q = \langle q_1, q_2 \rangle \in \mstates_u$ such that $(u,v)$ is $q$-normalized and $v \in \accLanguage(\pfufa^q)$. Let $\pfufa^q(v) \cap \afinal^{\pfufa^q} = \{f\}$ with $f = \langle f_1, f_2 \rangle$. To show that $\fufaH$ is power-unambiguous we must prove that $\afinal^{\pfufa^q} \cap \delta^{\pfufa^q}(f, v^n) \neq \emptyset$ for all $n \geq 1$. 
	
	By definition, $v \in \accLanguage(\pfufa^q)$ iff $v \in \accLanguage(\pfufa_1^{q_1}) \cap \accLanguage(\pfufa_2^{q_2})$. Since both $\fufa_1$ and $\fufa_2$ are power-unambiguous, it follows that $\afinal^{\pfufa_i^{q_i}} \cap \delta^{\pfufa^{q_i}}(f_i, v^n) \neq \emptyset$ for all $n \geq 1$ and $i \in \{1,2\}$. But then also $\afinal^{\pfufa^{q}} \cap \delta^{\pfufa^{q}}(f, v^n) \neq \emptyset$ for all $n \geq 1$. Hence, $\fufaH$ is power-unambiguous. 
	
	\smallskip
	\noindent
	\textbf{Regarding (5).} Let $w$ be an ultimately periodic word and let $(u_1, v_1), (u_2, v_2)$ be two normalized decompositions of $w$. Moreover, let $(u_1, v_1) \in \llbracket \fufaH \rrbracket$. We show that $(u_2, v_2) \in \llbracket \fufaH \rrbracket$. 
	
	Since $(u_1, v_1) \in \llbracket \fufaH \rrbracket$ it follows from  $\llbracket \fufaH \rrbracket = \llbracket \fufa_1 \rrbracket \cap \llbracket \fufa_2 \rrbracket$ that $(u_1, v_1) \in \llbracket \fufa_1 \rrbracket$ and $(u_1, v_1) \in \llbracket \fufa_2 \rrbracket$. As $\fufa_i$, $i \in \{1,2\}$, is saturated the (again normalized) decomposition $(u_2, v_2)$ must also be accepted by both $\fufa_i$, i.e., $(u_2, v_2) \in \llbracket \fufa_i \rrbracket$. But then $(u_2, v_2) \in \llbracket \fufa_1 \rrbracket \cap \llbracket \fufa_2 \rrbracket = \llbracket \fufaH \rrbracket$. 
\end{proof}
\section{\texorpdfstring{Proofs of \Cref{sec:dfufa-mc}}{Proofs of Section 5}}\label{app:proofs-dfufa-mc}
\LemBSCCWithoutAState*
\begin{proof}
	Let $\automaton$ be a DRA with $\accLanguage(\automaton) = E$, and let $\langle s, a \rangle$ be the state reached in $\dtmc \otimes \automaton$ when lifting $u$ to this DTMC. We prove that all BSCCs in $\dtmc \otimes \automaton$ reachable from $\langle s, a \rangle$ are rejecting, which directly implies the claim. 
	
	Towards a contradiction assume the opposite, i.e., that there is an accepting BSCC $B_\automaton$ of $\dtmc \otimes \automaton$ reachable on some $v$ from $\langle s, a \rangle$. Because the lifting of $u$ to $\dtmc \otimes \lfufa$ reaches a BSCC, the DTMC $\dtmc$ must be in a BSCC after $u$ as well. Hence, w.l.o.g. we assume that $v$ is such that the state reached in $\dtmc \otimes \automaton$ on $uv$ is of the form $\langle s, a' \rangle$ for some $a' \in \automaton$. Since after lifting $uv$ to $\dtmc \otimes \automaton$ the BSCC $B_\automaton$ is reached, there is a finite word $z$ that closes a cycle on $\langle s, a' \rangle$ and visits all states of $B_\automaton$ at least once. As $B_\automaton$ is accepting, \Cref{lem:accepting-and-rejecting-bscc-dra-product} yields that $uvz^\omega \in \accLanguage(\automaton) = E$, so also $uvz^\omega \in \accLanguage(\fufa)$.
	
	It follows that there are $x ,y$ such that $xy^\omega = uvz^\omega$ and $(x,y) \in \llbracket \fufa \rrbracket$. Hence, $(x,y)$ is $q$-normalized for $q = \lfufa(x)$ and $y \in \accLanguage(\pfufa^q)$. In particular, $q \in \mstates$. We consider two cases. 
	
	If $u$ is a prefix of $x$ then the facts that $\lfufa$ is deterministic and that $\lfufa(x) = q$ imply that there is a state $\langle s', q \rangle \in B_\lfufa$, for $s'$ the state of $\dtmc$ reached after $x$. As $q \in \mstates$ this contradicts the assumption that $B_\lfufa$ does not contain any state whose second component is in $\mstates$. 
	
	If $x$ is a prefix of $u$ then, as $xy^\omega = uvz^\omega$, there must be a $k \in \mathbb{N}$ such that $u$ is a prefix of $xy^k$, i.e., such that when following $xy^k$ the BSCC $B_\lfufa$ is entered. As $(x,y)$ is $q$-normalized we have $q = \delta^{\lfufa}(q, y)$ and so also $\lfufa(xy^k) = \delta^{\lfufa}(q, y^k) = q$. But then again there is a state $\langle s'', q \rangle \in B$, for $s''$ the state reached in $\dtmc$ upon $xy^k$, with $q \in \mstates$, a contradiction. 
	
	Hence, no accepting BSCC in $\dtmc \otimes \automaton$ can be reachable from $\langle s, a \rangle$ and so the probability of $\dtmc$ to generate a path in $E$ with prefix $u$ equals $0$.  
\end{proof}

\LemCharacterizationGoodBSCCinDFUFA*
\begin{proof}
	Let $\automatonD(s,q)$ be a DRA such that $\accLanguage(\automatonD(s,q)) = \accLanguage(\automatonB(s,q))$. Then the lemma is equivalent to the claim that $B_\automaton$ is accepting iff $\probMeasure^{\dtmc_s}(\accLanguage(\automatonD(s,q))) = 1$, which we prove in the following.
	
	\smallskip
	\noindent
	\textbf{Regarding ``$\boldsymbol{\Rightarrow}$''.} We use contraposition, i.e., we show that $\probMeasure^{\dtmc_s}(\accLanguage(\automatonD(s,q))) < 1$ implies that $B_\automaton$ is rejecting. Let $a$ be the state reached in the DRA $\automaton$ for $E$ upon reading $u$, and consider the product model $G = \dtmc_s \times \lfufa_q \times \automaton_a \times \automatonD(s,q)$. Since $\lfufa_q, \automaton_a$ and $\automatonD(s,q)$ are deterministic and (w.l.o.g.) total, $G$ is a DTMC. 
	
	Because $\probMeasure^{\dtmc_s}(\accLanguage(\automatonD(s,q))) < 1$ there is a $x \in \Paths^*(\dtmc_s)$ such that a rejecting BSCC $B_\automatonD$ is reachable in the DTMC  $\dtmc_s \times \automatonD(s,q)$. Because $G$ is a DTMC, w.l.o.g. we can choose $x$ such that, upon lifting $x$ to $G$, a BSCC $B_G$ is reached. Moreover, the BSCC $B_\lfufa$ of $\dtmc \otimes \lfufa$ is entered after following $u$, so we can additionally assume that the state reached in $G$ upon $x$ is of the form $\langle s, q, a', d \rangle$ for some $a' \in \automaton$ and $d \in \automatonD(s,q)$. By construction, it must hold for any $\langle s', q', a', d' \rangle \in B_G$ that $\langle s', q' \rangle \in B_{\lfufa}$, $\langle s', a' \rangle \in B_{\automaton}$ and $\langle s', d' \rangle \in B_\automatonD$. 
	
	Now let $z \in \Paths^*(\dtmc_s)$ be such that (1) $z$ closes a loop on $\langle s, q \rangle$ in $\dtmc \otimes \lfufa$, (2) all states of $B_\automaton$ are visited at least once when following $z$ from \emph{any} state $\langle s, a \rangle$ in $B_\automaton$, and (3) all states of $B_\automatonD$ are visited at least once when following $z$ from \emph{any} state $\langle s, d \rangle$ in $B_\automatonD$.
	
	Such a $z$ can be obtained as follows. Let $\langle s, a_1 \rangle, \ldots, \langle s, a_m \rangle$ and $\langle s, d_1 \rangle, \ldots, \langle s, d_n \rangle$ be arbitrary orderings of the states in $B_\automaton$ and $B_\automatonD$, respectively, that have $s$ in their first component. 
	We start with the states $\langle s, a_i \rangle$ in $B_\automaton$ and choose $x_1$ to be such that upon following $x_1$ from $\langle s, a_1 \rangle$, all states in $B_\automaton$ are visited at least once. Since $B_\automaton$ is a BSCC of $\dtmc \otimes \automaton$, such an $x_1$ is guaranteed to exist. Let $\langle s_1^*, q_1^* \rangle$ be the state reached in $\dtmc \otimes \lfufa$ when following $x_1$ from $\langle s, q \rangle$. Since $\langle s, q \rangle \in B_\lfufa$ there is a finite $y_1$ such that $z_1 = x_1y_1$ closes a loop on $\langle s, q \rangle$ in $\dtmc \otimes \lfufa$. Hence, $z_1$ is such that its lifting to $\dtmc \otimes \automaton$ from $\langle s, a_1 \rangle$ visits all states of $B_\automaton$ at least once and simultaneously closes a loop on $\langle s, q \rangle$ in $\dtmc \otimes \lfufa$. 
	
	Now consider the lifting of $z_1$ to $\dtmc \otimes \automaton$ on $z_1$ starting in $\langle s, a_2 \rangle$. If, along this finite path, all states of $B_\automaton$ are visited at least once we set $z_2 = \varepsilon$. Otherwise, we set $z_2 = x_2 y_2$ such that, upon moving according to $z_1x_2$ from $\langle s, a_2 \rangle$ all states of $B_\automaton$ are visited at least once, and such that $x_2y_2$ closes a loop on $\langle s, q \rangle$ in $\dtmc \otimes \lfufa$. Then it is ensured that, upon following $z_1z_2$ from both $\langle s, a_1 \rangle$ and $\langle s, a_2 \rangle$, all states of $B_\automaton$ are visited at least once and a loop on $\langle s, q \rangle$ is closed. Repeating this procedure yields words $z_i$, $1 \leq i \leq m$, such that $z_1z_2 \ldots z_m$ satisfies (1) and (2). Afterwards, we turn our attention to $B_\automatonD$ and follow a similar approach to find $z_1', \ldots, z_n'$ such that, upon following $z_1 \ldots z_mz_1' \ldots z_i'$ from $\langle s, d_i \rangle \in B_\automatonD$ all states of $B_\automatonD$ are visited at least once and a loop on $\langle s, q \rangle$ is closed. In the end, setting $z = z_1 \ldots z_m z_1' \ldots z_n'$ ensures that (1), (2) and (3) hold. 
	
	Let now $v = xz$. Then both $x$ and $z$ close loops on $\langle s, q \rangle$ in $\dtmc \otimes \lfufa$, so $v$ does as well. Following powers of $v$ in $\dtmc_s \times \automaton_a$ yields
	\begin{align}
		\langle s, a \rangle \overset{v}{\longrightarrow} \langle s, a_1 \rangle \overset{v}{\longrightarrow} \langle s, a_2 \rangle \overset{v}{\longrightarrow} \ldots \overset{v}{\longrightarrow} \langle s, a^* \rangle \overset{v}{\longrightarrow} \ldots \overset{v}{\longrightarrow} \langle s, a^* \rangle, \label{eq:runs-on-powers-of-cycle-word}
	\end{align}
	where the finiteness of $\automaton$ implies that there are $j_1, k_1 \in \mathbb{N}$ with $j_1 < k_1$ such that the state $\langle s, a^* \rangle$ reached after $v^{j_1}$ and $v^{k_1}$ is the same. 
	Similarly, there are $j_2, k_2 \in \mathbb{N}$ with $j_2 < k_2$ and such that the state $\langle s, d^* \rangle$ reached in $\dtmc_s \times \automatonD(s,q)$ after $v^{j_2}$ and $v^{k_2}$ is the same. 
	
	Now let $n_1 = k_1 - j_1$, $n_2 = k_2 - j_2$, and $n = j_1 \cdot j_2 \cdot n_1 \cdot n_2$. Then $n \geq j_1, n \geq j_2$ and $n$ is divisible by both $n_1$ and $n_2$. It follows from arguments in \cite{FDFAAORL} that the state $\langle s, \overline{a} \rangle$ reached in $\dtmc_s \times \automaton_a$ when moving according to $v^n$ is also reached upon $v^{2n}$, i.e., that $v^n$ closes a loop on $\langle s, \overline{a} \rangle$ in this DTMC. The same holds for the state $\langle s, \overline{d} \rangle$ reached in $\dtmc_s \times \automatonD(s,q)$ on $v^n$. 
	
	Since $v = xz$ starts with $x$, the state $\langle s, \overline{d} \rangle$ is in the rejecting BSCC $B_\automatonD$ of $\dtmc_s \times \automatonD(s,q)$. Moreover, by the construction of $z$ it follows that all states of $B_\automatonD$ are visited at least once when following $v^n$ from $\langle s, \overline{d} \rangle$. Because $v^n$ additionally closes a cycle on $\langle s, \overline{d} \rangle$, \Cref{lem:accepting-and-rejecting-bscc-dra-product} implies that $(v^n)^\omega = v^\omega \notin \accLanguage(\automatonD(s,q)) = \accLanguage(\automatonB(s,q))$. Hence, no run of $I_{\langle s, q \rangle}$ on $v^\omega$ visits infinitely many states $\langle s, q, f \rangle$ with $f \in \afinal^{\pfufa^q}$. 
	
	But then it follows that $\pfufa^q(v) \cap \afinal^{\pfufa^q} = \emptyset$. To see this assume the opposite, i.e., that there is a $f \in \afinal^{\pfufa^q}$ with $f \in \pfufa^q(v)$. Because $\fufa$ is power-unambiguous, $\afinal^{\pfufa^q} \cap \delta^{\pfufa^q}(f, v^l) \neq \emptyset$ for all $l \geq 1$. As $v$ is constructed such that it closes a loop on $\langle s, q \rangle$ in $\dtmc \otimes \lfufa$, there is a lifting of $v^\omega$ to $I_{\langle s, q \rangle}$ that, after every repetition of $v$, is in a state of the form $\langle s, q, f_i \rangle$ for some $f_i \in \afinal^{\pfufa^q}$. But then $v^\omega \in  \accLanguage(\automatonD(s,q))$, a contradiction. 
	
	Therefore, $\pfufa^q(v) \cap \afinal^{\pfufa^q} = \emptyset$ and so $v \notin \accLanguage(\pfufa^q)$. Because $v$ closes a loop on $q$ in $\lfufa$, $(u,v)$ is a $q$-normalized decomposition of $uv^\omega$ that is not accepted by $\fufa$, so $uv^\omega \notin E$ because $\fufa$ is saturated.
	From \Cref{lem:accepting-and-rejecting-bscc-dra-product} it then follows that $B_\automaton$ must be rejecting, since $v^n$ closes a cycle on the state $\langle s, \overline{a} \rangle$ in $B_\automaton$ reached after lifting $uv^n$ to $\dtmc \otimes \automaton$ that visits all states of $B_\automaton$ at least once, and $uv^n(v^n)^\omega = uv^\omega \notin E = \accLanguage(\automaton)$. 
	
	Therefore, if $\probMeasure^{\dtmc_s}(\accLanguage(\automatonD(s,q))) < 1$ then $B_\automaton$ must be rejecting. Thus, if $B_\automaton$ is accepting then $\probMeasure^{\dtmc_s}(\accLanguage(\automatonD(s,q))) \geq 1$. Because $\probMeasure^{\dtmc}(\accLanguage(\automatonD(s,q))) \leq 1$ as a probability, the claim follows.
	
	\smallskip
	\noindent
	\textbf{Regarding ``$\boldsymbol{\Leftarrow}$''.} Let $\probMeasure^{\dtmc_s}(\accLanguage(\automatonD(s,q))) = 1$. We prove that $B_\automaton$ must be accepting. 
	
	From $\probMeasure^{\dtmc_s}(\accLanguage(\automatonD(s,q))) = 1$ it follows that for every $x \in \Paths^*(\dtmc_s)$ there is an extension $y$ such that, upon reading $xy$, a state $\langle s, q, f \rangle$ for an $f \in \afinal^{\pfufa^q}$ can be reached in $I_{\langle s, q \rangle}$. If this would not be the case then there is a $x \in \Paths^*(\dtmc)$ with $Cyl(x) \cap \accLanguage(\automatonD(s,q)) = \emptyset$. But then $\probMeasure^{\dtmc_s}(\accLanguage(\automatonD(s,q))) \leq 1 - \probMeasure(x) < 1$, a contradiction to our assumption. 
	
	Now consider the DTMC $G = \dtmc \otimes \lfufa \otimes \automaton$ and let $v$ be a finite extension of $u$ such that the lifting of $uv$ to $G$ reaches some BSCC $B_G$. Note that $\dtmc \otimes \automaton$ and $\dtmc \otimes \lfufa$ enter the BSCCs $B_\automaton$ and $B_\lfufa$ after $u$, and so stay there along $v$. W.l.o.g. we can thus assume that the lifting of $uv$ to $G$ ends in a state $\langle s, q, a \rangle \in B_G$ for some $a \in \automaton$, i.e., that the run of $\dtmc \otimes \lfufa$ on $uv$ ends in $\langle s, q \rangle$. Because all considered automata are deterministic and (w.l.o.g.) total, and since we are in a BSCC of $G$ after following $uv$, there is a finite $x$ such that (1) $x$ closes a loop on $\langle s, q \rangle$ in $\dtmc \otimes \lfufa$ and (2) upon reading $x$ from any state $\langle s, a' \rangle \in B_\automaton$, all state of $B_\automaton$ are visited at least once. Such a $x$ can be constructed by using a procedure similar to the one described in the first part of the proof. As argued above, because $\probMeasure^{\dtmc_s}(\accLanguage(\automatonD(s,q))) = 1$ there is an extension $y$ of $x$ such that there is a lifting of $xy$ to $I_{\langle s, q \rangle}$ that ends in some $\langle s, q, f \rangle$ with $f \in \afinal^{\pfufa^q}$. Let $z = xy$. 
	
	Similar to \Cref{eq:runs-on-powers-of-cycle-word}, the runs starting in $\langle s, a \rangle$ of $\dtmc \otimes \automaton$ on powers of $z$ must, for some $j, k \in \mathbb{N}$ with $j < k$, end in a common state $\langle s, a^* \rangle$ because $\automaton$ is finite. Let $n = k-j \geq 1$ and $m = n \cdot j$. Then $m \geq j$ and $m$ is divisible by $n$, so the determinism of $\automaton$ implies that $z^m$ closes a loop on the state $\langle s, \overline a \rangle$ reached from $\langle s, a \rangle$  on $z^m$. Moreover, since $m \geq 1$ and $z^m$ starts with $x$, a run of $\dtmc \otimes \automaton$ on $z^m$ starting in $\langle s, \overline{a} \rangle$ must visit all states of $B_\automaton$ at least once. By \Cref{lem:accepting-and-rejecting-bscc-dra-product} it follows that if $uvz^\omega \in E = \accLanguage(\automaton)$ then $B_\automaton$ is accepting. 
	
	To see that this is indeed the case observe that, by the choice of $z$, there is a lifting of $z$ to $I_{\langle s, q \rangle}$ that ends in a state $\langle s, q, f \rangle$ with $f \in \afinal^{\pfufa^q}$. By projecting to the third component we obtain a run of $\pfufa^q$ on $z$ that ends in $f$, so $z \in \accLanguage(\pfufa^q)$. Moreover, $z$ closes a loop on $\langle s, q \rangle$ in $\dtmc \otimes \lfufa$, and so $q = \delta^\lfufa(q, z)$. Because the lifting of $uv$ to $\dtmc \otimes \lfufa$ ends in $\langle s, q \rangle$, we additionally have that $q = \lfufa(uv)$. Hence,  $uvz^\omega \in \accLanguage(\fufa) \subseteq E$, finishing the proof.
\end{proof}

\LemAlmostSurelyForAllOrNoneFUFA*
\begin{proof}
	Let $u$ be a path whose lifting to $\dtmc \otimes \lfufa$ ends in $\langle s', q' \rangle$. Moreover, let $v$ be an extension of $u$ such that the lifting of $uv$ to $\dtmc \otimes \automaton$ ends in a BSCC $B_\automaton$ of $\dtmc \otimes \automaton$ and closes a loop on $\langle s', q' \rangle$ in $\dtmc \otimes \lfufa$. This is always possible since $\dtmc \otimes \automaton$ is a DTMC, so a BSCC is reachable from all of its states, and because $\langle s', q' \rangle \in B_{\lfufa}$ yields that any finite path of $\dtmc_{s'}$ can be extended such that its lifting to $\dtmc \otimes \lfufa$ closes a loop on $\langle s', q' \rangle$. 
	Lastly, let $x$ be an extension of $uv$ such that the lifting of $uvx$ to $\dtmc \otimes \lfufa$ ends in $\langle s, q \rangle$, which is guaranteed to exist because $\langle s', q' \rangle$ and $\langle s, q \rangle$ are in the same BSCC $B_\lfufa$ of $\dtmc \otimes \lfufa$. 
	
	Then $uvx$ is a finite path whose lifting to $\dtmc \otimes \lfufa$ reaches $\langle s, q \rangle \in B_\lfufa$ and whose lifting to $\dtmc \otimes \automaton$ reaches the BSCC $B_\automaton$. Since $\probMeasure^{\dtmc_s}(\accLanguage(\automatonB(s,q))) = 1$ by assumption it follows from \Cref{lem:characterization-bscc-dfufa} that $B_\automaton$ is accepting. 
	
	However, $B_\automaton$ is already entered after $uv$ and the lifting of $uv$ to $\dtmc \otimes \lfufa$ ends in $\langle s', q' \rangle \in B_\lfufa$. Therefore, it follows again from \Cref{lem:characterization-bscc-dfufa} that $\probMeasure^{\dtmc_{s'}}(\accLanguage(\automatonB(s', q'))) = 1$. 
\end{proof}

\LemZeroOnePropertyFUFA*
\begin{proof}
	Let $B_{\lfufa}$ be good. We prove that all BSCCs reachable in $\dtmc \otimes \automaton$ after $u$ are accepting. Let $\langle s, a \rangle$ be the unique state reached on the lifting of $u$ to $\dtmc \otimes \automaton$, and let $B_{\automaton}$ be a BSCC of $\dtmc \otimes \automaton$ reachable from $\langle s, a \rangle$ on some $v$. Since $B_{\lfufa}$ is good there is a state $\langle s', q' \rangle$ in $B$ with $q' \in \mstates$ and $\probMeasure^{\dtmc_{s'}}(\accLanguage(\automatonB(s',q'))) = 1$, and because $B_{\lfufa}$ is already reached upon $u$ in $\dtmc \otimes \lfufa$ there must be a word $w$ such that the lifting of $y = uvw$ to $\dtmc \otimes \lfufa$ ends in $\langle s', q' \rangle$. But then $y$ is a finite word that simultaneously reaches $B_\automaton$ and $\langle s', q' \rangle \in B_{\lfufa}$, so \Cref{lem:characterization-bscc-dfufa} implies that $B_\automaton$ is accepting. It follows that the finite paths of $\dtmc$ with prefix $u$ almost surely end up in accepting BSCCs when lifted to $\dtmc \otimes \automaton$. Therefore, almost surely a path generated by $\dtmc$ with prefix $u$ satisfies $E$. 
	
	Now let $B_{\lfufa}$ be bad. If $B_{\lfufa}$ does not contain a state $\langle s, q \rangle$ with $q \in \mstates$ then the claim follows from \Cref{lem:no-A-state-prob-0}.
	Otherwise, we can argue similar to the previous case that there are finite paths $v$ and $w$ such that upon $uv$ a BSCC $B_\automaton$ is reached in $\dtmc \otimes \automaton$, and such that the lifting of $uvw$ to $\dtmc \otimes \lfufa$ ends in some $\langle s, q \rangle \in B_{\lfufa}$ with $q \in \mstates$. But then $\probMeasure^{\dtmc_s}(\accLanguage(\automatonB(s,q))) < 1$ because $B_{\lfufa}$ is bad, and it follows from \Cref{lem:characterization-bscc-dfufa} that $B_\automaton$ must be rejecting. Hence, all BSCCs in $\dtmc \otimes \automaton$ reachable after $u$ are rejecting. Therefore, almost surely a path generated by $\dtmc$ with prefix $u$ violates $E$.
\end{proof}
\section{Model Checking DTMCs Against UFA with Prefix-dependent Languages}\label{app:cex-ufa-mc}
\begin{figure}[t]
	\centering
	\begin{tikzpicture}[->,>=stealth',shorten >=1pt,auto, semithick]
		\tikzstyle{every state} = [text = black, scale = 0.9]
		
		\node[state] (q0) {$q_0$}; 
		\node[right of = q0, node distance = 2cm] (temp) {}; 
		\node[state, above of = temp, accepting, node distance = 0.8cm] (q1) {$q_1$};
		\node[state, right of = q1, node distance = 2cm] (q2) {$q_2$};
		\node[state, below of = temp, node distance = 0.8cm] (q3) {$q_3$};
		\node[state, right of = q3, accepting, node distance = 2cm] (q4) {$q_4$};
		\node[left of = q0, node distance = 1cm] (init) {}; 
		
		\path 
		(init) edge (q0)
		(q0) edge node {$s$} (q1)
		(q0) edge node {$s$}(q3)
		(q1) edge [bend left = 10] node {$s$}(q2)
		(q2) edge [bend left = 10] node {$s$}(q1)
		(q3) edge [bend left = 10] node {$s$}(q4)
		(q4) edge [bend left = 10] node {$s$}(q3)
		;

		\node[right of = temp, node distance = 4cm, state] (s) {$s$}; 
		\node[left of = s, node distance = 1cm] (initm) {}; 
		
		\path 
		(s) edge [loop above] node {$1$} (s)
		(initm) edge (s)
		;
		
		\node[state, right of = s, node distance = 2cm] (q0) {$s,q_0$}; 
		\node[right of = q0, node distance = 2cm] (temp) {}; 
		\node[state, above of = temp, accepting, node distance = 0.8cm] (q1) {$s,q_1$};
		\node[state, right of = q1, node distance = 2cm] (q2) {$s,q_2$};
		\node[state, below of = temp, node distance = 0.8cm] (q3) {$s,q_3$};
		\node[state, right of = q3, accepting, node distance = 2cm] (q4) {$s,q_4$};
		
		\path 
		(q0) edge (q1)
		(q0) edge (q3)
		(q1) edge [bend left = 10] (q2)
		(q2) edge [bend left = 10] (q1)
		(q3) edge [bend left = 10] (q4)
		(q4) edge [bend left = 10] (q3)
		;
		
		\node[above of = initm, node distance = 0.7cm] {$\dtmc$}; 
		\node[above of = init, node distance = 0.7cm] {$\automaton$}; 
		\node[above of = q0, node distance = 0.7cm, xshift = -0.3cm] {$G$}; 
	\end{tikzpicture}
	\caption{A UFA $\automaton$ for the language $s^+$ (left), a DTMC $\dtmc$ (middle) and the directed graph $G$ constructed according to the definitions given in \cite{MCMCAUBA} (right).}
	\label{fig:cex-ufa-mc}
\end{figure}
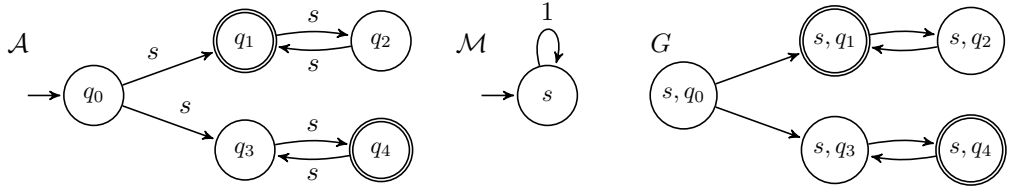

Let $\dtmc$ be a DTMC and  $\automaton$ a UFA. Moreover, let $\stateSpace$ be the set of states of $\dtmc$, let $Q$ be the set of vertices of $\automaton$, and let $\accLanguage(\automaton)$ be the language accepted by $\automaton$. W.l.o.g. we assume that the alphabet of $\automaton$ coincides with $\stateSpace$ and that $\labelFunction(s) = \{s\}$ for all $s \in \stateSpace$, which can always be achieved by existential renaming \cite{MCMCAUBA}. The model checking procedure proposed in \cite{MCMCAUBA} to compute the probability that $\dtmc$ generates a path with a prefix accepted by $\automaton$, i.e., to compute $\probMeasure^{\dtmc}(\accLanguage(\automaton) \cdot S^\omega)$, works as follows:  

\begin{enumerate}
	\item Construct a directed graph $G$ with vertices in $V = \stateSpace \times Q$ and where there is an edge from $(s,q)$ to $(s', q')$ iff $\prob(s,s') > 0$ and $q' \in \delta^\automaton(q, s')$. 
	\item Set $V^{acc} = \{(s,q) \mid q \in \afinal^{\automaton}\}$, $V^{dead}$ the set of all states in $V$ that cannot reach $V^{acc}$, and $V^? = V \setminus (V^{acc} \cup V^{dead})$. 
	\item Solve the following system of linear equations: 
	\begin{alignat*}{100}
		x_{s,q} &= 1 &&\quad \text{if} (s,q) \in V^{acc} \\
		x_{s,q} &= 0 &&\quad \text{if} (s,q) \in V^{dead} \\
		x_{s,q} &= \sum_{s' \in S} \sum_{q' \in \delta^\automaton(q, s')} \prob(s,s') \cdot x_{s',q'} &&\quad \text{if } (s,q) \in V^?
	\end{alignat*}
	\item Return $\probMeasure^{\dtmc}(\accLanguage(\automaton) \cdot \stateSpace^\omega) = x_{\initialState, \ainit}$, where $\initialState$ is the initial state of $\dtmc$ and $\ainit$ is the initial state of $\automaton$.
\end{enumerate}

Now consider the UFA $\automaton$ on the left of \Cref{fig:cex-ufa-mc} and the DTMC $\dtmc$ in the middle of the figure. The (reachable fragment of the) directed graph $G$ obtained from $\dtmc$ and $\automaton$ is shown on the right of the figure. States represented as double-circles in $G$ correspond to states in $V^{acc}$, and $V^{dead} = \emptyset$. Solving the system of linear equations from above yields
\begin{align*}
	x_{s, q_0}	&= \sum_{s' \in \stateSpace} \sum_{q'\in \delta^\automaton(q_0, s')} \prob(s,s') \cdot x_{s', q'} 
	= \prob(s,s) \cdot \underbrace{x_{s, q_1}}_{= 1 \text{ as} (s,q_1) \in V^{acc}} + \prob(s,s) \cdot x_{s, q_3}
	\\&= 1 + x_{s, q_3}
	= 1 +  \sum_{s' \in \stateSpace} \sum_{q' \in \delta^\automaton(q_3, s')} \prob(s,s') \cdot x_{s', q'} 
	= 1 + \underbrace{x_{s, q_4}}_{= 1 \text{ as } (s, q_4) \in V^{acc}} 
	= 2.
\end{align*}
Since $\probMeasure^\dtmc(\accLanguage(\automaton) \cdot \stateSpace^\omega)$ is a probability, it cannot coincide with $x_{s, q_0} = 2$ in this case. 

The reason for the erroneous computation of $x_{s, q_0}$ in the above example is that the system of linear equations counts paths of $\dtmc$ starting with $s$ multiple times. More precisely, in the directed graph $G$ as in \Cref{fig:cex-ufa-mc}, any infinite path that starts with $s$ is already ``accepted'' when moving to $\langle s, q_1 \rangle \in V^{acc}$. However, there are two outgoing edges of $q_0$ on $s$, yielding a second transition of $G$ from $\langle s, q_0 \rangle$ to $\langle s, q_3 \rangle$. To obtain $x_{s, q_0}$, the system of linear equations sums up the values of $x_{s, q_1}$ and $x_{s, q_3}$, which both equal $1$ since the DTMC $\dtmc_s$ almost surely generates a path accepted by the automata $\automaton_{q_1}$ and $\automaton_{q_3}$. Hence, the probability of $\dtmc$ to generate a path starting with $s$ is counted twice, yielding $x_{s, q_0} = x_{s, q_1} + x_{s, q_3} = 2$. 

This behavior cannot occur if the language of $\automaton$ is prefix-free or, more generally, if for all $x, xy \in \accLanguage(\automaton)$ the accepting run of $\automaton$ on $x$ is a prefix of the accepting run of $\automaton$ on $xy$. Given that (one of) these requirements are (is) satisfied, the algorithm of \cite{MCMCAUBA} correctly computes the desired probabilities.
We are not aware of any results that allow the computation of satisfaction probabilities of DTMCs against arbitrary UFA-specifications in polynomial time. 
\end{document}